\documentclass[opre,a4paper,12pt]{article}[geometry]

\usepackage{amsmath,amssymb,amsfonts,amsthm}
\usepackage{adjustbox}
\usepackage{afterpage}
\usepackage{anysize}
\usepackage{array}
\usepackage{bibentry,natbib}
\usepackage{blindtext}
\usepackage{bm,bbm}
\usepackage{booktabs}
\usepackage{breakurl}
\usepackage{calrsfs}
\usepackage{cancel}
\usepackage{caption}
\usepackage{changepage}
\usepackage{dsfont}
\usepackage{enumitem}
\usepackage{epsfig,epstopdf}
\usepackage{etoolbox}
\usepackage{fancyhdr} 
\usepackage{float}
\usepackage{setspace}
\usepackage[norule]{footmisc}
\usepackage[colorlinks=true,linkcolor=black,urlcolor=black,citecolor=black]{hyperref}
\usepackage{footnotebackref}
\usepackage{framed}
\usepackage[margin=1in]{geometry}
\usepackage{graphicx}
\usepackage{indentfirst}
\usepackage{lipsum}
\usepackage{longtable}
\usepackage{lscape}
\usepackage{mathrsfs,mathtools,mathptmx}
\usepackage{makecell}
\usepackage{multicol,multirow}
\usepackage{natbib}
\usepackage{pgf,pgfplots}
\usepackage{pdflscape}
\usepackage{pdfpages}
\usepackage{placeins}
\usepackage{rotate,rotating,capt-of,varwidth}
\usepackage{setspace}
\usepackage{sgame,sgamevar}
\usepackage{subfigure, subcaption}
\usepackage{tabularx,tabulary}
\usepackage[para]{threeparttable}
\makeatletter
\patchcmd{\TPT@measurement}
  {\xdef\TPT@hsize}
  {\ifdim\hsize>\textwidth\hsize\textwidth\fi\xdef\TPT@hsize}
  {}{\PackageError{Revision tables}{Unable to constrain table-note width}{}}
\makeatother
\usepackage{tikz}
\usepackage{times}
\usepackage{titlesec}
\usepackage{todonotes}
\usepackage{ulem}
\usepackage{url}
\usepackage{verbatim}
\usepackage{xcolor}
\usepackage{yhmath}

\usepackage[titletoc]{appendix}

\makeatletter
\renewcommand{\thanks}[1]{%
  \footnotemark
  \protected@xdef\@thanks{\@thanks
    \protect\footnotetext[\the\c@footnote]{%
      \protect\setstretch{\setspace@singlespace}#1}%
  }%
}
\makeatother

\graphicspath{{./Graphs/}}
\usepackage{natbib}
\bibpunct[, ]{(}{)}{,}{a}{}{,}%
\newtheorem{thm}{\protect\theoremname}

\theoremstyle{definition}
\newtheorem{defn}[thm]{\protect\definitionname}
\theoremstyle{plain}
\newtheorem*{prop*}{\protect\propositionname}
\newtheorem{prop}{Proposition}
\newtheorem{result}{Result}

\newtheorem{pred}{Prediction}

\newtheorem{hypo}{Hypothesis}

\newcolumntype{L}[1]{>{\raggedright\let\newline\\\arraybackslash\hspace{0pt}}m{#1}}
\newcolumntype{C}[1]{>{\centering\let\newline\\\arraybackslash\hspace{0pt}}m{#1}}
\newcolumntype{R}[1]{>{\raggedleft\let\newline\\\arraybackslash\hspace{0pt}}m{#1}}
\newcolumntype{d}[1]{>{\rightdots{#1}}r<{\endrightdots}}

\providecommand{\definitionname}{Definition}
\providecommand{\propositionname}{Proposition}
\providecommand{\theoremname}{Theorem}
\providecommand{\conjecturename}{Conjecture}
\providecommand{\theoremname}{Theorem}

\usepackage{calrsfs}
\DeclareMathAlphabet{\pazocal}{OMS}{zplm}{m}{n}

\title{\Large 

Biased Agents, Extreme Beliefs: \\ Motivated Reasoning Under Competing Models\thanks{First draft: August 31, 2026. I am indebted to Timothy Cason, David Gill, Matthew Kovach, and Colin Sullivan for invaluable guidance. I thank Cuimin Ba, James Bland, John Conlon, Tony Fan, Stanton Hudja, Andrew Olsen, Yaroslav Rosokha, Peter Schwardmann, Elchin Suleymanov, Stephanie Wang, Zitian Wang, Daniel Woods, Sevgi Yuksel, Junya Zhou, and audiences at various seminars and conferences for helpful comments. The study was reviewed by Purdue University's Institutional Review Board (IRB-2024-1786). Financial support from the Krannert Graduate Institute at Purdue University and Institute for Humane Studies is acknowledged.}

}

\author{\large{Zhongheng Qiao}\thanks{Purdue University; qiao50@purdue.edu} \\ \normalsize{Purdue University}}

\author{
     Zhongheng Qiao\thanks{Purdue University, 403 Mitch Daniels Blvd, West Lafayette, IN 47907, qiao50@purdue.edu}
}

\date{This Version: September 17, 2026}

\begin{document}

\vspace{-1cm}

\hypersetup{pageanchor=false}
\maketitle
\thispagestyle{empty}

\begin{center}
\vspace{-1cm}
\href{https://zhonghengqiao.com/between_stories.pdf}{\textcolor{blue}{Click here for the latest version}}
\end{center}  

\begin{spacing}{1}
\begin{center}
\begin{abstract}
\noindent People often face environments where multiple models compete to explain the same observations. This paper examines how people update beliefs in such settings and how preferences over payoff-relevant states shape model selection and belief updating. This paper first develops a framework where preference-driven bias distorts the perceived model, affecting Bayesian and best-fit updating differently. In a laboratory experiment, most participants are classified as Bayesian updaters, who average across models, while a substantial minority are classified as best-fit updaters, who select the model that best fits the observed signal. Within-participant comparisons between the symmetric payoff and asymmetric payoff conditions indicate that asymmetric payoffs shift reported beliefs toward the preferred state, particularly among participants classified as best-fit updaters. Relative to symmetric payoffs, asymmetric payoffs increase the reported belief of the preferred state by about 8 percentage points among best-fit updaters, while the estimated effect among Bayesian updaters is close to zero. These findings help us better understand model-based learning and have implications for domains such as political polarization and financial investment, where competing narratives and strong preferences often coexist. \\

\vspace{0.25cm}
    \vspace{0.25cm}
\end{abstract}
\end{center}
\end{spacing}



\clearpage
\setcounter{page}{1}
\hypersetup{pageanchor=true}
\newpage

\vspace{-0.5cm}

\section{Introduction}\label{sec:introduction}

Individuals routinely encounter complex learning problems in which the same information can be explained in different ways. Voters are exposed to competing stories about election candidates and economic policies. Investors receive conflicting evaluations of corporate performance. Jurors hear competing narratives about the same evidence. People can draw different conclusions from the same information \citep{hastorf1954they,lord1979biased,gaines2007same}, and their preferences or motivations may affect which interpretation they find convincing \citep{kunda1990case}. These examples illustrate two common features of the environments people face. First, multiple competing models can explain the same observed information and imply different inferences about the unobserved state of the world. Second, people may have intrinsic preferences over unobserved states of the world, such as a favored candidate being the legitimate election winner or a defendant being guilty.

These features raise a central question: how do individuals use competing models to interpret the same information? Under Bayesian updating, individuals update the weight assigned to each model and then average the model-specific posterior probabilities using those posterior model weights. An alternative is to select a single model according to a decision-relevant criterion, such as how well it fits the observed information \citep{schwartzstein2021using,aina2023tailored}. Motivated reasoning may work differently under these behavioral rules. Individuals who average across competing models may place more weight on a model that supports their preferred state. Individuals who select a single model may switch to another model, adopting a different explanation \citep{angrisani2024competing}. The same preferences therefore may have different effects on posterior beliefs for individuals following different behavioral rules.

This paper examines how motivated reasoning affects individuals' selection among models and belief updating in a controlled multi-model environment. Here, a model is a data-generating process that maps unobserved states of the world into distributions over observable signals.\footnote{Throughout the paper, ``models'' denote these data-generating processes. ``Narratives'' and ``stories'' refer to causal interpretations linking observed information to unobserved states.} Specifically, this paper focuses on three interrelated questions: (1) To what extent do individuals' preferences over states affect model selection and belief updating? (2) What behavioral rules do individuals follow to select models and update beliefs in such environments? (3) Do motivated-reasoning effects differ across individuals who follow different behavioral rules?

To answer these questions, I develop a theoretical framework and design a laboratory experiment. In the framework, agents want to learn an unobserved state from information that can be explained by competing models. They have priors over states and models, and each state is associated with a payoff. Preference-driven distortion operates through agents' mental representation of the information structure. States associated with higher payoffs receive greater weight in this mental representation. The distortion affects both how agents interpret the observed information under each model and how well each model appears to fit that information. I apply the same distortion under the two behavioral rules. For agents following the Bayesian behavioral rule (Bayesian agents), it affects model weights and the posteriors implied by each model. For agents following the Best-fit behavioral rule (Best-fit agents), it affects the posteriors implied by each model and may also reverse model selection. An agent may switch to a model under which the same information implies a higher probability for the preferred state. Such a switch can generate a large change in posterior beliefs. The framework characterizes the distortion required for these switches and compares the belief changes across behavioral rules. 

The experiment uses a \(2 \times 2 \times 2\) belief-updating task: two states, two signals, and two models. In the classic ball-and-urn setup, participants observe two carts, A and B; two urns in each cart, X and Y; and purple and orange balls in each urn. The computer selects X or Y according to the given state prior and independently selects either cart with equal probability. Participants observe the color of a ball drawn from the selected urn and report their belief that Urn X was selected. The labels X and Y represent the states, ball colors represent signals, and the two carts represent models. Each participant completes a symmetric payoff (\textit{SP}) block and an asymmetric payoff (\textit{AP}) block of belief updating tasks in random order. In the symmetric payoff block, participants receive the same state-dependent payment whether Urn X or Urn Y is selected. This block provides a benchmark for their updating behavior without a preference for either state. In the asymmetric payoff block, the payment is higher when Urn X is selected. Participants therefore have a reason to prefer Urn X, while the priors and information structures remain unchanged. Each block contains the same nine decision problems, with each problem repeated three times. The decision problems vary in state priors and model pairs. I use reported beliefs from the symmetric payoff block to distinguish the behavioral rules participants follow. I then compare reported beliefs across blocks, controlling for the decision problem and observed signal. This within-subject design allows me to examine how preferences affect belief updating for participants following different behavioral rules. The experiment also measures participants' probabilistic reasoning ability, overconfidence, and risk attitudes, and elicits their opinions about competing economic narratives. The theoretical framework and experimental design motivate three hypotheses. First, asymmetric payoffs shift reported beliefs toward the preferred state. Second, this effect is larger on average among Best-fit participants than among Bayesian participants. Third, the average effect among Best-fit participants increases as the objective fit gap, the difference in how well the two models fit the observed signal, narrows.

The laboratory results show that asymmetric payoffs shift beliefs toward the preferred state, with the largest effects among Best-fit participants. I first apply a finite mixture model to classify participants using only their reported beliefs in the symmetric payoff block. More than 60\% of participants are classified as Bayesian, and around 17\% are classified as Best-fit. The remainder are classified as Non-movers. Thus, most participants weight competing models, while a substantial minority rely on a single model that fits the signal best, which provides evidence on the Best-fit behavioral rule studied in recent theoretical and experimental work \citep{aina2023tailored,aina2025weighting}. Asymmetric payoffs increase the reported beliefs of the preferred state by 1.7 percentage points overall. The effect is 8.4 percentage points among Best-fit participants and 0.5 percentage points among Bayesian participants. Best-fit participants respond significantly more strongly to asymmetric payoffs than Bayesian participants. This difference remains robust under alternative sets of behavioral types. It also holds when I allow different types to have different levels of noise. In that specification, about one-third of participants are classified as Best-fit, and asymmetric payoffs still have a significantly larger effect among Best-fit participants than among Bayesian participants.

These results support Hypotheses~1 and~2. The estimates for Hypothesis~3 have the predicted sign but are not statistically significant. These findings suggest that the behavioral rules individuals follow matter for motivated reasoning. The same preferences may affect reported beliefs differently depending on how individuals use competing models to interpret information. Best-fit participants also report more extreme beliefs in the symmetric payoff block. This difference persists in the decision problem where both behavioral rules predict the same beliefs. Finally, an unincentivized survey suggests a connection between participants' updating behavior in the laboratory and their views about the world. Participants who hold more extreme views on economic narratives are more likely to be classified as Best-fit.

The remainder of this paper is organized as follows. Section~\ref{sec:literature} reviews the related literature and positions this paper's contributions. Section~\ref{sec:theory} presents the theoretical framework, characterizing model selection and belief updating under preference-driven distortion. Section~\ref{sec:experiment} presents the experimental design and procedures. Section~\ref{sec:results} presents the main results and further discussion. Section~\ref{sec:conclusion} concludes. 

\section{Related Literature} \label{sec:literature}

Learning often takes place when several models can explain the same observed information. In such settings, agents need to draw inferences within each model and apply a rule for using the models. This paper studies how preferences over payoff-relevant states affect both parts of this problem.\footnote{\citet{kunda1990case} and \citet{benabou2016mindful} provide broad discussions of motivated beliefs. \citet{battigalli2022belief} review belief-dependent motivations and psychological games.} The framework has three central features. First, the same information can have opposite implications for the preferred state under different models. Second, a preference-driven distortion changes model-specific posterior probabilities and how agents perceive each model. Third, the effect of this distortion depends on the behavioral rule agents follow to weight or select models.

This paper first contributes to the literature on motivated reasoning. The same information may imply different posterior probabilities under different models, and whether the information supports the preferred state depends on how agents weight or select the models. Preferences can therefore affect model-specific posterior probabilities, model weights, or model selection, which creates an additional channel for motivated reasoning. In most existing studies, information is unambiguously good or bad news, which leaves no room to change its valence. The closest empirical work on motivated reasoning falls into two broad groups. The first group studies motivated reasoning about self-relevant beliefs.\footnote{For example, this literature studies intelligence and perceived ability (e.g., \citealp{eil2011good, mobius2022managing}), task performance (e.g., \citealp{ertac2011does, coutts2019good, drobner2024motivated}), self-serving fairness (e.g., \citealp{babcock1995biased}), persuasion and self-persuasion (e.g., \citealp{schwardmann2019deception, schwardmann2022self}), and motivated or biased memory (e.g., \citealp{zimmermann2020dynamics, huffman2022persistent, conlon2026memory}).} \citet{coutts2024blame} show that aggregate feedback jointly determined by own and teammate performance leads participants to distort beliefs about both dimensions, which supports self-serving beliefs about their own ability. \citet{bolte2024motivated} find that participants judge identical ego-favorable reports as more likely to represent independent information than identical ego-unfavorable reports. They find no significant asymmetry in posterior beliefs about own test performance. Like my experiment, these designs allow participants to interpret the same information in different ways. However, whether the information is good or bad news remains fixed in these settings.


The second group studies motivated reasoning about uncertain future outcomes or other world states that are not self-relevant. \citet{mayraz2011wishful} shows that randomly assigned stakes shift beliefs toward payoff-favorable outcomes. The shift is larger when subjective uncertainty is greater. \citet{charness2017confirmation} also induce preferences over states through payoffs within a known signal structure. Participants with a preferred state update closer to Bayes overall. \citet{barron2021belief} varies state-contingent payoffs within a known information structure and finds approximately Bayesian updating on average, without a significant good-news asymmetry.\footnote{Evidence for asymmetric belief updating is mixed across settings, including self-relevant ones. Some experiments find no asymmetry, stronger responses to bad news, or effects that depend on the setting \citep{ertac2011does, coutts2019good, thaler2026good}. \citet{drobner2022motivated} finds optimistic updating when participants expect no immediate resolution and neutral updating under immediate resolution. \citet{engelmann2024anticipatory} find wishful thinking in the loss domain but not in the gain domain, and stronger when the evidence is more ambiguous.} Like these studies, I use state-dependent payoffs to create a preferred state and extend this setting to competing models, so preferences can also affect how agents weight or select models.\footnote{Related work studies politically motivated beliefs about news, events, or facts \citep{kimball2024happiness, thaler2024fake, barron2025motivated}. Another strand studies how motivated beliefs are exchanged or supplied in social settings \citep{oprea2022social, thaler2021supply} and how senders react to wishful thinkers \citep{filiz2026information}. Other work studies motivated reasoning in asset valuation or project continuation \citep{wang2026motivated, qiao2025sunk}.}

Recent work studies motivated reasoning when the source or accuracy of information is uncertain. \citet{thaler2024fake} identifies motivated reasoning through trust in information sources. After reporting a guess on politicized questions or on their own performance, subjects receive a message on whether the correct answer lies above or below that guess. The message may come from one of two sources, one always correct and one always incorrect. The experiment shows that subjects overtrust messages that favor their own political party or performance. \citet{avtar2026attribution} study how motivated reasoning affects beliefs about the source of ego-relevant information. Subjects report the probability that a message came from the more informative of two sources. They find that motivation shifts source beliefs in an ego-favorable direction, especially when knowing the source matters more for inferring the state. \citet{hu2025motivated} instead study ego-relevant beliefs about the state when feedback accuracy is known, compound-uncertain with stated probabilities, or ambiguous. They find that asymmetric updating is strongest with known accuracy and absent under ambiguity. I use state-dependent payoffs in an abstract learning environment to isolate preferences from self-image. In this learning environment, competing models provide conflicting inferences about the uncertain states, so the valence of information depends on how subjects use the models. In addition, even when the model is held fixed, motivated reasoning may affect the model-specific posteriors. This paper therefore examines how motivated reasoning affects belief updating through both channels, and how its effect differs across the behavioral rules subjects follow.

This paper also develops a theoretical framework linking motivated reasoning to model uncertainty. \citet{kovach2020twisting} axiomatizes wishful thinking by allowing an agent's endowed act to shift beliefs toward states in which that act yields higher utility. Building on this idea, I apply a preference-driven distortion separately within each competing model. The distortion raises the perceived likelihood of states that yield higher utility.\footnote{\citet{ba2024over} develop and test a two-stage model. Agents first form a mental representation of the information structure by overweighting states most representative of the observed signal. They then process this representation imprecisely. The present paper also uses a distorted mental representation of the information structure that overweights the preferred states.} Within each model, it changes model-specific posterior beliefs and subjective fit. \citet{chambers2023coherent} characterize distortions that produce the same posterior whether they are applied before or after conditioning on a signal. The distortion in the present framework belongs to this class. \citet{liu2026confirmation} studies two sources of bias within a known information structure. Beliefs shift toward states that initially appear more likely and toward states that yield higher utility. I focus on the second effect and study how it changes inference across competing models.\footnote{Another strand of the literature on non-Bayesian updating studies related distortions that depend on the agent's current beliefs over states (e.g., \citealp{rabin1999first, fryer2019updating}; see \citealp{benjamin2019errors} for a systematic review). In this paper, the distortion depends on the payoff from each state.} A complementary literature gives beliefs or expectations a direct role in utility \citep{brunnermeier2005optimal, battigalli2009dynamic, battigalli2019incorporating}. \citet{brunnermeier2005optimal} let subjective expectations of future utility flows enter current felicity. Optimal beliefs balance anticipatory gains from optimism against losses caused by distorted decisions.\footnote{\citet{battigalli2009dynamic} allow utility to depend on updated beliefs and beliefs about other players' beliefs. \citet{battigalli2019incorporating} embed belief-dependent motivations into dynamic games.} In the present framework, utility depends only on the realized state. State-dependent payoffs determine the preference-driven distortion of model fits and posterior beliefs.

This paper also contributes to the literature on model selection. It develops a common framework in which preference-driven distortion operates under the Bayesian and best-fit behavioral rules. The framework predicts that the distortion affects agents following these rules differently. The closest theoretical work studies fit-based model selection.\footnote{The best-fit, or maximum likelihood, rule is also used in other settings. \citet{gilboa1993updating} axiomatize the maximum likelihood update rule for multiple priors, in which the agent keeps the priors that make the observed event most likely. Later work extends the rule (e.g., \citealp{epstein2007learning, kovach2024ambiguity, suleymanov2025robust}). In this paper, the priors over states and models are explicitly given to the participants.} \citet{schwartzstein2021using} study persuasion when a receiver selects the model that makes the observed data most likely, given the receiver's prior over states. \citet{aina2023tailored} builds on this fit-based rule in a setting where a persuader supplies multiple models before the signal is realized. The author also notes Bayesian model averaging as an alternative, under which posterior model probabilities weight model-specific posteriors.\footnote{Related theory also focuses on criteria other than fit \citep{yang2023criterion}, strategic supply of models \citep{jain2023informing}, and model sharing \citep{schwartzstein2022shared}. Other theory studies belief updating on different grounds, such as minimal movement from priors \citep{dominiak2023inertial, dominiak2025inertial}.} I adopt the best-fit behavioral rule from this work alongside the Bayesian rule and connect both rules to motivated reasoning.\footnote{A separate literature studies learning when the true model may be absent from the learner's set \citep{bohren2021learning, montiel2022competing, fudenberg2023misspecifications, ba2026robust}.}

Among experiments with competing models, \citet{aina2025weighting} provide the closest benchmark in a setting without induced preferences over states: subjects update beliefs when the signal may come from two conflicting models. Subjects observe both models' predictions, which separates model weighting from errors in probabilistic reasoning. They find that the most frequent updating rule gives full weight to the model that best fits the observed data. I study a similar updating problem and ask how motivated reasoning affects participants who follow different behavioral rules. The symmetric payoff (\textit{SP}) block in my paper provides a comparable environment, except that participants do not observe these model-specific predictions. Based on beliefs from the \textit{SP} block, a nontrivial proportion of participants are classified as Best-fit updaters. However, the most common behavioral rule in the \textit{SP} block is the Bayesian rule, which differs from their finding. The asymmetric payoff (\textit{AP}) block adds the motivated reasoning component, which is the margin this paper studies.

Related experiments explore how people use other forms of models. \citet{musolff2025model} study uncertainty between two models with varying complexity. They find that greater complexity increases neglect of model uncertainty in participants' actions, while participants' elicited beliefs still reflect model uncertainty. \citet{ambuehl2025choice} focus on how people choose between conflicting causal models. They find that subjects tend to focus on worst-case outcomes and choose cautiously when they cannot identify the correct model. \citet{farina2026communicating} focus on strategic choice of data-generating processes and find that receivers often neglect how an informed sender's private type shapes the sender's unobserved choice of data-generating process. Other experiments examine different forms of model uncertainty. \citet{epstein2024hard} focus on settings where signals have multiple interpretations. They compare belief updating from noisy signals with updating from ambiguous signals, and find that ambiguous signals increase deviations from Bayes' rule. \citet{liang2025learning} studies compound and ambiguous uncertainty about source accuracy. Participants respond less to reports from uncertain sources than to reports with known accuracy, especially after good news.\footnote{A broader literature studies belief updating under ambiguity and compound uncertainty. For example: \citet{ngangoue2021learning, shishkin2023ambiguous} study belief updating under ambiguity, \citet{moreno2016learning,bland2021learning} compare belief updating under compound risk with ambiguity, and \citet{kovach2026learning} study belief updating with qualitative AI recommendations.}  The present paper introduces an additional component, \textit{motivated reasoning}, that shapes model selection and belief updating. It documents that motivated reasoning has substantially different impacts on people who follow different behavioral rules, and the effects are especially large for best-fit agents.

More broadly, the paper also connects to research on mental models and narratives.\footnote{Broad discussions and reviews include \citet{shiller2017narrative, roos2024narratives, barron2024narrativereview}. Related experiments study responses to supplied causal models \citep{charles2025causal,ambuehl2025choice, he2026consensus}, mental-model formation from data \citep{esponda2024mental, frechette2024extracting, kendall2024complexity, fan2024forming}, and narrative construction and persuasion \citep{barron2024narrative, andre2026narratives, liu2026counteracting}. Related field evidence documents the adoption of competing narratives \citep{angrisani2024competing}.} \citet{eliaz2020model} represent a narrative as a causal model that maps actions into consequences. Agents adopt the narrative that maximizes anticipatory utility among those that correctly predict the observed consequences. \citet{eliaz2026news} study news media that supply information and narratives to maximize consumers' anticipatory utility. \citet{fan2026narratives} show that constructing or receiving a well-fitting narrative makes beliefs more extreme and volatile than learning directly from the same data. In the present paper, Bayesian model weights respond continuously to the distortion. Best-fit selection can switch the chosen model and move posterior beliefs discretely. The distinction between updating rules identifies an individual-level channel for discrete changes in posterior beliefs. Its implications for extremeness depend on the model-specific posteriors and the information structure.\footnote{Related work also links political values or rationalization to belief divergence \citep{barron2025motivated,yaouanq2025rationalizations}.}

\section{Theoretical Framework}\label{sec:theory}

This section proposes a framework that links motivated reasoning to belief updating with competing models. A preference-driven distortion changes the perceived information structure and therefore affects posterior beliefs. The impact of this distortion depends on the behavioral rule people follow. Best-fit agents choose the model that best fits the observed information, whereas Bayesian agents update model weights across competing models.

\subsection{Environment}\label{sec:environment}

An agent wants to learn an unknown state \(\omega\) from the finite set \(\Omega \equiv \{\omega_1,\ldots,\omega_N\}\), where \(N>1\). The state prior is \(\mu_0\in\Delta(\Omega)\), and \(\mu_0(\omega)\) denotes the prior probability of state \(\omega\).

The finite signal space is \(S\equiv\{s_1,\ldots,s_K\}\), where \(K>1\). A model \(m\) maps unobserved states of the world into distributions over observable signals, which is often called an information structure in the belief-updating literature. \(\pi_m(s\mid\omega)\) denotes the probability of signal \(s\) in state \(\omega\) under model \(m\). For every \(m\) and \(\omega\), \(\sum_{s\in S}\pi_m(s\mid\omega)=1\). Let \(\mathbb{M}\) denote the set of all such maps from \(\Omega\) to \(\Delta(S)\).

The agent is exposed to a finite set of competing models \(M\subseteq\mathbb{M}\) and has a model prior \(\rho_0\in\Delta(M)\). The state and model are independent under the prior: \(\Pr(\omega,m)=\mu_0(\omega)\rho_0(m)\).\footnote{This assumption makes the state prior common across models, as in \citet{aina2023tailored}.}

Conditional on model \(m\) and observed signal \(s\), Bayes' rule gives the model-specific posterior \(\mu_m(\omega\mid s)\). The objective fit \(P_m(s)\) is the likelihood of observing signal \(s\) under model \(m\), given the state prior:
\begin{subequations}\label{eq:model-posterior-fit}
\begin{align}
    \mu_m(\omega\mid s)
    &=
    \frac{\mu_0(\omega)\pi_m(s\mid\omega)}
    {\sum_{\omega'\in\Omega}\mu_0(\omega')\pi_m(s\mid\omega')},
    \label{eq:model-specific-posterior}\\
    P_m(s)
    &=
    \sum_{\omega\in\Omega}\mu_0(\omega)\pi_m(s\mid\omega).
    \label{eq:objective-fit}
\end{align}
\end{subequations}

This section focuses on two behavioral rules for updating beliefs with competing models: the best-fit rule and the Bayesian rule. Best-fit agents select the model with the highest objective fit and update only under that model \citep{schwartzstein2021using,aina2023tailored}. \(m_s^*\) denotes the model with the highest objective fit to signal \(s\), and \(\mu_F(\omega\mid s)\) denotes the posterior belief of a best-fit agent:
\begin{equation}\label{eq:best-fit-rule}
    m_s^*\in\arg\max_{m\in M}P_m(s),
    \qquad
    \mu_F(\omega\mid s)=\mu_{m_s^*}(\omega\mid s).
\end{equation}

Bayesian agents retain uncertainty over the competing models. After observing \(s\), they update each model's weight according to its prior probability and objective fit. \(\rho(m\mid s)\) denotes the weight assigned to model \(m\) after observing signal \(s\):
\begin{equation}\label{eq:bayesian-weight}
    \rho(m\mid s)=
    \frac{\rho_0(m)P_m(s)}{\sum_{m'\in M}\rho_0(m')P_{m'}(s)}.
\end{equation}
Bayesian agents then average the model-specific posterior probabilities using these updated model weights:
\begin{equation}\label{eq:bayesian-belief}
    \mu_B(\omega\mid s)=\sum_{m\in M}\rho(m\mid s)\mu_m(\omega\mid s).
\end{equation}
Here, \(\mu_B(\omega\mid s)\) denotes the posterior belief of a Bayesian agent.

I impose the following regularity conditions. The state and model priors have full support. Every conditional signal distribution also has full support:
\[
    \mu_0(\omega)>0,\qquad
    \rho_0(m)>0,\qquad
    \pi_m(s\mid\omega)>0
\]
for every \(\omega\in\Omega\), \(m\in M\), and \(s\in S\). These conditions ensure that all posterior probabilities are strictly positive and that the ratios used below are well-defined.

Models in \(M\) are distinct, so any two models differ for at least one state-signal pair. A fixed tie-breaking rule resolves all fit ties and applies uniformly across signals and distortions.

\subsection{Preferences over States}\label{sec:distortion}

To introduce preference-driven distortions, suppose each state is associated with a payoff to the agent. Let \(U(\omega)\) denote the agent's utility from the payoff associated with state \(\omega\). Without loss of generality,
I can order the states so that
\(
    U(\omega_1)\geq U(\omega_2)\geq\cdots\geq U(\omega_N).
\)

Motivated reasoning may shift beliefs toward states that give the agent higher utility. Following \citet{kovach2020twisting}, I use a strictly increasing function \(f:\mathbb{R}\rightarrow\mathbb{R}_{++}\) to model preference-driven distortion. The term \(f(U(\omega))\) determines the relative distortion associated with each state. Because \(f\) is strictly increasing, states with higher utility receive greater relative weight. For a given model \(m\) and observed signal \(s\), the distorted model-specific posterior probability of state \(\omega\), \(\hat{\mu}_m(\omega\mid s)\), is given by:
\begin{equation}\label{eq:subjective-posterior}
    \hat{\mu}_m(\omega\mid s)
    =\frac{\mu_0(\omega)\pi_m(s\mid\omega)f(U(\omega))}
    {\sum_{\omega'\in\Omega}\mu_0(\omega')\pi_m(s\mid\omega')f(U(\omega'))}
    =\frac{\mu_m(\omega\mid s)f(U(\omega))}
    {\sum_{\omega'\in\Omega}\mu_m(\omega'\mid s)f(U(\omega'))}.
\end{equation}

\subsubsection{Interpretation of the Distortion}

I interpret the distortion as modifying the agent's mental representation of the information structure. For each model \(m\), the distorted mental representation is defined by:
\begin{equation}\label{eq:distorted-information-structure}
    \hat\pi_m(s\mid\omega)=\pi_m(s\mid\omega)f(U(\omega)).
\end{equation}
Substituting \(\hat\pi_m\) for \(\pi_m\) in Bayes' rule yields the distorted model-specific posterior in Equation~\eqref{eq:subjective-posterior}. This posterior is well-defined, although \(\hat\pi_m\) need not be a conditional probability distribution.\footnote{\citet{ba2024over} use the term ``pseudo-information structure'' for a similar distorted mental representation. Their model overweights signal likelihoods in states that are more representative of the observed signal.} \(\hat\pi_m\) is a conditional probability distribution only when \(f(U(\omega))=1\) for every state.


Preference-driven distortion also changes the agent's perception of how well each model fits the observed signal. For model \(m\) and signal \(s\), subjective fit, \(\hat P_m(s)\), aggregates the distorted mental representation across states using the state prior:
\begin{equation}\label{eq:subj-fit}
\begin{aligned}
    \hat P_m(s)
    &=
    \sum_{\omega\in\Omega}
        \mu_0(\omega)\hat\pi_m(s\mid\omega)\\
    &=
    \sum_{\omega\in\Omega}
        \mu_0(\omega)\pi_m(s\mid\omega)f(U(\omega))\\
    &=
    P_m(s)
    \sum_{\omega\in\Omega}
        \mu_m(\omega\mid s)f(U(\omega)).
\end{aligned}
\end{equation}
Equation~\eqref{eq:subj-fit} shows that subjective fit is the product of objective fit and the posterior expectation of \(f(U)\) under model \(m\). Therefore, preference-driven distortion can favor models under which the observed signal indicates a higher probability on states the agent prefers. Subjective fits determine which model best-fit agents select. Together with model priors, they determine model weights Bayesian agents assign across models.\footnote{Alternatively, the preference-driven distortion can be incorporated into the prior over states while leaving the information structure unchanged. Replacing the prior term \(\mu_0(\omega)\) in Bayes' rule with \(\mu_0(\omega)f(U(\omega))\) yields the same distorted posterior belief as Equation~\eqref{eq:subjective-posterior}. \citet{liu2024essays} discusses the equivalence between placing a state-dependent multiplicative distortion in the prior and in the signal likelihoods. In my setting, the two interpretations generate the same model selections for Best-fit agents and the same distorted model weights for Bayesian agents.}



\subsubsection{Belief Updating under Preference-Driven Distortion}

Preference-driven distortion can affect beliefs in two ways. It changes posterior probabilities within each model and changes how the agent chooses or weights competing models.

Within any model \(m\), Equation~\eqref{eq:subjective-posterior} implies the following relationship between objective and distorted posterior odds. For any two states \(\omega,\omega'\in\Omega\),
\begin{equation}\label{eq:rel-subj-belief}
    \frac{\hat{\mu}_m(\omega\mid s)}
    {\hat{\mu}_m(\omega'\mid s)}
    =
    \frac{\mu_m(\omega\mid s)}
    {\mu_m(\omega'\mid s)}
    \frac{f(U(\omega))}
    {f(U(\omega'))}.
\end{equation}
If \(U(\omega)>U(\omega')\), then
\(f(U(\omega))/f(U(\omega'))>1\). The distorted posterior odds of \(\omega\) relative to \(\omega'\) therefore exceed the posterior odds without distortion. If \(f(U(\omega))\) is constant across states, preference-driven distortion leaves every model-specific posterior unchanged.

Subjective fit, defined in Equation~\eqref{eq:subj-fit}, enters the two behavioral rules as follows. A best-fit agent selects the model with the highest subjective fit and updates her posterior belief under that model. \(\hat m_s^*\) denotes the model with the highest subjective fit to signal \(s\), and \(\hat\mu_F(\omega\mid s)\) denotes the distorted posterior belief of a best-fit agent:
\begin{equation}\label{eq:distorted-best-fit}
    \hat m_s^*
    \in
    \arg\max_{m\in M}\hat P_m(s),
    \qquad
    \hat\mu_F(\omega\mid s)
    =
    \hat\mu_{\hat m_s^*}(\omega\mid s).
\end{equation}
A Bayesian agent uses subjective fits and model priors to update model weights. The agent then averages the distorted model-specific posteriors using these weights. \(\hat\rho(m\mid s)\) denotes the distorted weight assigned to model \(m\) after observing signal \(s\), and \(\hat\mu_B(\omega\mid s)\) denotes the distorted posterior belief of a Bayesian agent:
\begin{equation}\label{eq:distorted-weight}
    \hat\rho(m\mid s)
    =
    \frac{\rho_0(m)\hat P_m(s)}
    {\sum_{m'\in M}\rho_0(m')\hat P_{m'}(s)},
    \qquad
    \hat\mu_B(\omega\mid s)
    =
    \sum_{m\in M}
    \hat\rho(m\mid s)\hat\mu_m(\omega\mid s).
\end{equation}

\subsection{Competing Models}

This section first examines how preference-driven distortion affects model selection and then examines its effects on posterior beliefs. I begin with a general result that allows any finite number of states, signals, and models. After that, the full-switch result requires two signals and two models, while the sensitivity result focuses on the \(2\times2\times2\) environment that matches my experiment. Finally, I derive implications for posterior beliefs.

\subsubsection{Definitions}

I begin by defining a condition that compares how well two models fit the same signal.
\begin{defn}[Relative Representativeness of Two Models]
Fix two distinct models \(m,m'\in M\) and a signal \(s\in S\). Model \(m\) is more representative than \(m'\) for \(s\) if
\begin{enumerate}
    \item \(\pi_m(s\mid\omega)\geq\pi_{m'}(s\mid\omega)\) for every \(\omega\in\Omega\); and
    \item \(\pi_m(s\mid\omega)>\pi_{m'}(s\mid\omega)\) for some \(\omega\in\Omega\).
\end{enumerate}
\end{defn}

This condition means that, in every state, signal \(s\) is at least as likely under model \(m\) as under model \(m'\). As a result, \(m\) always has higher objective fit and subjective fit than \(m'\). Proposition~\ref{prop:fit-representative} formalizes this implication.

The following definition classifies how preference-driven distortion changes the model selected by a best-fit agent.

\begin{defn}[Model Switch]
For best-fit agents, preference-driven distortion produces no switch if
\(\hat m_s^*=m_s^*\) for every \(s\in S\). It produces a full switch if
\(\hat m_s^*\neq m_s^*\) for every \(s\in S\). It produces a partial switch if the selected model changes for some signals and remains unchanged for others.
\end{defn}

\subsubsection{Predictions}

I first link relative representativeness to subjective fit in a setting with any finite number of states, signals, and models.

\begin{prop}[Model Fit]\label{prop:fit-representative}
Suppose \(m\) is more representative than \(m'\) for signal \(s\). Then the following statements hold:
\begin{enumerate}
    \item \(\hat P_m(s)>\hat P_{m'}(s)\);
    \item
    \(
        \frac{\hat\rho(m\mid s)}
        {\hat\rho(m'\mid s)}
        >
        \frac{\rho_0(m)}
        {\rho_0(m')};
    \)
    \item \(m'\notin\arg\max_{m''\in M}\hat P_{m''}(s)\). If \(M=\{m,m'\}\), \(\hat m_s^*=m\).
\end{enumerate}
\end{prop}

\begin{proof}
See Appendix~\ref{app:proof-fit-representative}.
\end{proof}

For Bayesian agents, part 2 implies that the posterior model odds of \(m\) relative to \(m'\) under preference-driven distortion are greater than the prior model odds. For Best-fit agents, part 3 implies that \(m'\) cannot be selected while \(m\) is available. When these are the only two models, the agent selects \(m\).

I next consider the case with two signals, allowing any finite number of states and models. In this case, differences in objective and subjective fit between any two models \(m,m'\in M\) reverse across signals:
\begin{equation}\label{eq:fit-reversal}
    P_m(s_2)-P_{m'}(s_2)
    =
    -\big[P_m(s_1)-P_{m'}(s_1)\big],
    \qquad
    \hat P_m(s_2)-\hat P_{m'}(s_2)
    =
    -\big[\hat P_m(s_1)-\hat P_{m'}(s_1)\big].
\end{equation}
Appendix~\ref{app:fit-identities} formally derives these results, which lead to the following prediction for two models and any finite state space.

\begin{pred}[Full Switch]\label{pred:full-switch}
    Suppose \(S=\{s_1,s_2\}\), \(M=\{m,m'\}\), and the models have distinct model fits. For a best-fit agent, preference-driven distortion produces either a full switch or no switch.
\end{pred}

\begin{proof}
See Appendix~\ref{app:proof-full-switch}.
\end{proof}

If the two models have equal model fit, the selected model is determined by the tie-breaking rule, and a partial switch may occur.

I next characterize the distortion required for a switch in the experimental \(2\times2\times2\) environment. Let
\(
    \Omega=\{\omega_1,\omega_2\},
    \,
    S=\{s_1,s_2\},
    \,
    M=\{m,m'\}.
\)
Suppose \(U(\omega_1)>U(\omega_2)\), so that \(\omega_1\) is the preferred state.

For signal \(s_1\), suppose neither model is more representative than the other. Suppose further that
\[
    P_m(s_1)>P_{m'}(s_1)
    \qquad\text{and}\qquad
    \mu_m(\omega_1\mid s_1)
    <
    \mu_{m'}(\omega_1\mid s_1).
\]
Thus, model \(m\) has higher objective fit for \(s_1\), while model \(m'\) yields a higher model-specific posterior probability on the preferred state.

Define
\begin{equation}\label{eq:switch-gaps}
\begin{aligned}
     A&=
    \mu_0(\omega_1)
    \big[
        \pi_{m'}(s_1\mid\omega_1)
        -
        \pi_m(s_1\mid\omega_1)
    \big],
    \\
    B&=
    \mu_0(\omega_2)
    \big[
        \pi_m(s_1\mid\omega_2)
        -
        \pi_{m'}(s_1\mid\omega_2)
    \big].
\end{aligned}
\end{equation}

The term \(A\) captures how much more likely signal \(s_1\) is under \(m'\) than under \(m\) when \(\omega_1\) is true, accounting for the prior probability \(\mu_0(\omega_1)\). Similarly, \(B\) captures how much more likely \(s_1\) is under \(m\) than under \(m'\) when \(\omega_2\) is true, accounting for \(\mu_0(\omega_2)\). Because
\(
    P_m(s_1)-P_{m'}(s_1)=B-A,
\)
we have \(B>A>0\).

Define
\[
    R=\frac{A}{B}\in(0,1),
    \qquad
    t=\frac{f(U(\omega_1))}{f(U(\omega_2))}>1.
\]
The ratio \(R\) compares the two differences in the objective fit calculation. Because \(R<1\), a value closer to one means that \(A\) and \(B\) more nearly offset each other. The ratio \(t\) compares the preference-driven distortion across the two states. A larger \(t\) places greater relative weight on the preferred state.

\begin{prop}[Sensitivity of Model Switch]\label{prop:sensitivity}
Given \(B>A>0\), preference-driven distortion produces a full switch if and only if
\[
    t>\frac{1}{R}.
\]
\end{prop}

\begin{proof}
See Appendix~\ref{app:proof-sensitivity}.
\end{proof}

Preference-driven distortion produces a full switch when $tA>B$. I call $1/R=B/A$ the switching threshold. When \(R\) is close to one, \(A\) nearly offsets \(B\), so the models have similar objective fits. A relatively small distortion toward \(\omega_1\) can then make \(m'\) have higher subjective fit than \(m\). When \(R\) is smaller, \(A\) offsets less of \(B\), which implies that a larger distortion is required to reverse the order of subjective fit.

I next examine how preference-driven distortion affects posterior beliefs. I return to the finite state, signal, and model spaces defined above and fix a signal \(s\in S\). For Bayesian agents, preference-driven distortion changes both model-specific posterior probabilities and model weights. Relative model weights shift toward models under which the observed signal supports states the agent prefers. Appendix~\ref{app:fit-identities} provides the pairwise relationship between posterior model odds before and after preference-driven distortion.

For every state \(\omega\in\Omega\), the distorted Bayesian posterior belief satisfies
\begin{equation}\label{eq:bayesian-distortion}
    \hat\mu_B(\omega\mid s)
    =
    \frac{\mu_B(\omega\mid s)f(U(\omega))}
    {\sum_{\omega'\in\Omega}
    \mu_B(\omega'\mid s)f(U(\omega'))}.
\end{equation}
Equation~\eqref{eq:bayesian-distortion} combines the changes in model-specific posteriors and model weights into a direct distortion of the Bayesian agent's posterior.

Consider the experimental \(2\times2\times2\) environment. Let
\(
    \Omega=\{\omega_1,\omega_2\},
    \,
    S=\{s_1,s_2\},
    \,
    M=\{m,m'\}.
\)
Suppose \(U(\omega_1)>U(\omega_2)\), so that \(\omega_1\) is the preferred state. 

\begin{prop}\label{prop:belief-direction}
For every signal \(s\in S\),
\[
    \hat\mu_B(\omega_1\mid s)
    >
    \mu_B(\omega_1\mid s),
    \qquad
    \hat\mu_F(\omega_1\mid s)
    >
    \mu_F(\omega_1\mid s).
\]
Moreover, the model selected under preference-driven distortion satisfies
\[
    \mu_{\hat m_s^*}(\omega_1\mid s)
    \geq
    \mu_{m_s^*}(\omega_1\mid s).
\]
\end{prop}

\begin{proof}
See Appendix~\ref{app:proof-belief-direction}.
\end{proof}

For Bayesian agents, Equation~\eqref{eq:bayesian-distortion} shows that preference-driven distortion shifts their posterior belief toward the preferred state. For best-fit agents, the distortion raises the posterior probability on the preferred state within every model, regardless of model switch. If a switch occurs, the new model yields at least as high a model-specific posterior probability on the preferred state as the original model. The within-model change and the change in model selection therefore move the distorted best-fit posterior belief in the same direction. A switch produces an additional discrete increase in the posterior belief.




I next examine how posterior beliefs change as preference-driven distortion toward the preferred state becomes stronger. Let
\(
    \Omega=\{\omega_1,\omega_2\}
    \, \  \text{and}\, \ 
    U(\omega_1)>U(\omega_2).
\)
Write
\[
    t=
    \frac{f(U(\omega_1))}
    {f(U(\omega_2))} > 1.
\]
A larger \(t\) places greater relative weight on the preferred state \(\omega_1\).

Let \(\hat m_s^*(t)\) denote the model selected by a best-fit agent after signal \(s\) when the relative distortion is \(t\). The term
\(
    \mu_{\hat m_s^*(t)}(\omega_1\mid s)
\)
denotes the model-specific posterior probability on \(\omega_1\) under that selected model when there is no distortion. The term \(\hat\mu_F(\omega_1\mid s)\) denotes the distorted posterior belief for best-fit agents.

\begin{prop}\label{prop:stronger-distortion}
For every signal \(s\in S\):
\begin{enumerate}
    \item The distorted Bayesian posterior belief
    \(\hat\mu_B(\omega_1\mid s)\) is continuous and strictly increasing in \(t\) on \((1,\infty)\).

    \item The model-specific posterior probability on \(\omega_1\) under the selected model,
    \(\mu_{\hat m_s^*(t)}(\omega_1\mid s)\), is non-decreasing in \(t\) on \((1,\infty)\). It remains constant between switching thresholds.

    \item The distorted best-fit posterior belief
    \(\hat\mu_F(\omega_1\mid s)\) is strictly increasing in \(t\) on \((1,\infty)\). It is continuous except at finitely many switching thresholds. A switch to a model with a strictly higher model-specific posterior probability on the preferred state produces an upward jump.
\end{enumerate}
\end{prop}

\begin{proof}
See Appendix~\ref{app:proof-stronger-distortion}.
\end{proof}

For Bayesian agents, Equation~\eqref{eq:bayesian-distortion} shows that a higher \(t\) continuously raises the distorted Bayesian posterior belief on the preferred state. For Best-fit agents, a higher \(t\) makes the likelihood of \(s\) in \(\omega_1\) more influential in subjective fit. The selected model therefore remains unchanged or switches to a model with an equal or higher model-specific posterior probability on \(\omega_1\). This explains why \(\mu_{\hat m_s^*(t)}(\omega_1\mid s)\) is non-decreasing and may remain constant over an interval.

Even when the selected model remains unchanged, a higher \(t\) raises the distorted model-specific posterior probability on \(\omega_1\) within that model. The distorted best-fit posterior belief is therefore strictly increasing. A switch produces an upward jump when the new model has a strictly higher model-specific posterior probability on the preferred state.

Thus, both \(\hat\mu_B(\omega_1\mid s)\) and \(\hat\mu_F(\omega_1\mid s)\) strictly increase with \(t\). The distorted Bayesian posterior belief changes continuously. The distorted best-fit posterior belief is continuous between switching thresholds and may jump upward when the selected model changes.

\section{Experimental Design} \label{sec:experiment}

To answer the research questions, this paper implements a lab experiment based on the classical ball-and-urn setup to simulate a multi-model learning environment. As formally defined in Section~\ref{sec:environment}, a \emph{model} is a data-generating process that maps unobserved states of the world into distributions over observable signals, which is often called an information structure in the belief-updating literature. Participants are explicitly informed about the priors over states and over models. This feature provides a clear Bayesian benchmark for belief updating and helps distinguish heterogeneous behavioral rules for model selection and belief updating in this environment. The experiment employs a within-subject design with respect to participants' preferences over states: Each participant completes one block of belief-updating tasks with equal state-dependent payoffs and another block with unequal state-dependent payoffs. This design allows the experiment to identify whether, and to what extent, preferences over unobserved states affect model selection and belief updating.

The experiment consists of five parts. In Part 1, participants complete two blocks of belief-updating tasks in random order: the \textit{Symmetric Payoff Belief Updating Task (SP)} and the \textit{Asymmetric Payoff Belief Updating Task (AP)}.\footnote{The order of the \textit{SP} and \textit{AP} blocks was randomized at the individual level. The instructions inform participants that Part 1 consists of two blocks of 27 scenarios each. The details for Block 2 are not disclosed before participants complete Block 1.} Each block contains 27 belief-updating scenarios. Earnings are summed across all scenarios. In Part 2, the experiment elicits participants' risk attitudes using the bomb risk task \citep{crosetto2013bomb}. In Part 3, participants complete eight cognitive tasks related to information processing and statistical reasoning, presented in random order. Seven of these tasks are adapted from \citet{enke2023confidence}, while the eighth task presents two identical carts and measures participants' Bayesian updating under the experiment's environment.\footnote{\label{fn:cognitive-tasks}The seven tasks from \citet{enke2023confidence} include base-rate neglect, correlation neglect, balls-and-urns updating, gambler's fallacy, sample-size neglect, regression to the mean, and the acquiring-company problem. \citet{cason2025innovation} also use a separate urn task to measure participants' ability to update beliefs in an environment similar to the main experiment.} Each correct response adds \$0.40 to participants' final earnings. After completing these tasks, participants are asked to predict their task score from 0 to 8 and whether it exceeds the session average score. Each correct prediction earns an additional \$0.40. In Part 4, participants are presented with six pairs of conflicting economic narratives and report their views on each pair using an unincentivized 11-point Likert scale. Finally, Part 5 collects demographic information. Appendix~\ref{app:instructions} shows the instructions and the main screens participants saw.

\subsection{Belief Updating Task} \label{sec:updating-task}

Blocks \textit{SP} and \textit{AP} each consist of 27 scenarios. I use the direct method to elicit posterior beliefs. In each scenario, participants observe one randomly drawn signal. They then report their posterior belief using the displayed priors, the two competing information structures, and the observed signal. Signals are drawn independently across scenarios. A decision problem represents a unique combination of the prior over states and the information structures provided by the pair of competing models. The same nine decision problems appear in both blocks, and participants face each decision problem \(3\) times within a block.\footnote{For each decision problem, three repetitions allow me to observe most participants' beliefs under both signal realizations within a block.} The order of the 27 scenarios within each block is randomized separately for each participant.

\begin{figure}[h!]
    \centering
    \includegraphics[width=0.75\linewidth]{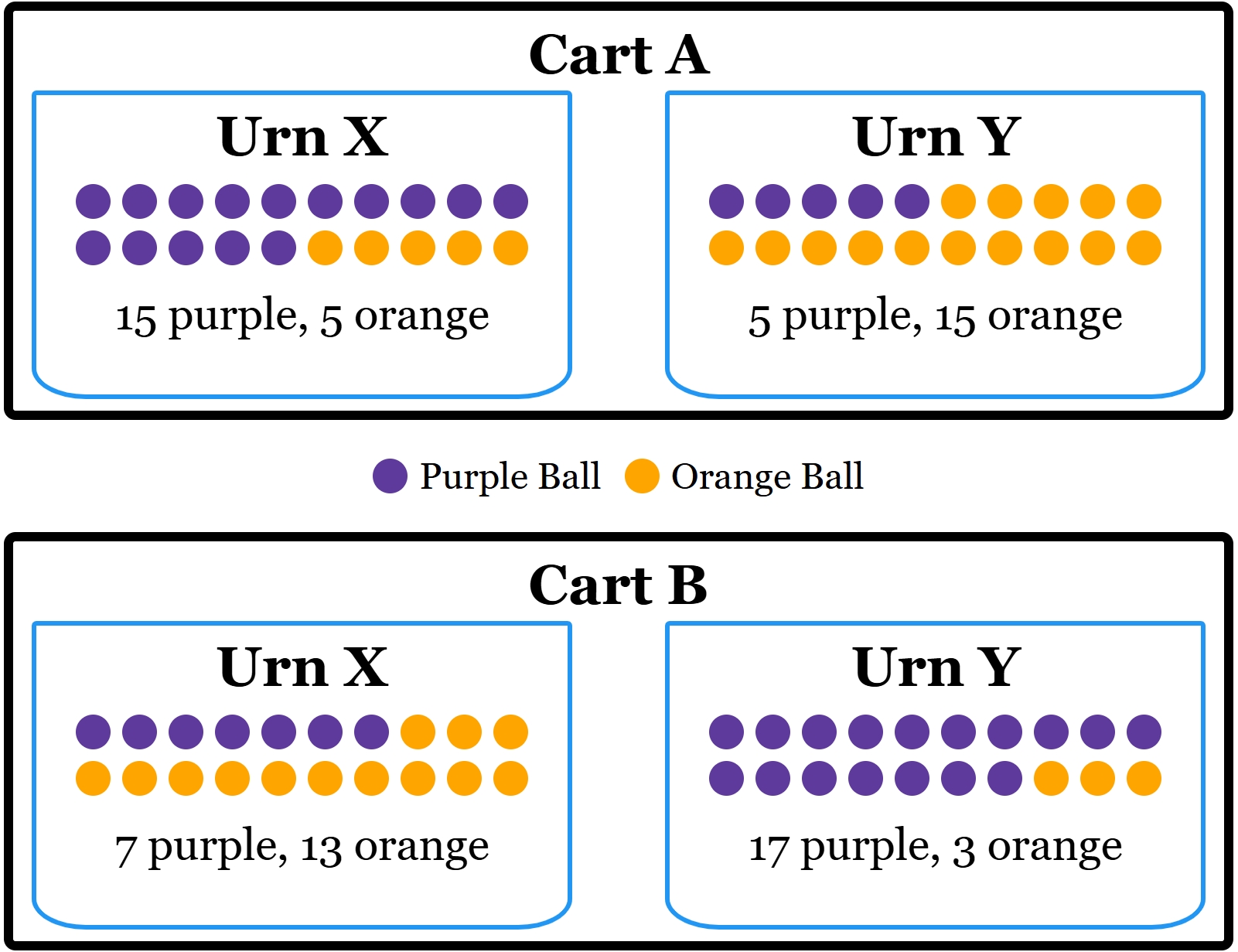}
    \caption{Ball Composition from One Scenario}
    \label{fig:ball_composition}
\end{figure}

In each scenario, participants see \textbf{2} carts (Cart A and Cart B) and \textbf{4} urns (two urns per cart), as Figure \ref{fig:ball_composition} shows.\footnote{The two information structures are randomly assigned to the cart labels and corresponding display positions, with equal probability. For example, the cart shown as Cart A in Figure \ref{fig:ball_composition} can instead be labeled Cart B and displayed at the bottom.} Each cart contains two urns: one labeled X and the other labeled Y. Each urn contains \textbf{20} balls, consisting of a mixture of purple and orange balls. Participants are informed that the computer randomly selects one urn using a two-step process. (1) The computer rolls a fair 10-sided die numbered 1 to 10 to determine whether Urn X or Urn Y is selected. The die ranges selecting Urns X and Y vary across decision problems and are displayed on the screen. For a prior probability of 0.50 on Urn X, rolls 1--5 select Urn X and rolls 6--10 select Urn Y. (2) Independently of the die roll, the computer flips a fair coin to determine which cart is selected. Heads selects Cart A, and tails selects Cart B, each with probability 0.50. Participants do not observe which urn is selected. One ball is randomly drawn from the selected urn, and participants observe only its color (purple or orange). After seeing the ball's color, participants report a probability from 0 to 100 using a slider. This number represents their guess of the probability, expressed as a percentage, that the computer selected Urn X from either cart.

In Part 1 of the experiment, participants are paid based on their decisions from all 54 scenarios across both blocks. Participants' earnings in each scenario come from two sources: the urn payment and the guess prize. The urn payment depends on the urn selected by the computer. In the \textit{SP} block, participants receive 5 tokens regardless of which urn is selected. In the \textit{AP} block, participants receive 25 tokens if Urn X is selected and 5 tokens if Urn Y is selected. Tokens are converted to dollars at a rate of 150 tokens per dollar. For the guess prize, I use the binarized scoring rule \citep{savage1971elicitation, hossain2013binarized} to elicit participants' beliefs. Participants' guesses determine the probability of winning the guess prize of 80 tokens in the corresponding scenario.\footnote{Following the discussion of belief elicitation methods by \citet{danz2022belief, danz2024evaluating, healy2025belief}, I do not directly provide participants with the details of the binarized scoring rule on the interface. However, they can access the rule and a calculator by clicking the corresponding buttons on the interface. The rule description follows \citet{wilson2018paired}, while the interface of the belief-updating task is adapted from \citet{bland2024rounding}.} After each scenario, participants receive no feedback about the true state, their urn payment, or whether they won the guess prize. After completing the experiment, they can view a summary table covering all 54 scenarios.

In this environment, the selected urn's label (X or Y) corresponds to the true state of the world, while the selected cart corresponds to the true model (signal-generating process). The observed ball's color corresponds to the observed data. The urn payment induces participants' preferences over the two states. It depends only on the selected urn, regardless of the selected cart or observed ball color. The urn payment changes payoffs across states without changing the objective information structure. Under Bayesian updating, it does not affect posterior beliefs. Under risk neutrality, it also does not affect the optimal report under the binarized scoring rule.\footnote{With a concave utility function, the state-dependent payment may create hedging incentives. In the \textit{AP} block, a risk-averse participant may report a lower belief in Urn X than they hold. A lower report raises the chance of winning the guess prize when Urn Y is selected, which offsets the lower urn payment in that case. This would move reported beliefs away from the preferred state, which is the opposite direction of motivated reasoning, so it would tend to make the estimated effects of motivated reasoning smaller. I elicit participants' risk attitudes and control for them in regression analysis to further address this concern.} This experimental design allows me to examine whether these induced preferences affect participants' model selection and belief updating processes.

\subsection{Treatment Overview}

The primary objective in designing the experiment is to create a paradigm that allows observation of participants' posterior beliefs in a multi-model environment and to manipulate participants' preferences over states within that environment. In order to identify the behavioral rules participants follow in the experiment, certain parameters are varied across scenarios. The treatment structure is as follows.

\textbf{Within-subject design.} Participants complete 27 scenarios in each of two blocks; block order and scenario order within each block are randomized. The tasks vary across scenarios and blocks along the following dimensions:
\begin{itemize}
    \item \textit{Participants' preference over payoff-relevant states:} In the \textit{SP} block, participants are indifferent between Urns X and Y; in the \textit{AP} block, participants prefer Urn X over Urn Y.
    \item \textit{Prior belief over states:} This parameter may vary across scenarios, allowing observation of different updating behaviors.
    \item \textit{Pair of models provided:} Model pairs vary across decision problems. In 8 out of 9 decision problems, participants are presented with competing models that offer conflicting interpretations of the same observed signal.
\end{itemize}



The same decision problem can generate different signal realizations in different scenarios. Table \ref{tab:dp-list} provides a summary of the 9 decision problems. Let \(s_p\) and \(s_o\) denote the purple and orange signals, respectively.

In Decision Problem 4, after either signal, the objectively best-fit model also produces the larger posterior probability on Urn X. In Decision Problem 5, Cart A and Cart B have the exact same ball composition, reducing the problem to a classical single-model belief-updating task. This decision problem provides a single-model benchmark within the experiment.



\begin{table}
    \caption{Decision Problems}
    \label{tab:dp-list}
        \resizebox{\textwidth}{!}{%
        \renewcommand{\arraystretch}{1.5}
        \small
        \begin{tabular}{
          lccccccccc
        }
            \toprule
            \multirow{2}{*}{No.} & \multirow{2}{*}{State prior \(\mu_0(X)\)} & \multicolumn{2}{c}{Model A} & \multicolumn{2}{c}{Model B} & \multicolumn{2}{c}{Objective Fit for \(s_p\)} & \multicolumn{2}{c}{Best-fit Agents} \\
            \cmidrule(r){3-4} \cmidrule(r){5-6} \cmidrule(r){7-8} \cmidrule(r){9-10}
            & & \(\pi_A(s_p\mid X)\) & \(\pi_A(s_p\mid Y)\) & \(\pi_B(s_p\mid X)\) & \(\pi_B(s_p\mid Y)\) & \(P_A(s_p)\) & \(P_B(s_p)\) & \(\mu_F(X \mid s_p)\) & \(\mu_F(X \mid s_o)\) \\
            \midrule
            1 & 0.50 & 0.70 & 0.25 & 0.35 & 0.85 & 0.475 & 0.600 & 0.292 & 0.286 \\
            2 & 0.50 & 0.70 & 0.25 & 0.40 & 0.95 & 0.475 & 0.675 & 0.296 & 0.286 \\
            3 & 0.50 & 0.65 & 0.15 & 0.40 & 0.95 & 0.400 & 0.675 & 0.296 & 0.292 \\
            4 & 0.50 & 0.85 & 0.35 & 0.25 & 0.70 & 0.600 & 0.475 & 0.708 & 0.714 \\
            5 & 0.50 & 0.75 & 0.20 & 0.75 & 0.20 & 0.475 & 0.475 & 0.789 & 0.238 \\
            6 & 0.30 & 0.80 & 0.35 & 0.30 & 0.75 & 0.485 & 0.615 & 0.146 & 0.117 \\
            7 & 0.30 & 0.75 & 0.25 & 0.35 & 0.85 & 0.400 & 0.700 & 0.150 & 0.125 \\
            8 & 0.70 & 0.60 & 0.15 & 0.40 & 0.85 & 0.465 & 0.535 & 0.523 & 0.523 \\
            9 & 0.70 & 0.55 & 0.05 & 0.45 & 0.95 & 0.400 & 0.600 & 0.525 & 0.525 \\
            \bottomrule
        \end{tabular}
        }
\begin{minipage}{\textwidth}\footnotesize
\vspace{0.5em}
\textit{Notes:} Urns X and Y are the two states, and $s_p$ and $s_o$ denote a purple and an orange ball. $\pi_m(s_p\mid \omega)$ is the probability of drawing a purple ball from Urn $\omega$ under Model $m$, so $\pi_m(s_o\mid \omega)=1-\pi_m(s_p\mid \omega)$. $P_m(s_p)$ is the objective fit of Model $m$ after a purple ball, defined in Equation~\eqref{eq:objective-fit}. After an orange ball, $P_m(s_o)=1-P_m(s_p)$. $\mu_F(X\mid s)$ is the Best-fit benchmark, the posterior probability of Urn X under the model with the higher objective fit, defined in Equation~\eqref{eq:best-fit-rule}. In Decision Problem 5, the two models are identical.
\end{minipage}
\end{table}

\subsection{Hypotheses}

The experiment allows me to examine the behavioral rules participants follow and how asymmetric payoffs affect belief updating.

The Bayesian and best-fit rules in Section~\ref{sec:environment} generate benchmark posterior beliefs for each decision problem and signal realization. I use participants' 27 reported beliefs from the \textit{SP} block to classify them into types based on the behavioral rule they are likely to follow. I then compare reported beliefs between the \textit{SP} and \textit{AP} blocks, controlling for the decision problem and signal realization. I test the following hypotheses using these within-subject comparisons, with participant types determined from \textit{SP} block behavior.

Proposition~\ref{prop:belief-direction} predicts that preference-driven distortion raises the posterior probability of the preferred state under both behavioral rules. In the \textit{AP} block, the urn payment makes Urn X the preferred state. This yields the first hypothesis:

\begin{hypo}[General Motivated Reasoning Effect]
    For identical decision problems, participants will report higher posterior beliefs that Urn X was selected in the \textit{AP} block compared to the \textit{SP} block, conditional on the same signal realization.
\end{hypo}

As described in Section~\ref{sec:environment}, Bayesian agents average model-specific posterior probabilities using updated model weights. Best-fit agents select a single model and update their beliefs using that model. Preference-driven distortion changes model-specific posterior beliefs under both rules. Proposition~\ref{prop:stronger-distortion} shows that the beliefs of Bayesian agents change continuously with the strength of the distortion. The beliefs of best-fit agents change continuously when they do not switch models and can increase discretely when they switch models.

Prediction~\ref{pred:full-switch} and Proposition~\ref{prop:sensitivity} show that sufficiently strong preference-driven distortion can produce a full switch in Decision Problems 1--3 and 6--9. Without distortion, best-fit agents select the model with the lower posterior probability of Urn X after either signal, as Table~\ref{tab:dp-list} shows. Following a full switch, they select the model with the higher posterior probability after each signal. The beliefs of Bayesian agents lie between the two model-specific posterior probabilities, both before and after distortion. For the same degree of distortion, a full switch therefore produces a larger increase in the beliefs of best-fit agents than in those of Bayesian agents.\footnote{When there is no model switch, best-fit agents are affected by preference-driven distortion in a moderate and continuous way, as the distortion affects the model-specific posterior probabilities. However, in that case, the sizes of the belief changes for best-fit agents and Bayesian agents are not obvious.} This comparison motivates the following hypothesis about average motivated reasoning effects across agent types:

\begin{hypo}[Heterogeneity in Motivated Reasoning Effect]
    Asymmetric payoffs have a larger average effect on reported beliefs among participants classified as best-fit agents than among those classified as Bayesian agents.
\end{hypo}

Proposition~\ref{prop:sensitivity} characterizes the distortion required for a full switch. Define the objective fit gap as \(\lvert P_A(s)-P_B(s)\rvert\), the difference between the objective fits of the two models after signal \(s\) (Equation~\eqref{eq:objective-fit}). A full switch is possible in Decision Problems 1--3 and 6--9. Among these problems, those with the same state prior (Decision Problems 1--3, 6--7, and 8--9) have lower switching thresholds when their objective fit gaps are smaller.\footnote{The switching threshold is \(B/A=1+(B-A)/A\), with \(A\) and \(B\) defined in Equation~\eqref{eq:switch-gaps}. Its value depends on both the objective fit gap and \(A\).} For a given distribution of preference-driven distortion, lower thresholds can increase the fraction of best-fit agents who switch models. These switches move agents toward the model that implies a higher posterior probability of selecting Urn X. More frequent switches can therefore amplify the average upward shift in reported beliefs. This switching channel motivates the following empirical hypothesis about average belief effects across decision problems:

\begin{hypo}[Sensitivity in Motivated Reasoning]
    For participants classified as best-fit agents, the average effect of motivated reasoning on reported beliefs increases as the objective fit gap narrows.
\end{hypo}

\subsection{Experimental Procedure}\label{sec:procedure}

I collected experimental data from 193 student participants at the Vernon Smith Experimental Economics Laboratory (VSEEL) at Purdue University in March 2025. The experiment was pre-registered at the AEA RCT Registry in January 2025 (\href{https://doi.org/10.1257/rct.15136}{https://doi.org/10.1257/rct.15136}).\footnote{I pre-registered a target sample of 200 participants. The data were collected across 7 in-person sessions on March 5, March 6, and March 13, 2025. Due to fluctuations in the show-up rate, I eventually collected data from 193 participants.} I drew participants from the VSEEL subject pool (administered using ORSEE; \citealp{greiner2015subject}). The sample is balanced by gender, with 49.7\% female participants. The experiment was programmed in oTree \citep{chen2016otree}. Upon arrival, participants were assigned random participant IDs, and they electronically provided consent to participate. I provided the instructions to each participant on the computer screen, together with printed instructions. The experimenter read the instructions aloud at the front of the room. After that, participants completed ten incentivized quiz questions and earned \$0.50 for each question answered correctly. Participants received immediate feedback to clarify any misunderstandings they might have had, regardless of whether they answered the questions correctly.

\begin{figure}[h!]
    \centering
    \includegraphics[width=0.9\linewidth]{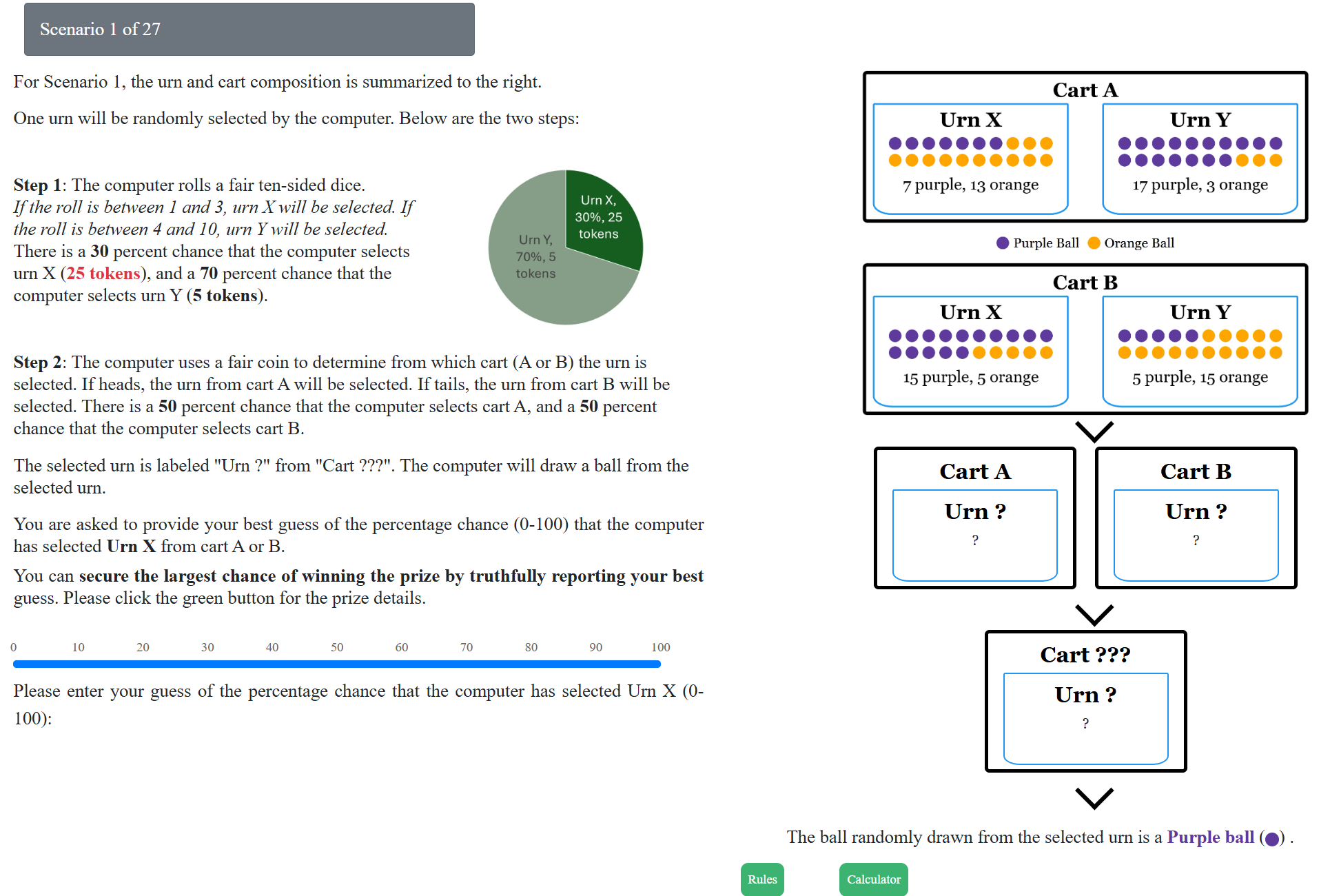}
    \caption{Screenshot of an Example Scenario}
    \label{fig:example_task}
\end{figure}

After the comprehension quiz, the main part of the experiment began. Each participant completed two blocks of belief-updating tasks and made one decision for each scenario. Figure \ref{fig:example_task} presents a screenshot of an example scenario.\footnote{Figures~\ref{fig:instruction_page_8} -- \ref{fig:instruction_page_12} in Appendix~\ref{app:instructions} show the full sequence of screens in a scenario.} The experiment lasted approximately 70 minutes, with an average total payment of \$22.15. Table~\ref{tab:demo} in Appendix~\ref{app:sample_characteristics} provides a summary of the demographic statistics.

The experiment also measures participants' risk attitude, probabilistic reasoning, overconfidence, and opinions on economic narratives in Parts 2--4. In Part 2, I use the bomb risk task from \citet{crosetto2013bomb}. Participants see 100 boxes, and one box is collected every second after they press Start. Each collected box is worth 15 tokens, but one randomly located box contains a bomb that destroys all earnings from the task. Participants decide when to stop collecting, and I use the number of boxes collected to measure risk attitude. A risk-neutral participant collects 50 boxes, and collecting more boxes indicates greater risk tolerance. In Part 3, I measure participants' probabilistic reasoning with eight cognitive tasks (see footnote~\ref{fn:cognitive-tasks} for the list of tasks). Participants have 10 minutes to complete this part. The cognitive score is the number of correct responses across the eight tasks. After these tasks, participants predict how many of the eight tasks they answered correctly. Overconfidence is an indicator equal to one if this predicted score exceeds the actual score. In Part 4, participants rate each of six pairs of conflicting economic narratives on a scale from zero to ten. Narrative extremeness is the mean absolute distance of the six ratings from the midpoint of five. It ranges from zero to five, and higher values indicate more extreme stated opinions.

\FloatBarrier

\section{Results}\label{sec:results}

This section first describes participants' reported beliefs and classifies them using their decisions from the \textit{Symmetric Payoff (SP)} block. It then examines the effect of motivated reasoning on reported beliefs.

\FloatBarrier

\subsection{Description of Behavioral Patterns}\label{sec:descriptive}

Each participant reports beliefs in 27 scenarios in the \textit{SP} block. For each scenario, I derive benchmark beliefs under the Bayesian and Best-fit behavioral rules from the decision problem and signal realization. Table~\ref{tab:unbiased_round_match} shows the shares of reported beliefs within two or five percentage points of these benchmarks. Nearly 29\% of symmetric payoff reports lie within two percentage points of the Bayesian benchmark, compared with 5.5\% for the Best-fit benchmark. With a five-percentage-point window, these shares rise to 47.7\% and 12.2\%, respectively.\footnote{A reported belief can lie within five percentage points of both benchmarks when they are sufficiently close. In Decision Problem 5, for example, both behavioral rules yield the same benchmark.} See Appendix~\ref{sec:robust} for details on the follow-up study supplying model-specific posteriors and its comparison with the main study.

\vspace{0.5em}

\begin{table}[htbp]
\caption{Classification of Round-level Beliefs in the \textit{SP} Block}
\label{tab:unbiased_round_match}
\centering
\begin{threeparttable}
\begin{tabular}{lcccc}
\toprule
\multirow{2}{*}{Type of Guess} & \multicolumn{2}{c}{Within 2 p.p.} & \multicolumn{2}{c}{Within 5 p.p.} \\
\cmidrule{2-3} \cmidrule{4-5}
 & \% & 95\%-CI & \% & 95\%-CI \\
\midrule
Bayesian & 28.98 & [25.29, 32.67] & 47.75 & [43.21, 52.28] \\
Best-fit & 5.51 & [4.47,  6.54] & 12.15 & [10.35, 13.94] \\
\bottomrule
\end{tabular}

\begin{tablenotes}\footnotesize
\item \textit{Notes:} The table reports the percentage of beliefs within the indicated distance of each benchmark. The sample contains 5,211 reports from 193 participants in the \textit{SP} block. The 95\% confidence intervals are based on standard errors calculated from participant-level shares.
\end{tablenotes}
\end{threeparttable}
\end{table}

\begin{figure}[htbp]
\centering
\includegraphics[width=0.9\linewidth]{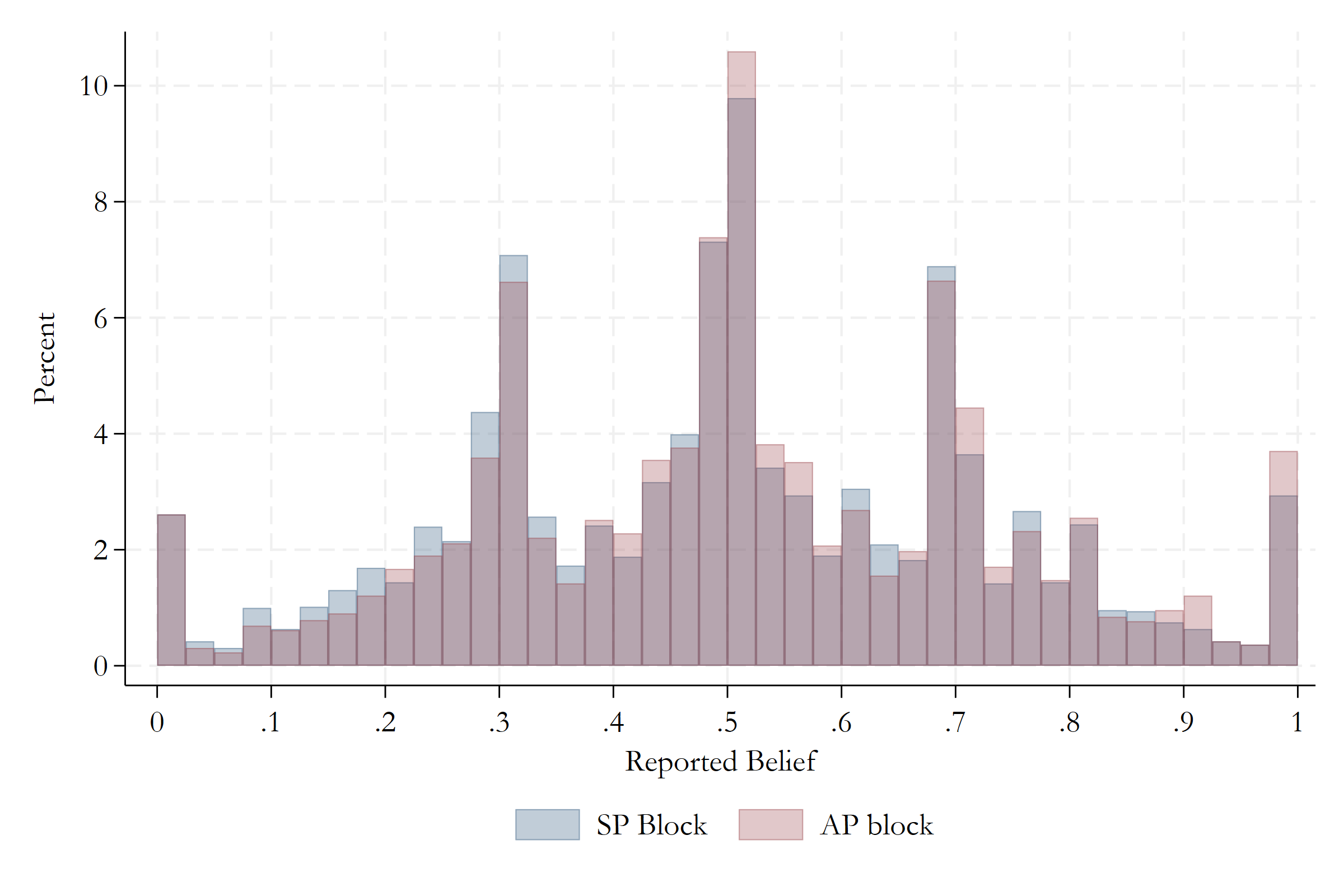}
\begin{minipage}{0.9\linewidth}\footnotesize
\textit{Notes:} The histograms pool all 10,422 scenarios from 193 participants. The horizontal axis shows the reported beliefs of Urn X. Blue bars represent the \textit{SP} block and red bars represent the \textit{AP} block.
\end{minipage}
\caption{Distribution of Reported Beliefs Across All Scenarios}
\label{fig:belief_overall}
\end{figure}

Figure~\ref{fig:belief_overall} displays the distribution of reported beliefs in the \textit{SP} and \textit{AP} blocks. Reported beliefs concentrate around the state priors of 0.3, 0.5, and 0.7 listed in Table~\ref{tab:dp-list}. A nontrivial share of beliefs also occurs at 0 and 1. The distributions overlap a lot, with the \textit{AP} distribution shifted slightly to the right.

Figure~\ref{fig:dp1_descriptive} plots participants' reported beliefs in Decision Problem 1. Each point represents a participant's vector of beliefs following the two signal realizations.\footnote{Participants face this decision problem three times in each block. The figure includes only participants who observe both signal realizations from Decision Problem 1 in both blocks (N = 120). For the signal realization observed twice within a block, the figure uses the first of the participant's two reported beliefs.} The yellow diamond represents the vector of Bayesian benchmark beliefs, and the green diamond represents the vector of Best-fit benchmark beliefs. Following \citet{aina2023tailored}, a vector of posterior beliefs is Bayes-consistent if the state prior is a strict convex combination of the posteriors across signals.\footnote{With a state prior of 0.5, this requires one posterior probability of Urn X to lie above 0.5 and the other to lie below 0.5, or both to equal 0.5.} The shaded areas mark the region of vectors of beliefs that are not Bayes-consistent. The left panel plots each participant's vector of beliefs in the \textit{SP} block. Many points cluster near the Bayesian benchmark, while a smaller group clusters near the Best-fit benchmark in the shaded area. Similar to \citet{aina2025weighting}, about half of the vectors are Bayes-inconsistent.

\begin{figure}[htbp]
\centering
\includegraphics[width=0.9\linewidth]{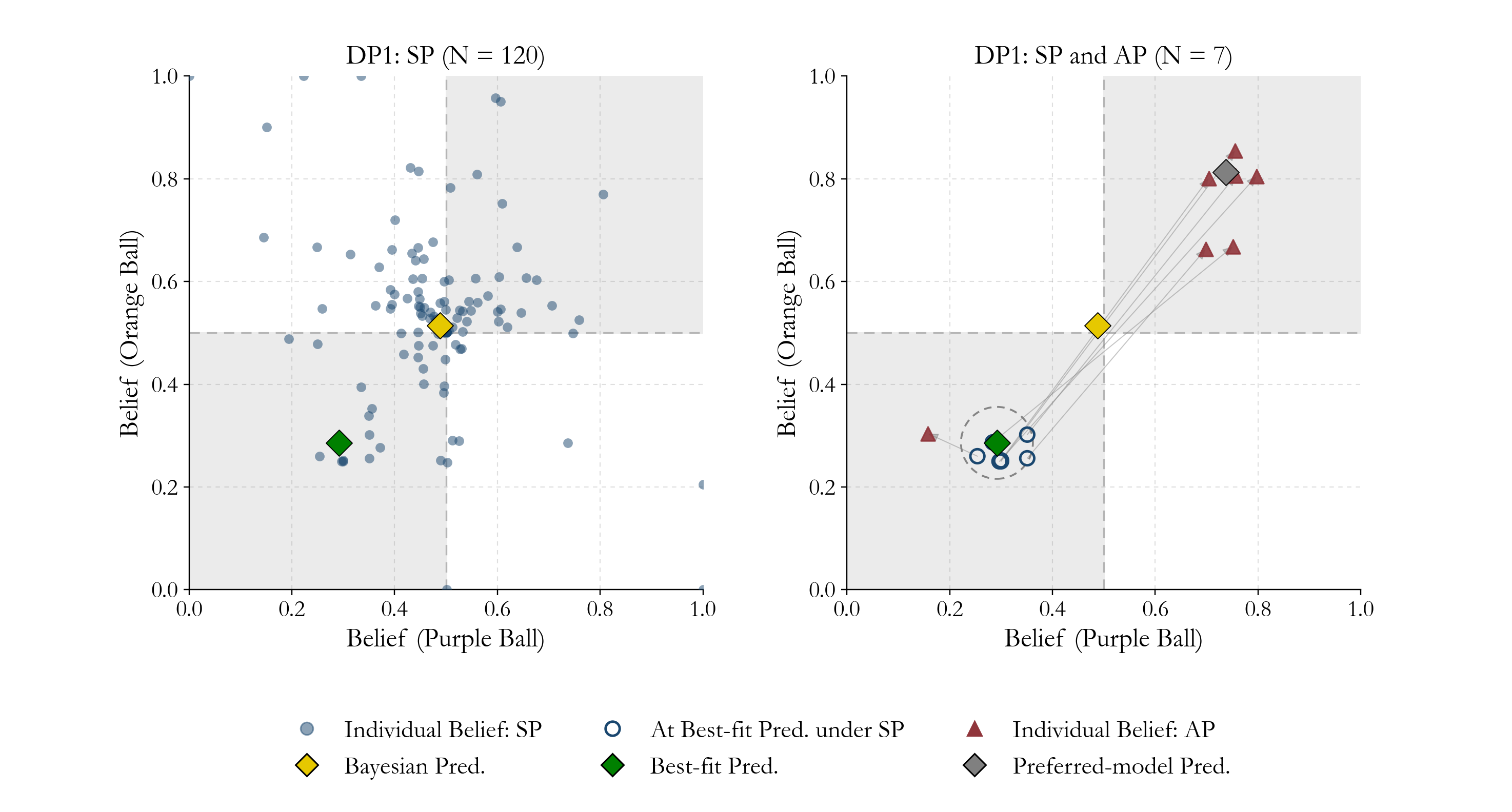}
\begin{minipage}{0.9\linewidth}\footnotesize
\textit{Notes:} Both panels use Decision Problem 1 and include only participants who observe both signal realizations in both blocks. Each belief is the participant's first reported belief after the given signal realization within a block. The horizontal axis shows beliefs after observing a purple ball, and the vertical axis shows beliefs after observing an orange ball. The left panel plots each participant's pair of beliefs in the \textit{SP} block. The right panel includes the participants whose \textit{SP} pair lies within a Euclidean distance of $0.05\sqrt{2}$ of the Best-fit benchmark, shown by the dashed circle. For each of these participants, a hollow circle marks the \textit{SP} pair, a triangle marks the \textit{AP} pair, and an arrow connects the two. The yellow and green diamonds represent the Bayesian and Best-fit benchmarks. The grey diamond represents the benchmark of using the model that implies a higher probability of Urn X. Shaded areas mark the Bayes-inconsistent region.
\end{minipage}
\caption{Vector of Reported Beliefs in Decision Problem 1}
\label{fig:dp1_descriptive}
\end{figure}

The right panel of Figure~\ref{fig:dp1_descriptive} provides a simple illustration of how participants near the Best-fit benchmark respond to asymmetric payoffs. It plots the seven participants whose \textit{SP} beliefs in Decision Problem 1 lie inside the dashed circle around the Best-fit benchmark. In the \textit{AP} block, six of them exhibit a large, discrete jump in their belief vectors, and their vectors are close to the grey diamond, which represents the beliefs implied by using the model with the higher posterior probability of Urn X after each signal realization. None of the seven moves to the Bayesian benchmark. This descriptive pattern is consistent with the full switch in Prediction~\ref{pred:full-switch}, in which a best-fit agent changes the selected model after both signal realizations. Section~\ref{sec:preference_effects} estimates the effects of asymmetric payoffs by behavioral
type.

\FloatBarrier

\subsection{Structural Classification of Agent Types}\label{sec:structural_classification}

To better understand these descriptive patterns and formally test the theoretical hypotheses, I use a finite mixture model to classify participants by their behavioral rules. This classification allows me to examine how asymmetric payoffs affect beliefs across behavioral types.

Classifying participants into distinct behavioral types presents several challenges. First, the behavioral rules participants follow are not directly observed. Second, different rules can yield similar or identical predictions in some decision problems. Third, computational and reporting errors can cause observed beliefs to deviate from a rule's predictions, even when a participant intends to follow that rule \citep{alos2023part}.

The structural classification therefore models errors explicitly. Without errors, each rule predicts a single posterior belief in each scenario. Reported beliefs rarely match these predictions exactly. Only 47.7\% of reports in the \textit{SP} block lie within five percentage points of the Bayesian benchmark, and 12.1\% lie within five percentage points of the Best-fit benchmark (Table~\ref{tab:unbiased_round_match}). A classification without errors would therefore leave many participants inconsistent with every behavioral rule. As in \citet{el1995people}, allowing for errors gives each behavioral rule a positive likelihood for every observed behavior. For example, a Best-fit updater may misidentify the best-fitting model because of a mistake in calculating model fit. This mistake matters most when the two models fit the observed signal similarly, because a small error in fit calculation can switch the selected model and move the reported belief to the other model's posterior. Participants may also report beliefs that differ from their intended responses.\footnote{When participants see the same decision problem and signal more than once in the \textit{SP} block, their reports differ by 9.3 percentage points on average, and only 14\% of repeated reports match exactly (Appendix~\ref{app:repeatability}).} Allowing for these errors means that the same reported belief can be produced by multiple behavioral rules, with different likelihoods. The model evaluates each participant's beliefs across all 27 scenarios in the \textit{SP} block to help distinguish behavioral rules from pure errors. As a result, a single report that deviates from a type's prediction carries limited weight in the classification.\footnote{The classification excludes all reports from the \textit{AP} block, because deriving benchmark beliefs for this block would require additional assumptions about participants' preference distortion term \(f(U(\omega))\).}

In the finite mixture model, each participant belongs to one of three latent types specified \textit{ex ante}: Bayesian, Best-fit, and Non-mover.\footnote{The Bayesian and Best-fit types are directly informed by the theory. The main specification includes a separate Non-mover type because Bayesian benchmarks lie close to the state prior in some decision problems. This type helps distinguish Bayesian updating from non-moving behaviors. See Appendix~\ref{app:robustness_types} for details on the estimation of the alternative two-, four-, and five-type classifications.} Bayesian agents average model-specific posteriors using their model weights; Best-fit agents select the model with the highest fit; and Non-movers retain the state prior. Following \citet{el1995people}, the model assumes that each participant follows the same underlying rule across scenarios.

The model allows two layers of mean-zero normal errors, with standard deviations shared across types. The first layer is an error in calculating model fit. In the \textit{SP} block, a participant perceives the objective fit of each model, defined in Equation~\eqref{eq:objective-fit}, with a normal error. Bayesian agents weight the model-specific posteriors using these perceived fits, and Best-fit agents select the model with the higher perceived fit. The second layer is a reporting error. It is added to the belief implied by the participant's rule, so the reported belief can deviate from that belief.\footnote{With the reporting error, the belief implied by the rule and the error can fall below zero or above one, but participants can only report beliefs between zero and one. Therefore, the likelihood uses normal densities for reports strictly between zero and one. For reports at zero and one, it uses the normal probabilities below zero and above one, respectively.} Non-movers have only reporting errors, as they only focus on the state prior. 

For example, consider a purple signal in Decision Problem 1 (see Table~\ref{tab:dp-list}). The objective fits of Models A and B are 0.475 and 0.600, and their model-specific posterior probabilities of Urn X are 0.737 and 0.292. Without errors, a Best-fit agent selects Model B and reports 0.292, and a Bayesian agent reports 0.488. A fit-calculation error can reverse the order of the two perceived fits. A Best-fit agent then selects Model A, and the belief implied by the rule becomes 0.737. Such a reversal is more likely when the two objective fits are close. A reporting error can move the reported belief away from the belief implied by the rule, for example from 0.292 to 0.302.

I estimate population shares and the two error standard deviations, one for each layer, by maximum likelihood, using all 5,211 reports from the \textit{SP} block. Population shares are the estimated fractions of each type in the population. Table~\ref{tab:parameter_estimates} presents the estimates, and Appendix~\ref{app:structural_details} describes the likelihood and classification procedure.



\begin{table}[htbp]
\caption{Three-Type Parameter Estimates}
\label{tab:parameter_estimates}
\centering
\begin{threeparttable}
\begin{tabular}{lcccc}
\toprule
\textbf{Parameter} && \textbf{Estimate} && \textbf{Bootstrap SE} \\
\hline
\multicolumn{2}{l}{\textit{Population shares}} \\
\(p_1\) (Bayesian) && 0.683 && 0.037 \\
\(p_2\) (Best-fit) && 0.174 && 0.026 \\
\(p_3\) (Non-mover) && 0.144 && 0.029 \\
\addlinespace
\multicolumn{2}{l}{\textit{First-layer (fit-calculation) error SD}} \\
\(\sigma_{1}\) (Bayesian and Best-fit) && 0.143 && 0.011 \\
\(\sigma_{1,\text{Non-mover}}\) && \multicolumn{1}{c}{---} \\
\addlinespace
\multicolumn{2}{l}{\textit{Second-layer (reporting) error SD}} \\
\(\sigma_{2}\) (all types) && 0.132 && 0.006 \\
\bottomrule
\end{tabular}

\begin{tablenotes}\footnotesize
\item \textit{Notes:} The table reports maximum likelihood estimates of population shares and error standard deviations shared across types. The sample contains 5,211 beliefs from 193 participants in the \textit{SP} block, with 27 beliefs per participant. Population shares are the estimated fractions of each type in the population. Standard errors are calculated from 500 bootstrap samples. See Appendix~\ref{app:structural_details} for details.
\end{tablenotes}
\end{threeparttable}
\end{table}

 Using Decision Problem 1 again as an example, the objective fits of the two models differ by 0.125 after either signal. At the estimated \(\sigma_1\) of 0.143, the model implies that a Best-fit agent selects the model with the lower objective fit after about 27\% of signals, and that a Bayesian agent places a higher weight on the model with the lower objective fit after about 27\% of signals in this problem.

\FloatBarrier

For each participant, I use the estimated population shares as prior probabilities of belonging to each behavioral type. I then update these priors using the likelihood of the participant's 27 reports in the \textit{SP} block under each type. This likelihood measures how well the type's predictions and estimated error standard deviations account for the participant's beliefs in the \textit{SP} block. A higher likelihood increases a type's posterior probability relative to other types, holding the prior probabilities fixed. Appendix~\ref{app:structural_details} explains how these likelihoods are calculated. I assign each participant to the type with the highest posterior probability and report the classification in Table~\ref{tab:individual_estimates}. Of 193 participants, 132 are classified as Bayesian, 33 are classified as Best-fit, and 28 are classified as Non-movers.\footnote{For 167 of the 193 participants, the assigned type has a posterior probability above 90\%. For 192 of the 193 participants, the assigned type has a posterior probability above 50\%. For one participant, the assigned type has a posterior probability of 0.42. Among the 167 participants whose assigned type has a posterior probability above 90\%, 31 (19\%) are classified as Best-fit and 118 (71\%) as Bayesian.}

\begin{table}[htbp]
\caption{Participant Classification}
\label{tab:individual_estimates}
\centering
\begin{threeparttable}
\begin{tabular}{lcccc}
\toprule
Type & Assigned & Above 50\% & Above 70\% & Above 90\% \\
\midrule
Bayesian & 132 & 132 & 129 & 118 \\
Best-fit & 33 & 32 & 32 & 31 \\
Non-mover & 28 & 28 & 25 & 18 \\
Unclassified & - & 1 & 7 & 26 \\
\bottomrule
\end{tabular}

\begin{tablenotes}\footnotesize
\item \textit{Notes:} The table reports the number of participants classified as each behavioral type. The sample contains 193 participants. The Assigned column classifies each participant as the type with the highest posterior probability. The Above 50\%, Above 70\%, and Above 90\% columns count only participants whose highest posterior probability exceeds 50\%, 70\%, and 90\%, respectively.
\end{tablenotes}
\end{threeparttable}
\end{table}

\begin{samepage}
\begin{result}
Approximately two-thirds of participants are classified as Bayesian, while approximately one-sixth are classified as Best-fit.
\end{result}
\end{samepage}

The structural classification suggests that most participants weight information from competing models, while about one-sixth of participants rely on a single model when updating beliefs. This proportion relies on the assumption that all types share the same error standard deviations. This assumption is strong, because participants who follow different rules may also differ in how precisely they calculate model fits and report beliefs. I therefore also estimate a version that allows type-specific error standard deviations. This version classifies 63 participants (33\%) as Best-fit (see Appendix~\ref{app:type_specific_errors} for details). All 33 participants classified as Best-fit in the main specification remain Best-fit, and 26 of the 30 additional Best-fit participants are classified as Bayesian in the main specification.

More participants are classified as Best-fit because, when error standard deviations differ across types, the Best-fit type can fit noisy reports with larger errors. The estimated Best-fit type has a much larger fit-calculation error than the Bayesian type (Table~\ref{tab:ts_parameter_estimates} in Appendix~\ref{app:type_specific_errors}), so it often selects the model with the lower objective fit, whose posterior is closer to the Bayesian benchmark than to the Best-fit benchmark. It can therefore fit participants whose reports are more dispersed around the Bayesian benchmark. With shared error standard deviations, the classification relies less on how dispersed a participant's reports are, so I use shared error standard deviations in the main specification and report the type-specific version in Appendix~\ref{app:type_specific_errors}.

\begin{figure}[htbp]
\centering
\includegraphics[width=0.88\linewidth]{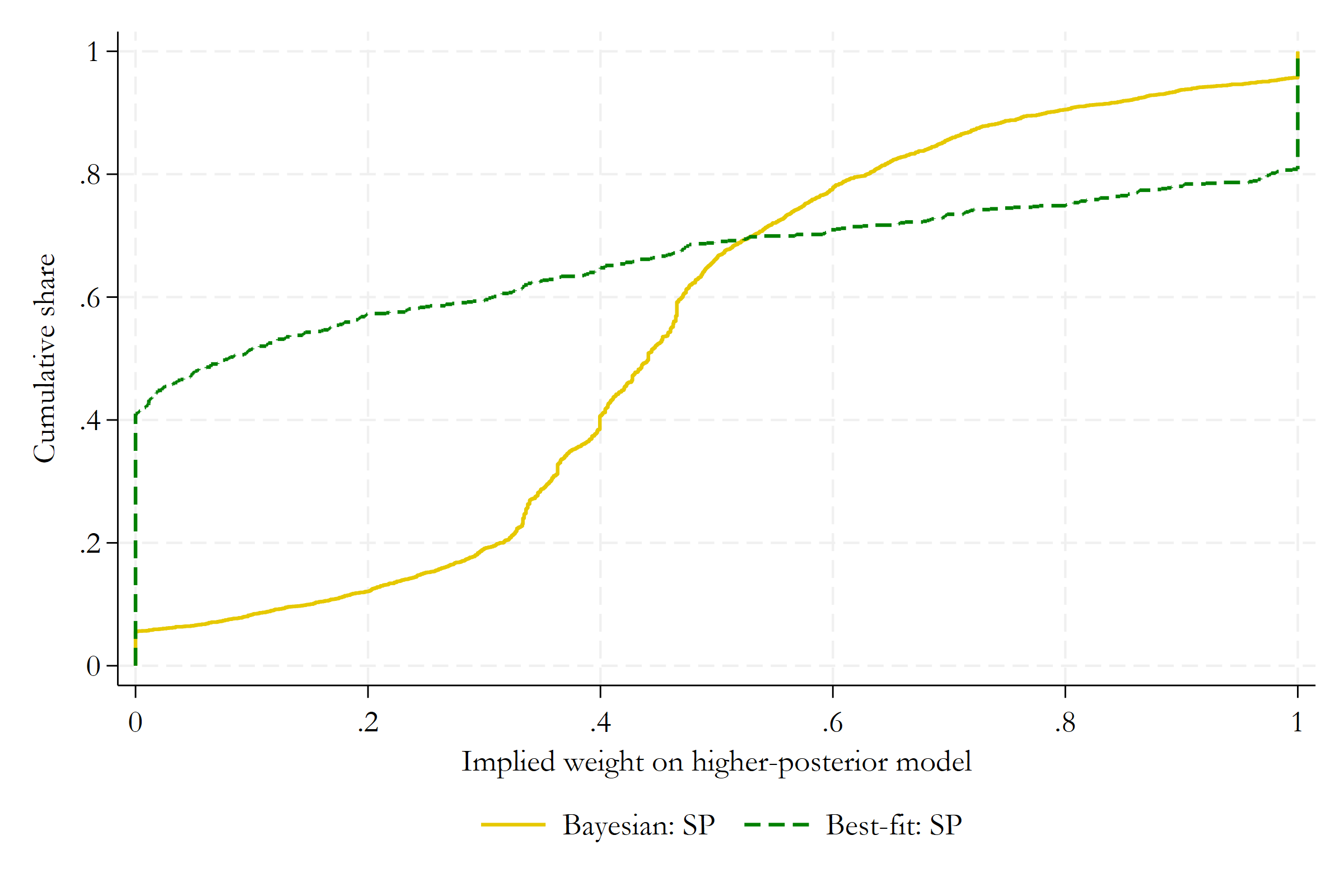}
\begin{minipage}{0.88\linewidth}\footnotesize
\textit{Notes:} The figure shows the cumulative distributions of implied weights on the higher model-specific posterior, separately for participants classified as Bayesian and Best-fit. The sample includes reports from Decision Problems 1--4 and 6--9 in the \textit{SP} block.
\end{minipage}
\caption{Report-level Implied Weights in the \textit{SP} Block}
\label{fig:weights_capped_up}
\end{figure}

Figure~\ref{fig:weights_capped_up} provides descriptive evidence on the classification results by showing how participants' beliefs relate to the two model-specific posteriors. For each report, I derive the implied weight on the higher model-specific posterior, which expresses the report as a weighted average of the two model-specific posteriors. In the example of Decision Problem 1 above, the two model-specific posteriors after a purple signal are 0.737 and 0.292. A reported belief of 0.5 is therefore a weighted average with a weight of \((0.5-0.292)/(0.737-0.292)\approx0.47\) on the higher posterior.\footnote{For beliefs below the lower model-specific posterior, the implied weight is set to zero. For beliefs above the higher model-specific posterior, it is set to one. In the \textit{SP} block, 17.3\% of reports in Decision Problems 1--4 and 6--9 lie outside the range of the two model-specific posteriors. In a related experiment, \citet{aina2025weighting} find that 12.1\% of guesses lie outside the range of the two model predictions.} Decision Problem 5 is excluded because the two model-specific posteriors coincide. The figure shows that implied weights for participants classified as Bayesian concentrate at intermediate values, while implied weights for Best-fit participants have larger masses at zero and one. These patterns are broadly consistent with the distinction between the two behavioral rules. See Appendix~\ref{app:weights} for report-level implied weights in the \textit{AP} block.

\FloatBarrier

\subsection{Effects of Motivated Reasoning}\label{sec:preference_effects}

This section compares beliefs between the \textit{SP} and \textit{AP} blocks to test the hypotheses about motivated reasoning. In the \textit{AP} block, participants receive a higher urn payment when Urn X is selected. Hypothesis 1 predicts higher beliefs of Urn X in the \textit{AP} block than in the \textit{SP} block. Hypothesis 2 predicts a larger average effect of motivated reasoning on reported beliefs among Best-fit participants than among Bayesian participants.

I first estimate the following specification:
\begin{equation}
\text{Belief}_{it}=\alpha+\beta_1\mathbbm{1}_{\text{Asym},it}
 +\boldsymbol{\beta}_2\mathbf{X}_i+\lambda_{d_{it}}+\eta_{s_{it}}+\tau_t+\epsilon_{it}.
\label{eqn:ols}
\end{equation}
Here, $\text{Belief}_{it}\in[0,1]$ is participant $i$'s reported belief of Urn X in round $t$. The indicator $\mathbbm{1}_{\text{Asym},it}$ equals one in the \textit{AP} block. The participant covariates $\mathbf{X}_i$ include gender, race, cognitive score, risk attitude, overconfidence, and narrative extremeness, defined in Section~\ref{sec:procedure}. I also include fixed effects for decision problem, signal color, and round. Standard errors are clustered by participant.\footnote{The results are robust to participant fixed effects. See Appendix~\ref{app:core_results} for details on these estimates and the first-block comparison.}

Figure~\ref{fig:mr_by_type} reports the estimated coefficient on the \textit{AP} block indicator for the pooled sample and separately for each behavioral type. Conditional on the specified participant covariates and fixed effects, this coefficient measures how asymmetric payoffs affect posterior beliefs of Urn X, and I refer to this coefficient as the effect of motivated reasoning. Consistent with Hypothesis 1, asymmetric payoffs increase the reported belief of Urn X by 1.7 percentage points overall ($p<0.001$). The estimated effect is 0.5 percentage points among Bayesian participants ($p=0.245$) and 8.4 percentage points among Best-fit participants ($p=0.001$). In addition, the estimated effect is significantly larger among Best-fit participants than among Bayesian participants ($p<0.001$), consistent with Hypothesis 2.\footnote{One may be concerned that participants who complete the \textit{AP} block first carry its effects into the \textit{SP} block, which would affect both their classification and the estimated effect of motivated reasoning. The order of the two blocks was randomized, and 91 participants completed the \textit{SP} block first. For these participants, types are classified only from beliefs reported before they faced asymmetric payoffs. Among them, the effect of motivated reasoning is 6.5 percentage points larger among Best-fit participants than among Bayesian participants ($p=0.014$).}

\begin{figure}[htbp]
\centering
\includegraphics[width=0.9\linewidth]{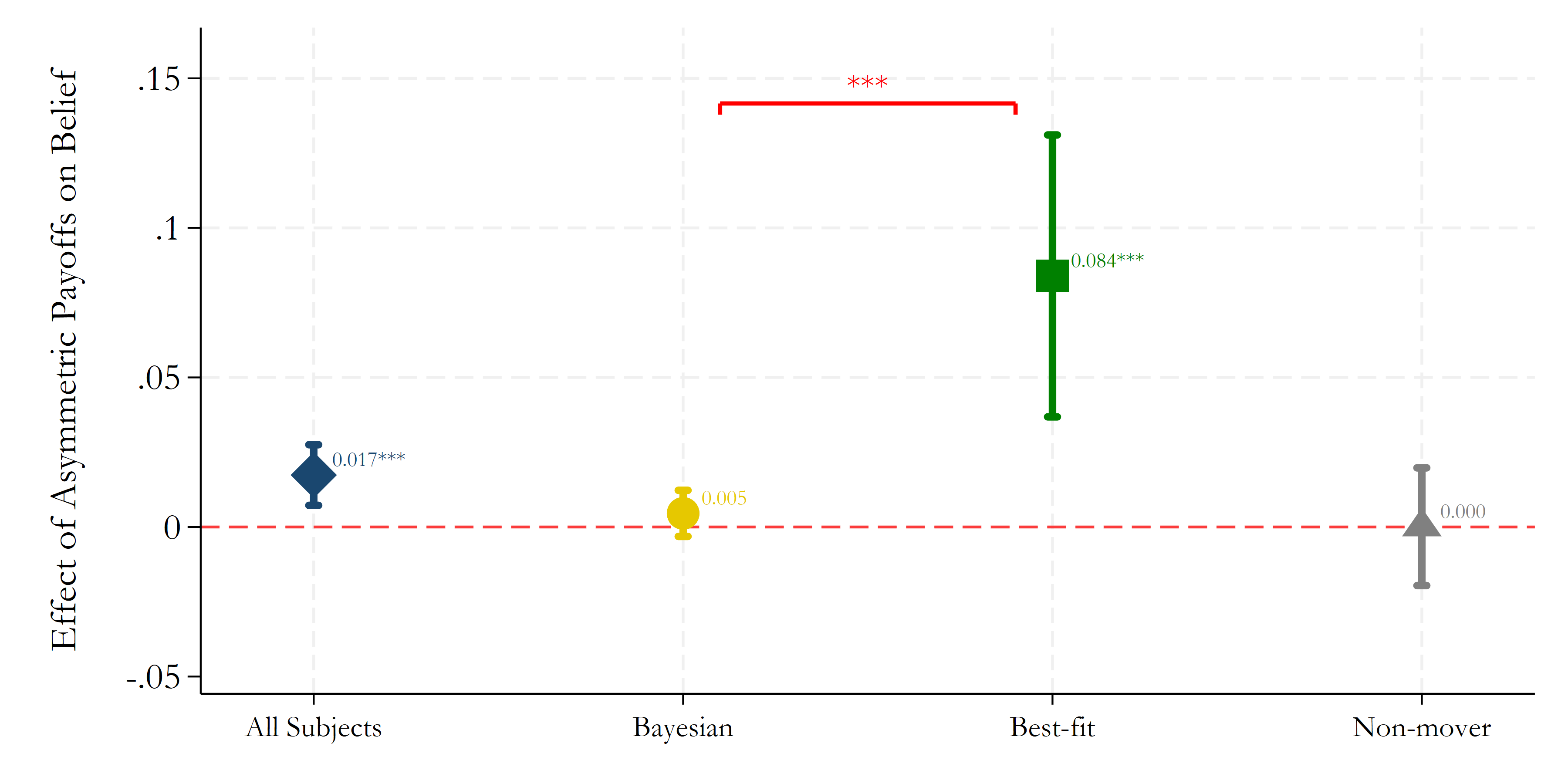}
\begin{minipage}{0.9\linewidth}\footnotesize
\textit{Notes:} The figure reports estimated coefficients on $\mathbbm{1}_{\text{Asym}}$ in Equation~\eqref{eqn:ols} for the pooled sample and separately for each behavioral type. The pooled, Bayesian, Best-fit, and Non-mover samples contain 193, 132, 33, and 28 participants, and each participant contributes 54 beliefs. Regressions control for gender, race, cognitive score, risk attitude, overconfidence, and narrative extremeness, and include fixed effects for signal color, round, and decision problem. Bars indicate 95\% confidence intervals based on standard errors clustered at the participant level. The red bracket reports the significance of the difference between the Best-fit and Bayesian coefficients. $^{*}p<0.10$, $^{**}p<0.05$, $^{***}p<0.01$.
\end{minipage}
\caption{Motivated Reasoning Effect by Type}
\label{fig:mr_by_type}
\end{figure}

\FloatBarrier
\vspace{0.25cm}

\begin{samepage}
\begin{result}
Asymmetric payoffs have a significantly larger effect on Best-fit participants' reported beliefs than on Bayesian participants' reported beliefs.
\end{result}
\end{samepage}

\vspace{0.25cm}
\FloatBarrier

I next examine how the effect of motivated reasoning varies with posterior type probabilities. I interact $\mathbbm{1}_{\text{Asym}}$ with each participant's posterior probabilities of being the Best-fit and Non-mover types, using the Bayesian type as the reference group. These probabilities are estimated from each participant's 27 reports in the \textit{SP} block. Table~\ref{tab:mr_likelihood} reports a positive coefficient of 0.082 on the interaction between $\mathbbm{1}_{\text{Asym}}$ and the Best-fit probability ($p<0.001$). Thus, the estimated effect of motivated reasoning increases with the Best-fit probability, which is consistent with the separate estimates in Figure~\ref{fig:mr_by_type}.\footnote{See Appendix~\ref{app:core_results} for details on effects of motivated reasoning and closeness to Best-fit versus Bayesian predictions measured using reports from other decision problems.}

\clearpage

\begin{table}[htbp]
\caption{Motivated Reasoning Effect by Likelihood of Type}
\label{tab:mr_likelihood}
\centering
\begin{threeparttable}
{
\def\sym#1{\ifmmode^{#1}\else\(^{#1}\)\fi}
\begin{tabular}{l*{4}{c}}
\toprule
DV: Reported Belief &&&& (1) \\
\midrule
\(\mathbbm{1}_{\text{Asym}}\) &&&& 0.004 \\
 &&&& (0.004) \\
\(\mathbbm{1}_{\text{Asym}}\) \(\times\) Likelihood of Best-fit Type &&&& 0.082\sym{***} \\
 &&&& (0.023) \\
\(\mathbbm{1}_{\text{Asym}}\) \(\times\) Likelihood of Non-mover Type &&&& -0.007 \\
 &&&& (0.011) \\
Likelihood of Best-fit Type &&&& -0.069\sym{***} \\
 &&&& (0.018) \\
Likelihood of Non-mover Type &&&& 0.004 \\
 &&&& (0.008) \\
\midrule
Observations &&&& 10422 \\
Participants &&&& 193 \\
Mean of DV &&&& 0.506 \\
\(R^2\) &&&& 0.343 \\
\bottomrule
\end{tabular}
}

\begin{tablenotes}\footnotesize
\item \textit{Notes:} The table reports OLS estimates of how the effect of motivated reasoning varies with posterior type probabilities. The dependent variable is reported belief. Unreported controls include gender, race indicators, cognitive score, risk attitude, overconfidence, and narrative extremeness. The specification includes fixed effects for signal color, round, and decision problem. Standard errors clustered at the participant level appear in parentheses. $^{*}p<0.10$, $^{**}p<0.05$, $^{***}p<0.01$.
\end{tablenotes}
\end{threeparttable}
\end{table}

The main results on motivated reasoning are robust to alternative two-, four-, and five-type classifications. Asymmetric payoffs have a significantly larger effect among Best-fit participants than among Bayesian participants under all four specifications. See Appendix~\ref{app:robustness_types} for details on the corresponding estimated effects of motivated reasoning and regressions using posterior type probabilities.

As described in Section~\ref{sec:structural_classification}, the main specification assumes that the error standard deviations are shared across types. When I relax this assumption and allow type-specific error standard deviations, asymmetric payoffs still have a significantly larger effect among Best-fit participants than among Bayesian participants. The estimated effects are 4.9 and 0.1 percentage points, respectively, compared with 8.4 and 0.5 percentage points in Figure~\ref{fig:mr_by_type}, and their difference is statistically significant ($p<0.001$). See Appendix~\ref{app:type_specific_errors} for details on this specification.

\FloatBarrier

\subsection{Tests of the Model-Switch Predictions}\label{sec:prediction_tests}
                                          
I next examine how the effect of motivated reasoning varies across decision problems. In Decision Problems 1--3 and 6--9, the objectively best-fitting model assigns a lower posterior probability to Urn X than the other model after either signal. A sufficiently strong preference-driven distortion may therefore lead Best-fit agents to select the model with the higher posterior probability. Prediction~\ref{pred:full-switch} and Proposition~\ref{prop:sensitivity} establish the conditions for this switch. In Decision Problem 4, the objectively best-fitting model already assigns the higher posterior probability to Urn X. In Decision Problem 5, participants see two identical models, so the two model-specific posteriors coincide. Neither problem therefore allows a switch from a model with a lower posterior probability to one with a higher posterior probability. Preference-driven distortion may still affect Best-fit agents' beliefs through changes in model-specific posterior probabilities.\footnote{Decision Problems 1 and 4 have the same state prior, objective fit gap, and gap between model-specific posteriors. The best-fitting model implies the lower posterior probability of Urn X in Decision Problem 1 and the higher posterior probability in Decision Problem 4. Appendix~\ref{app:core_results} compares their estimated effects to examine the role of model switching.}

\begin{figure}[htbp]
\centering
\includegraphics[width=0.9\linewidth]{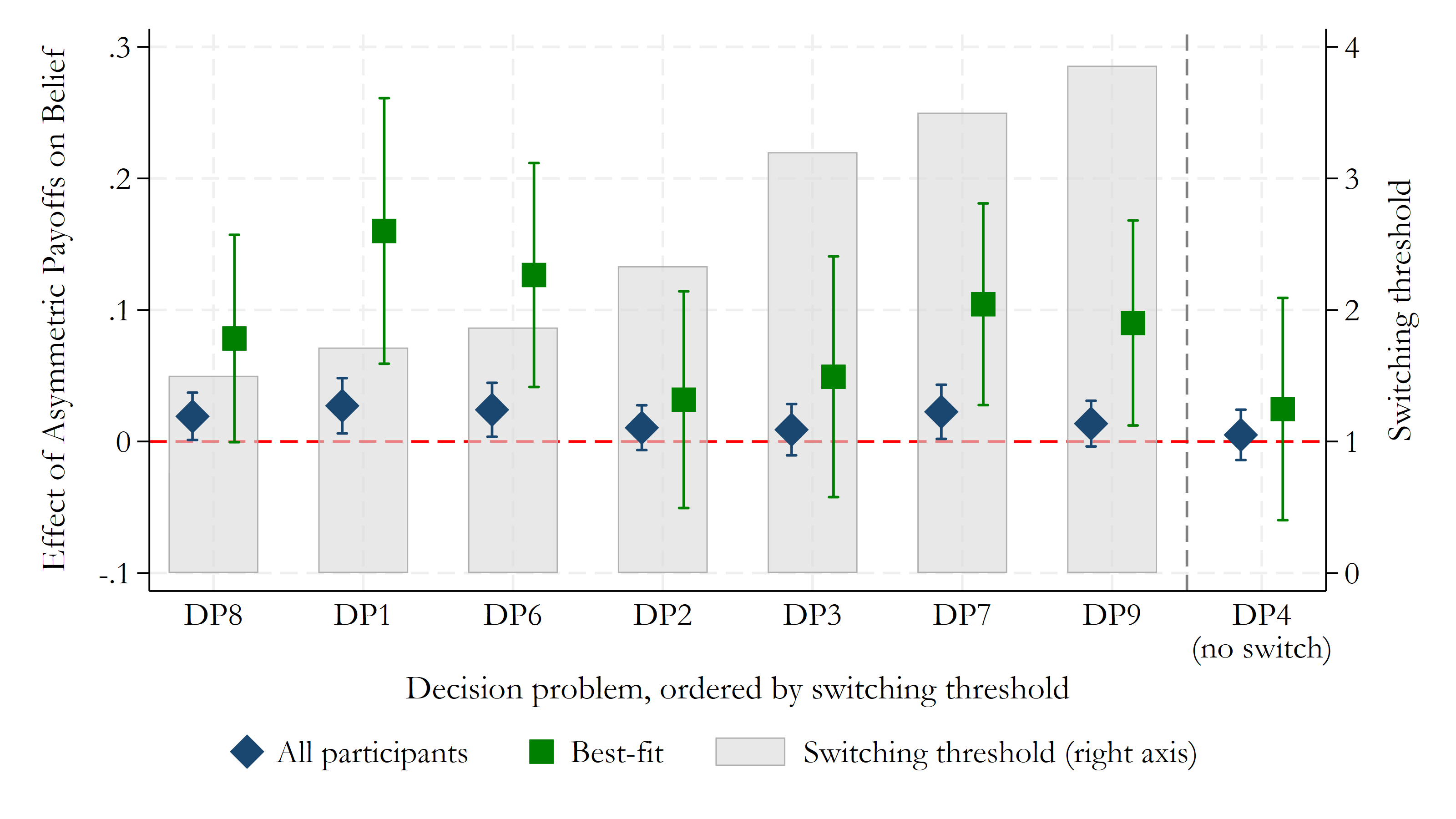}
\begin{minipage}{0.9\linewidth}\footnotesize
\textit{Notes:} The figure reports estimated effects of motivated reasoning on the reported beliefs of Urn X for each decision problem (left axis). Navy diamonds represent estimates for the pooled sample, and green squares represent estimates for participants classified as Best-fit. Decision problems are ordered by the switching threshold, shown by the grey bars (right axis). Decision Problem 4 appears last because no switch is possible. Decision Problem 5 is excluded because participants see identical models. The regressions include participant, decision-problem-by-signal, and round fixed effects. Bars indicate 95\% confidence intervals based on standard errors clustered at the participant level.
\end{minipage}
\caption{Motivated Reasoning Effects Across Decision Problems}
\label{fig:all_dp_effects}
\end{figure}

Figure~\ref{fig:all_dp_effects} reports the estimated effects of motivated reasoning in each decision problem separately. The decision problems are ordered by their switching threshold, \(1/R\), discussed in Proposition~\ref{prop:sensitivity}. The switching threshold is the smallest preference-driven distortion toward Urn X that leads a Best-fit agent to select the model with the higher posterior probability of Urn X. A lower threshold means that a weaker preference is enough to switch models. In Decision Problem 4, the best-fitting model already implies the higher posterior probability of Urn X, so no switch is possible. Several Best-fit estimates are significantly positive, but they are imprecise and do not decline steadily as the threshold rises.\footnote{State priors may affect how easily participants assess model fit. Decision Problems 1--3 share a state prior of 0.5, under which participants can determine each model's fit by counting the balls of the observed color in the cart. Among these three problems, the Best-fit estimate is largest in Decision Problem 1, which has the lowest threshold.}

I next test Hypothesis 3 with regressions. Hypothesis 3 implies that the effect among Best-fit participants should be larger in problems with lower thresholds. I use two measures of how much preference-driven distortion is needed to switch models: the objective fit gap, $|P_A(s)-P_B(s)|$, and the switching threshold, $1/R$, discussed in Proposition~\ref{prop:sensitivity} and used in Hypothesis 3. Using Decision Problems 1--3 and 6--9, I interact $\mathbbm{1}_{\text{Asym}}$ with each measure in a separate regression. The regressions include participant, decision-problem-by-signal, and round fixed effects.\footnote{Decision Problems 4--5 are excluded as no switch is possible.} They also interact $\mathbbm{1}_{\text{Asym}}$ with state-prior indicators, which allows the effect of motivated reasoning to vary across state priors.

\begin{table}[htbp]
\caption{Motivated Reasoning, Model Fit, and Switching Thresholds}
\label{tab:mr_gap}
\centering
\begin{threeparttable}
\begin{tabular}{@{}c@{}}
\begin{adjustbox}{max width=\linewidth}
{
\def\sym#1{\ifmmode^{#1}\else\(^{#1}\)\fi}
\begin{tabular}{l*{6}{c}}
\toprule
DV: Reported Belief && (1) & (2) && (3) & (4) \\
\midrule
\(\mathbbm{1}_{\text{Asym}}\) && 0.035\sym{**} & 0.177\sym{***} && 0.034\sym{**} & 0.157\sym{**} \\
 && (0.015) & (0.063) && (0.015) & (0.064) \\
\(\mathbbm{1}_{\text{Asym}}\) \(\times\) objective fit gap && -0.054 & -0.269 &&  &  \\
 && (0.044) & (0.168) &&  &  \\
\(\mathbbm{1}_{\text{Asym}}\) \(\times\) switching threshold &&  &  && -0.004 & -0.014 \\
 &&  &  && (0.004) & (0.014) \\
\midrule
Sample && All & Best-fit && All & Best-fit \\
Observations && 8106 & 1386 && 8106 & 1386 \\
Participants && 193 & 33 && 193 & 33 \\
Mean of DV && 0.505 & 0.465 && 0.505 & 0.465 \\
\(R^2\) && 0.513 & 0.380 && 0.513 & 0.379 \\
\bottomrule
\end{tabular}
}

\end{adjustbox}
\end{tabular}
\begin{tablenotes}\footnotesize
\item \textit{Notes:} The table reports OLS estimates of how the effect of motivated reasoning varies with the objective fit gap and switching threshold. The dependent variable is reported belief. The sample includes Decision Problems 1--3 and 6--9. Columns (1) and (3) use all participants, while columns (2) and (4) use participants classified as Best-fit. The objective fit gap is $|P_A(s)-P_B(s)|$, and the switching threshold is $1/R$, discussed in Proposition~\ref{prop:sensitivity}. All specifications include interactions between $\mathbbm{1}_{\text{Asym}}$ and state-prior indicators, together with participant, decision-problem-by-signal, and round fixed effects. Standard errors clustered at the participant level appear in parentheses. $^{*}p<0.10$, $^{**}p<0.05$, $^{***}p<0.01$.
\end{tablenotes}
\end{threeparttable}
\end{table}

Table~\ref{tab:mr_gap} reports negative interaction coefficients for both measures. The estimates are consistent with Hypothesis 3 directionally, although none of the interaction coefficients is statistically significant. With 33 Best-fit participants and seven decision problems, the design has limited power to detect how the effect of motivated reasoning varies across decision problems. Another possible reason is that preference-driven distortion differs across participants. A Best-fit participant with a strong distortion switches models in every decision problem, and one with a weak distortion never switches. Only participants with an intermediate distortion switch in decision problems with low thresholds but not in those with high thresholds. Differences in the average effect across decision problems therefore come only from this intermediate group, which may be quite small.

\subsection{Belief Extremeness and Individual Characteristics}\label{sec:determinants}

This section examines the extremeness of reported beliefs, measured by their absolute distance from 0.5 \citep{fan2026narratives}. A reported belief of 0.1 or 0.9 has an extremeness of 0.4, while a reported belief of 0.5 has an extremeness of zero. In Decision Problems 1--4 and 6--7, the state priors are 0.5 and 0.3, and the Best-fit benchmark is further from 0.5 than the Bayesian benchmark. In Decision Problems 8--9, where the state prior is 0.7, the Best-fit benchmark instead predicts less extreme beliefs than the Bayesian benchmark. The Bayesian and Best-fit benchmarks coincide in Decision Problem 5. Across all nine decision problems in the \textit{SP} block, mean belief extremeness is 0.292 among participants classified as Best-fit and 0.162 among those classified as Bayesian (\(p < 0.001\)). In Decision Problems 1--4, the corresponding means are 0.263 and 0.079 (\(p < 0.001\)).

Figure~\ref{fig:extremeness_types} shows the cumulative distributions of report-level extremeness from Decision Problems 1--5 in the \textit{SP} block. In Decision Problems 1--4, the Best-fit rule generates more extreme benchmark beliefs than the Bayesian rule. In Decision Problem 5, both rules produce the same benchmark. Best-fit participants report higher extremeness levels in both Decision Problems 1--4 (\(p < 0.001\)) and Decision Problem 5 (\(p = 0.038\)), while the difference in mean belief extremeness between Best-fit and Bayesian participants is significantly larger in Decision Problems 1--4 than in Decision Problem 5 (\(p<0.001\)). One may be concerned that the difference in belief extremeness is driven entirely by the parameters of the decision problems. The larger gap in Decision Problems 1--4 suggests that differences between the rules' predictions help explain why Best-fit participants report more extreme beliefs. However, Best-fit participants also report more extreme beliefs on average in Decision Problem 5, where both rules predict the same beliefs. This finding suggests that the difference in belief extremeness is not solely driven by benchmark predictions.

\begin{figure}[htbp]
\centering
\begin{minipage}{0.49\linewidth}\centering
\small DP 1--4\\
\includegraphics[width=\linewidth]{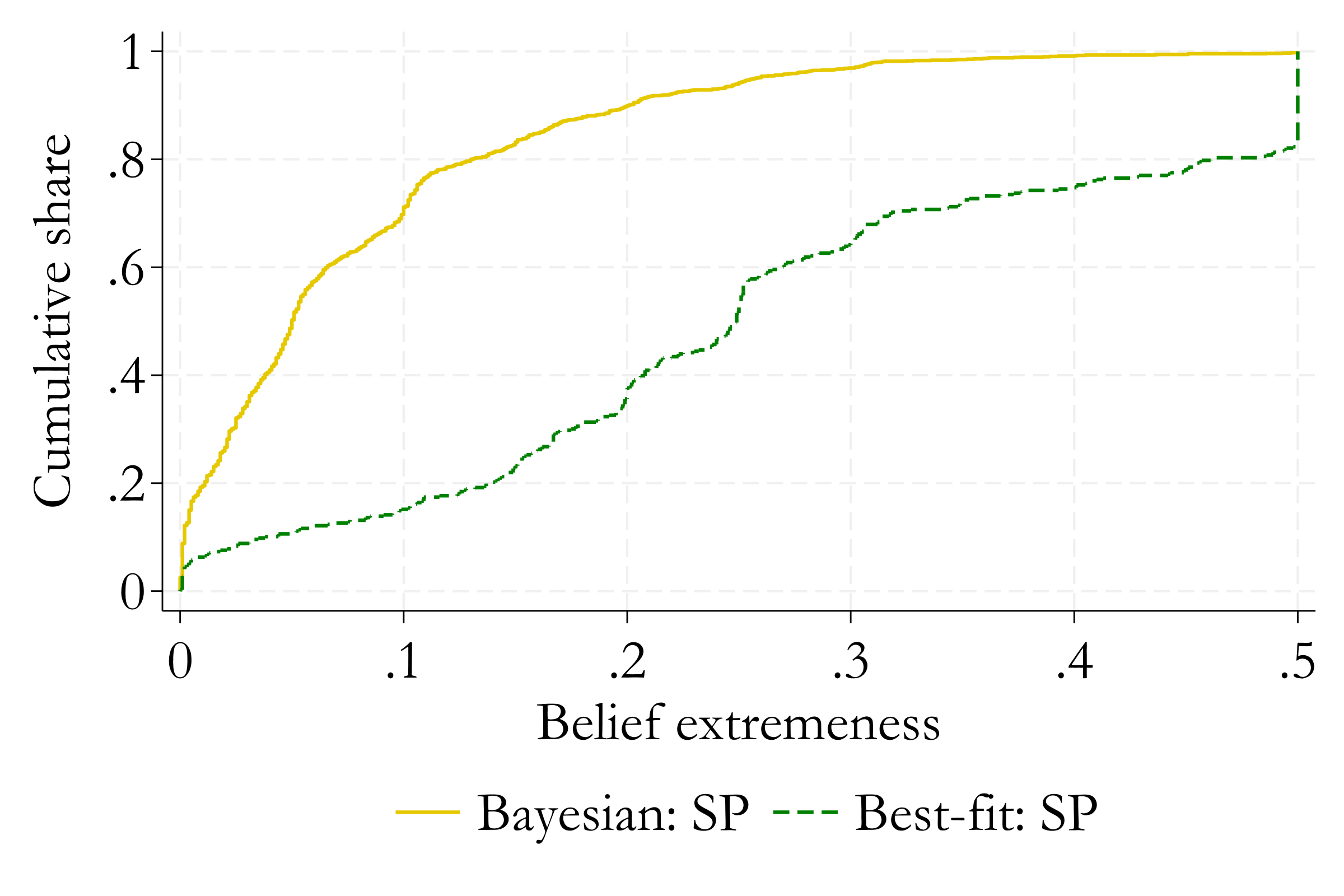}
\end{minipage}\hfill
\begin{minipage}{0.49\linewidth}\centering
\small DP 5\\
\includegraphics[width=\linewidth]{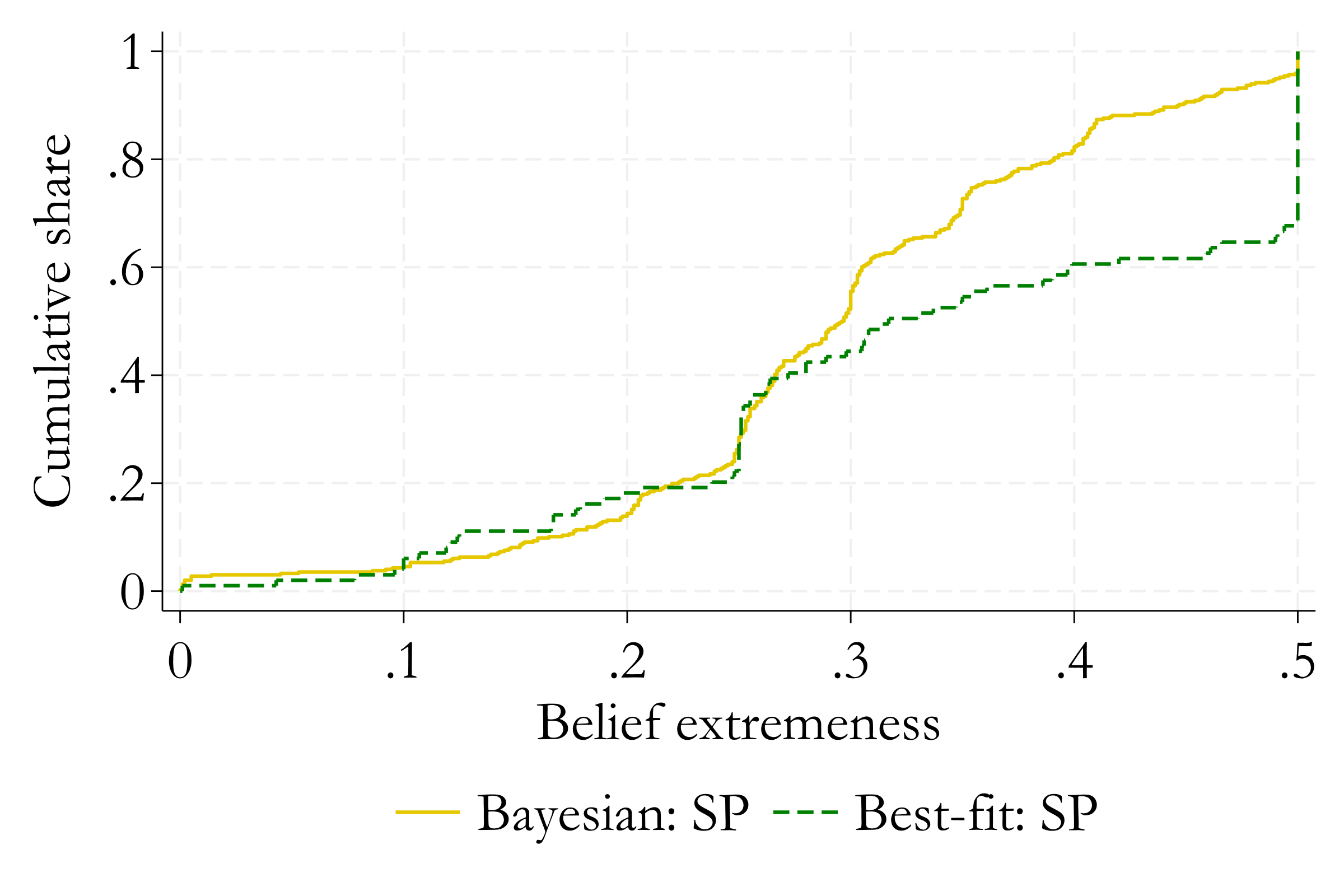}
\end{minipage}
\begin{minipage}{0.95\linewidth}\footnotesize
\textit{Notes:} The curves show the cumulative distributions of belief extremeness, $|\text{Belief}_{it}-0.5|$, in the \textit{SP} block for participants classified as Bayesian and Best-fit. The left panel uses Decision Problems 1--4, and the right panel uses Decision Problem 5.
\end{minipage}
\caption{Belief Extremeness by Type in the Symmetric Payoff Block}
\label{fig:extremeness_types}
\end{figure}

I also explore which individual characteristics are associated with behavioral type classification. Table~\ref{tab:type_determinant_logit} reports average marginal effects from separate logit regressions for Bayesian and Best-fit types. The regressions include gender, race indicators, cognitive score, risk attitude, overconfidence, narrative extremeness, quiz score, and previous experimental participation. Narrative extremeness, defined in Section~\ref{sec:procedure}, comes from the unincentivized narrative survey at the end of the experiment.

\begin{table}[htbp]
\caption{Individual Characteristics and Type Classification}
\label{tab:type_determinant_logit}
\centering
\begin{threeparttable}
{
\def\sym#1{\ifmmode^{#1}\else\(^{#1}\)\fi}
\begin{tabular}{l*{4}{c}}
\toprule
\multirow{2}{*}{DV: Type Classification} && Bayesian && Best-fit \\
\cmidrule{3-5}
&& (1) && (2) \\
\midrule
Female && 0.032 && -0.004 \\
 && (0.066) && (0.054) \\
Quiz Score && 0.056 && 0.025 \\
 && (0.040) && (0.037) \\
Cognitive Score && 0.010 && -0.008 \\
 && (0.028) && (0.023) \\
Risk Attitude && 0.005\sym{***} && -0.001 \\
 && (0.002) && (0.001) \\
Overconfidence && -0.075 && -0.034 \\
 && (0.080) && (0.064) \\
Narrative Extremeness && -0.094\sym{**} && 0.067\sym{**} \\
 && (0.037) && (0.029) \\
\midrule
Observations && 193 && 193 \\
Mean of DV && 0.684 && 0.171 \\
Pseudo \(R^2\) && 0.109 && 0.121 \\
\bottomrule
\end{tabular}
}

\begin{tablenotes}\footnotesize
\item \textit{Notes:} The table reports average marginal effects from binary logit regressions. The dependent variables are binary indicators for Bayesian and Best-fit types in columns (1) and (2). Each regression uses all 193 participants. All displayed characteristics enter both regressions. The regressions include gender, race indicators, cognitive score, risk attitude, overconfidence, narrative extremeness, quiz score, and previous experimental participation. Delta-method standard errors appear in parentheses. $^{*}p<0.10$, $^{**}p<0.05$, $^{***}p<0.01$.
\end{tablenotes}
\end{threeparttable}
\end{table}

The regression results show that narrative extremeness is positively associated with being classified as the Best-fit type in the belief updating tasks in this experiment. The corresponding association with Bayesian classification is negative, and both associations are statistically significant at the 5\% level. Risk attitude is positively associated with Bayesian classification ($p<0.01$).\footnote{This association is economically modest. A one-standard-deviation increase in the number of boxes collected in the bomb risk task (about 21 boxes) raises the probability of Bayesian classification by about 10 percentage points. Hedging against the urn payment is also unlikely to explain this association. The classification uses only beliefs from the \textit{SP} block, where the urn payment is the same for both urns, so participants have little reason to hedge. In the \textit{AP} block, hedging could lower reported beliefs of Urn X, especially among more risk-averse participants. Participants classified as Best-fit are more risk-averse on average, so hedging would, if anything, reduce their estimated effect of motivated reasoning relative to Bayesian participants.} The estimated coefficients of quiz score, cognitive score, and overconfidence are not statistically significant for either type.

This finding suggests that the behavioral rules participants use in the laboratory may sometimes reflect tendencies that extend beyond the abstract belief-updating tasks. Participants classified as Best-fit also tend to express stronger views on economic narratives. However, narrative extremeness is measured through an unincentivized survey. This association therefore provides weak evidence of a broader connection. See Appendix~\ref{app:narrative_associations} for details on type-assignment regressions including Non-movers and associations between narrative extremeness and belief updating.

\FloatBarrier

\subsection{Further Discussion}\label{sec:interpretation}

The results connect differences in belief updating and model selection to differences in responses to asymmetric payoffs. Asymmetric payoffs increase reported beliefs in the preferred state, which supports Hypothesis 1. This effect is significantly larger among participants classified as Best-fit than among those classified as Bayesian, which supports Hypothesis 2.

The results do not provide clear support for Hypothesis 3. Among Best-fit participants, the point estimates suggest a larger effect of motivated reasoning when the objective fit gap is smaller or the switching threshold is lower, but these estimates are imprecise and never statistically significant. The objective fit gap may not be the only determinant of Best-fit participants' model-switching behavior. Salience, complexity, or model-specific posteriors might also affect this process. In addition, Hypothesis 3 tests a further link between model switching and the effect of motivated reasoning. Preference-driven distortion affects Best-fit agents through two channels: model selection and model-specific posterior beliefs. The second channel still operates when model switching is less frequent. Together, these issues limit how clearly I can test this hypothesis.

I also explore the extremeness of reported beliefs across behavioral rules and find that Best-fit participants report more extreme beliefs on average. I further explore the potential connection between behavior in this belief updating task and behavior in broader settings. The results suggest that participants classified as Best-fit tend to hold stronger opinions on economic policy narratives.

\FloatBarrier

\section{Conclusion}\label{sec:conclusion}

This paper examines how preferences over payoff-relevant states affect model selection and belief updating when multiple competing models can explain the same information. I develop a theoretical framework in which preference-driven distortion affects how individuals perceive the information structure. The same distortion applies to Bayesian and Best-fit behavioral rules, and it distorts Bayesian agents' model weights and may switch Best-fit agents' model selection. In the experimental setting with two states, two signals, and two models, stronger bias shifts beliefs toward the preferred state under both rules. Bayesian beliefs change continuously as the distortion increases. Best-fit beliefs can also jump when individuals switch to a model that assigns a higher probability to the preferred state. This theoretical framework characterizes when these model switches occur and their dependence on the information structure.

The experiment shows that asymmetric payoffs have a larger effect on participants classified as best-fit updaters than on those classified as Bayesian updaters. Behavior in the symmetric payoff block therefore helps predict whose beliefs are more affected by preferences. This finding suggests that differences in model selection and belief updating matter for understanding motivated reasoning. Even when individuals face the same information and competing explanations, their responses to preferences can differ substantially.

The framework offers one explanation for this difference. Preferences can change how well a model appears to explain the data, making a model that favors the preferred state more likely to be selected. Individuals may therefore arrive at a preferred belief while following what they perceive as the best explanation of the evidence. 
The experimental results are consistent with this explanation, although the variation across decision problems predicted by Hypothesis 3 is not statistically significant. However, the experiment elicits only beliefs, so I cannot directly measure how much of the effect of motivated reasoning comes from model switches. Eliciting model choices alongside beliefs would help separate model switches from changes in model-specific posteriors within the selected model.

These results have implications for efforts to reduce biased beliefs by presenting competing explanations. Both models are available in the experiment, and individuals differ in whether they weight them or select one model to update beliefs. The effectiveness of providing alternative explanations may therefore depend on how individuals use them when updating beliefs. A useful next step is to examine whether encouraging individuals to weight multiple models could reduce the effect of motivated reasoning.

A further question is when individuals adopt different updating rules. The experiment uses simple models that participants can compare directly. In many settings, individuals face several complex models whose assumptions and predictions are difficult to assess. Comparing and evaluating these models may require considerable time and effort. As a result, individuals may then rely on a single model, and their preferences may influence which model they select. Future experiments could vary the number of models or the complexity of the model. This would help us understand whether individuals change their updating rules across settings and whether these changes affect the magnitude of motivated reasoning.

The supply of competing models also deserves further attention. In this experiment, the models are given by the experimenter, and their availability does not depend on participants' preferences. In political and financial settings, interested parties often choose which explanations to present. A sender who knows the receiver's preferences may offer a model that supports the receiver's preferred state while fitting the observed data. The optimal sending strategy may also depend on whether the receiver weights competing models or selects one. Future research could allow senders to choose models and explore how senders respond to differences in receivers' preferences and updating behavior, and whether their choices strengthen the effect of motivated reasoning.

\clearpage

\bibliographystyle{aea}
\bibliography{Literature.bib}

\newpage
\appendix

\renewcommand{\thefigure}{\Alph{section}-\arabic{figure}}
\setcounter{figure}{0}
\renewcommand{\thetable}{\Alph{section}-\arabic{table}}
\setcounter{table}{0}

\renewcommand{\thepage}{Appendix \Alph{section}, p.~\arabic{page}}
\setcounter{page}{1}

\renewcommand{\theequation}{\Alph{section}-\arabic{equation}}
\renewcommand{\theHequation}{appendix.\Alph{section}.\arabic{equation}}
\setcounter{equation}{0}

\section{Further Theoretical Analysis}

This appendix derives several identities used in Section~\ref{sec:theory} and proves each proposition and prediction. The assumptions specified in Sections~\ref{sec:environment} and \ref{sec:distortion} apply throughout.

\subsection{Model Fit}\label{app:fit-identities}

For every model, objective and subjective fits satisfy
\begin{equation}\label{eq:fit-sums}
    \sum_{s\in S}P_m(s)=1,
    \qquad
    \sum_{s\in S}\hat P_m(s)
    =
    \sum_{\omega\in\Omega}\mu_0(\omega)f(U(\omega)).
\end{equation}
Both equalities follow by exchanging the state and signal sums and using \(\sum_{s\in S}\pi_m(s\mid\omega)=1\). Their right-hand sides are independent of the model.

Suppose \(S=\{s_1,s_2\}\). Subtracting Equation~\eqref{eq:fit-sums} for models \(m\) and \(m'\) gives Equation~\eqref{eq:fit-reversal}. This derivation allows any finite number of states.

Taking the ratio of the weights assigned to \(m\) and \(m'\), before and after preference-driven distortion, and using Equation~\eqref{eq:subj-fit} gives
\[
    \frac{\hat\rho(m\mid s)}
    {\hat\rho(m'\mid s)}
    =
    \frac{\rho(m\mid s)}
    {\rho(m'\mid s)}
    \frac{
        \sum_{\omega\in\Omega}
        \mu_m(\omega\mid s)f(U(\omega))
    }{
        \sum_{\omega\in\Omega}
        \mu_{m'}(\omega\mid s)f(U(\omega))
    }.
\]
Preference-driven distortion increases the weight on \(m\) relative to \(m'\) exactly when the posterior expectation of \(f(U)\) is larger under \(m\).

A related observation holds with any finite signal space. Suppose \(M=\{m,m'\}\) and \(\pi_m(s\mid\omega)=\pi_{m'}(s\mid\omega)\) for every \(s\in S\) and \(\omega\neq\bar\omega\). Define \(\delta(s)=\pi_m(s\mid\bar\omega)-\pi_{m'}(s\mid\bar\omega)\). Then
\[
    P_m(s)-P_{m'}(s)=\mu_0(\bar\omega)\delta(s),
    \qquad
    \hat P_m(s)-\hat P_{m'}(s)
    =f(U(\bar\omega))\big[P_m(s)-P_{m'}(s)\big].
\]
The objective and subjective fit differences therefore have the same sign after every signal. Preference-driven distortion cannot change the model selected by a best-fit agent.

\clearpage

\subsection{Proofs}

\subsubsection{Proof of Proposition~\ref{prop:fit-representative}}
\label{app:proof-fit-representative}

\begin{proof}
Suppose \(m\) is more representative than \(m'\) for \(s\). The subjective fit difference is
\[
    \hat P_m(s)-\hat P_{m'}(s)
    =\sum_{\omega\in\Omega}\mu_0(\omega)f(U(\omega))
    \big[\pi_m(s\mid\omega)-\pi_{m'}(s\mid\omega)\big].
\]
Every term is nonnegative. At least one is positive because \(\mu_0(\omega)>0\), \(f(U(\omega))>0\), and one likelihood inequality is strict. Hence, \(\hat P_m(s)>\hat P_{m'}(s)\), which proves part 1.

The common denominator in Equation~\eqref{eq:distorted-weight} cancels from the ratio of model weights. Therefore,
\[
    \frac{\hat\rho(m\mid s)}{\hat\rho(m'\mid s)}
    =\frac{\rho_0(m)}{\rho_0(m')}
    \frac{\hat P_m(s)}{\hat P_{m'}(s)}
    >\frac{\rho_0(m)}{\rho_0(m')}.
\]
This proves part 2. Part 1 also implies that \(m'\) cannot maximize subjective fit, which proves part 3.
\end{proof}

\subsubsection{Proof of Prediction~\ref{pred:full-switch}}
\label{app:proof-full-switch}

\begin{proof}
By assumption, the objective and subjective fit differences between \(m\) and \(m'\) are both nonzero. Equation~\eqref{eq:fit-reversal} implies that each difference changes sign across the two signals.

Suppose the selections based on objective and subjective fit differ after \(s_1\). Both fit rankings reverse for \(s_2\), so the selections also differ after \(s_2\). The switch is therefore full.

If the selections agree after \(s_1\), both ranking reversals make them agree after \(s_2\). The selected model is then unchanged after both signals.
\end{proof}

\subsubsection{Proof of Proposition~\ref{prop:sensitivity}}
\label{app:proof-sensitivity}

\begin{proof}
Because \(S\) has two signals and the models are distinct, their likelihoods of \(s_1\) differ in at least one state. If they were equal in the other state, one model would be more representative. Thus, the likelihood differences are nonzero and have opposite signs.

To determine their signs, suppose that
\[
    \pi_m(s_1\mid\omega_1)>\pi_{m'}(s_1\mid\omega_1),
    \qquad
    \pi_m(s_1\mid\omega_2)<\pi_{m'}(s_1\mid\omega_2).
\]
Bayes' rule would imply
\[
    \frac{\mu_m(\omega_1\mid s_1)}
    {\mu_m(\omega_2\mid s_1)}
    >
    \frac{\mu_{m'}(\omega_1\mid s_1)}
    {\mu_{m'}(\omega_2\mid s_1)}.
\]
With two states, the posterior odds of \(\omega_1\) relative to \(\omega_2\) are strictly increasing in the posterior probability of \(\omega_1\). The displayed inequality contradicts \(\mu_m(\omega_1\mid s_1)<\mu_{m'}(\omega_1\mid s_1)\). Hence,
\[
    \pi_m(s_1\mid\omega_1)<\pi_{m'}(s_1\mid\omega_1),
    \qquad
    \pi_m(s_1\mid\omega_2)>\pi_{m'}(s_1\mid\omega_2).
\]
Equation~\eqref{eq:switch-gaps} therefore gives \(A>0\) and \(B>0\).

The objective fit difference is
\[
    P_m(s_1)-P_{m'}(s_1)=B-A>0.
\]
Thus, \(B>A>0\) and \(R=A/B\in(0,1)\).

The subjective fit difference equals
\[
    \hat P_m(s_1)-\hat P_{m'}(s_1)
    =f(U(\omega_2))B-f(U(\omega_1))A
    =f(U(\omega_2))(B-tA).
\]
It is negative exactly when
\[
    t>\frac{B}{A}=\frac{1}{R}.
\]
This condition reverses the selected model after \(s_1\). Prediction~\ref{pred:full-switch} then implies a full switch. Below the threshold, the selections agree after \(s_1\), so no switch occurs.

At equality, subjective fits tie after both signals. The fixed tie-breaking rule selects the same model after both signals. Objective selections differ across signals, so the switch is partial.
\end{proof}

\subsubsection{Proof of Proposition~\ref{prop:belief-direction}}
\label{app:proof-belief-direction}

\begin{proof}
Fix \(s\in S\). I first derive Equation~\eqref{eq:bayesian-distortion}. Let
\[
    D(s)=\sum_{m\in M}\rho_0(m)P_m(s),
    \qquad
    \hat D(s)=\sum_{m\in M}\rho_0(m)\hat P_m(s).
\]
Equations~\eqref{eq:bayesian-weight}, \eqref{eq:bayesian-belief}, and \eqref{eq:model-specific-posterior} give
\[
    \mu_B(\omega\mid s)
    =
    \frac{
        \mu_0(\omega)
        \sum_{m\in M}\rho_0(m)\pi_m(s\mid\omega)
    }{D(s)}.
\]
Likewise, Equations~\eqref{eq:subjective-posterior} and \eqref{eq:distorted-weight} give
\[
    \hat\mu_B(\omega\mid s)
    =
    \frac{
        \mu_0(\omega)f(U(\omega))
        \sum_{m\in M}\rho_0(m)\pi_m(s\mid\omega)
    }{\hat D(s)}.
\]
Using the definition of \(\hat D(s)\) and rearranging the sums gives
\[
\begin{aligned}
    \hat D(s)
    &=
    \sum_{\omega'\in\Omega}
    \mu_0(\omega')f(U(\omega'))
    \sum_{m\in M}\rho_0(m)\pi_m(s\mid\omega')\\
    &=
    D(s)\sum_{\omega'\in\Omega}
    \mu_B(\omega'\mid s)f(U(\omega')).
\end{aligned}
\]
Substituting this equality into the expression for \(\hat\mu_B\) gives Equation~\eqref{eq:bayesian-distortion}.

Now write \(x_m=\mu_m(\omega_1\mid s)\) and \(t=f(U(\omega_1))/f(U(\omega_2))>1\). Full support implies \(x_m\in(0,1)\) and \(\mu_B(\omega_1\mid s)\in(0,1)\). Equation~\eqref{eq:bayesian-distortion} gives
\[
    \hat\mu_B(\omega_1\mid s)
    =h(t,\mu_B(\omega_1\mid s)),
    \qquad
    h(t,x)=\frac{tx}{tx+(1-x)}.
\]
For \(t>1\) and \(x\in(0,1)\),
\[
    h(t,x)-x
    =
    \frac{(t-1)x(1-x)}{1+(t-1)x}>0.
\]
Therefore, \(\hat\mu_B(\omega_1\mid s)>\mu_B(\omega_1\mid s)\).

Let \(j=m_s^*\) and \(k=\hat m_s^*\). Equation~\eqref{eq:subj-fit} becomes
\[
    \hat P_m(s)=P_m(s)f(U(\omega_2))\big[1+(t-1)x_m\big].
\]
Optimality gives \(P_j(s)\geq P_k(s)\) and \(\hat P_k(s)\geq\hat P_j(s)\). If \(x_k<x_j\), then
\[
\begin{aligned}
    \hat P_k(s)
    &=P_k(s)f(U(\omega_2))\big[1+(t-1)x_k\big]\\
    &\leq P_j(s)f(U(\omega_2))\big[1+(t-1)x_k\big]\\
    &<P_j(s)f(U(\omega_2))\big[1+(t-1)x_j\big]\\
    &=\hat P_j(s).
\end{aligned}
\]
This contradicts the selection of \(k\), so \(x_k\geq x_j\). Hence,
\[
    \mu_{\hat m_s^*}(\omega_1\mid s)
    \geq
    \mu_{m_s^*}(\omega_1\mid s).
\]

Finally,
\[
    \hat\mu_F(\omega_1\mid s)
    =h(t,x_k)
    >x_k
    \geq x_j
    =\mu_F(\omega_1\mid s).
\]
This proves both claims for best-fit agents.
\end{proof}

\subsubsection{Proof of Proposition~\ref{prop:stronger-distortion}}
\label{app:proof-stronger-distortion}

\begin{proof}
Fix \(s\in S\). Multiplying all values of \(f(U(\omega))\) by the same positive constant leaves posterior beliefs, model weights, and model selection unchanged. Set \(f(U(\omega_2))=1\) and \(f(U(\omega_1))=t\).

The function
\[
    h(t,x)=\frac{tx}{tx+(1-x)}
\]
is continuous. For \(t>0\) and \(x\in(0,1)\),
\[
    \frac{\partial h}{\partial t}
    =\frac{x(1-x)}{[tx+(1-x)]^2}>0,
    \qquad
    \frac{\partial h}{\partial x}
    =\frac{t}{[tx+(1-x)]^2}>0.
\]
Equation~\eqref{eq:bayesian-distortion} gives
\[
    \hat\mu_B(\omega_1\mid s)=h(t,\mu_B(\omega_1\mid s)).
\]
Full support places \(\mu_B(\omega_1\mid s)\) in \((0,1)\). This proves part 1.

For the best-fit rule, define \(x_m=\mu_m(\omega_1\mid s)\) and
\[
    q_m(t)=P_m(s)\big[1+(t-1)x_m\big].
\]
Equation~\eqref{eq:subj-fit} gives \(\hat P_m(s)=q_m(t)\), so the selected model maximizes \(q_m(t)\). Fix \(1<t_1<t_2\), and let models \(a\) and \(b\) be selected at these values.

Optimality gives \(q_a(t_1)\geq q_b(t_1)\) and \(q_b(t_2)\geq q_a(t_2)\). Combining these inequalities gives
\[
    \frac{1+(t_2-1)x_b}{1+(t_1-1)x_b}
    \geq
    \frac{1+(t_2-1)x_a}{1+(t_1-1)x_a}.
\]
For fixed \(t_1\) and \(t_2\), the ratio on each side is strictly increasing in \(x\), with derivative
\[
    \frac{t_2-t_1}{[1+(t_1-1)x]^2}>0.
\]
Therefore, \(x_b\geq x_a\).

The model-specific posterior probability on \(\omega_1\) under the selected model is therefore non-decreasing in \(t\). It remains constant whenever the selected model does not change. This proves part 2.

The distorted best-fit posterior belief at \(t_1\) is \(h(t_1,x_a)\), while the belief at \(t_2\) is \(h(t_2,x_b)\). Because \(x_b\geq x_a\),
\[
    h(t_2,x_b)
    \geq
    h(t_2,x_a)
    >
    h(t_1,x_a).
\]
The strict inequality follows because \(x_a\in(0,1)\). Thus, the distorted best-fit posterior belief is strictly increasing in \(t\).

Each \(q_m(t)\) is affine in \(t\), and two distinct affine functions intersect at most once. Since \(M\) is finite, only finitely many switching thresholds exist. Identical functions do not create switches because the fixed tie-breaking rule selects the same model for every \(t\).

Between switching thresholds, the selected model is fixed. The distorted best-fit posterior belief is therefore continuous on each such interval. A switch to a model with a strictly higher model-specific posterior probability on \(\omega_1\) produces an upward jump. This proves part 3.
\end{proof}

\newpage

\renewcommand{\thefigure}{\Alph{section}-\arabic{figure}}
\setcounter{figure}{0}
\renewcommand{\thetable}{\Alph{section}-\arabic{table}}
\setcounter{table}{0}

\section{Further Empirical Analysis}

\subsection{Sample Characteristics}\label{app:sample_characteristics}

Table~\ref{tab:demo} summarizes the characteristics of the 193 participants.

\begin{table}[htbp]
\centering
\begin{threeparttable}
\small
\begin{tabular}{@{}c@{}}
\begin{adjustbox}{max width=\linewidth}
{
\def\sym#1{\ifmmode^{#1}\else\(^{#1}\)\fi}
\begin{tabular}{l*{1}{ccccc}}
\toprule
 & count & mean & sd & min & max \\
\midrule
Female & 193 & 0.497 & 0.501 & 0 & 1 \\
Asian & 193 & 0.415 & 0.494 & 0 & 1 \\
White & 193 & 0.472 & 0.500 & 0 & 1 \\
Cognitive Score & 193 & 2.751 & 1.283 & 0 & 6 \\
Prior Experiments & 193 & 3.710 & 3.072 & 0 & 17 \\
Quiz Score & 193 & 9.508 & 0.758 & 7 & 10 \\
\bottomrule
\end{tabular}
}

\end{adjustbox}
\end{tabular}
\begin{tablenotes}
\footnotesize
\item \textit{Notes:} The table reports summary statistics for the 193 participants. Means of indicator variables represent participant proportions. Quiz scores measure comprehension of the experimental instructions.
\end{tablenotes}
\end{threeparttable}
\caption{Sample Characteristics}
\label{tab:demo}
\end{table}

\clearpage

\subsection{Structural Classification}\label{app:structural_details}

This subsection presents the error structure and estimation procedure for the classification in Section~\ref{sec:structural_classification}.
The sample contains 5,211 reported beliefs from 193 participants, with 27 beliefs per participant in the \textit{SP} block.
Each participant follows one behavioral rule across these scenarios.
Let $k\in\mathcal T=\{\text{Bayesian},\text{Best-fit},\text{Non-mover}\}$ index behavioral types, and let $t$ index rounds.

\paragraph*{First layer: errors in model-fit calculations.}
For each decision problem, the objective fit of model $m$ after signal $s$ is
\begin{equation}
P_m(s)=\sum_{\omega\in\Omega}\mu_0(\omega)\pi_m(s\mid\omega).
\label{eq:obj_fit}
\end{equation}
Participant $i$ of type $k$ assesses this fit with a mean-zero normal error:
\begin{equation}
\hat P_{mit}(s)=P_m(s)+\varepsilon^1_{mit},
\qquad \varepsilon^1_{mit}\sim N(0,\sigma_{1,k}^2).
\label{eq:fit_error}
\end{equation}
The errors are independent across models, participants, and rounds.
Their standard deviation is common to the Bayesian and Best-fit types, $\sigma_{1,k}=\sigma_1$.
The model prior assigns equal probabilities to carts A and B.
Bayesian agents therefore weight the model-specific posteriors by their relative subjective fits.
For positive perceived fits, the weight on model A is\footnote{In structural estimation, subjective fits below $10^{-10}$ are set to $10^{-10}$, and model weights are therefore restricted to \([0,1]\).
The main results are robust when I modify the lower bound.}
\begin{equation}
\rho_{it}(A\mid s)
=\frac{\rho_0(A)\hat P_{Ait}(s)}
{\rho_0(A)\hat P_{Ait}(s)+\rho_0(B)\hat P_{Bit}(s)},
\qquad
\rho_{it}(B\mid s)=1-\rho_{it}(A\mid s).
\label{eq:bayes_weight}
\end{equation}
In the experiment, $\rho_0(A)=\rho_0(B)=1/2$, so the model-prior terms cancel. Bounding the weight between zero and one keeps the weighted belief between the two model-specific posteriors when fit errors are large.
Best-fit agents select model A when $\hat P_{Ait}(s)>\hat P_{Bit}(s)$ and model B otherwise.
Let $\hat m_{it}^*$ denote this selected model.
Non-movers retain the state prior and do not use model-fit calculations.

\paragraph*{Second layer: reporting errors.}
Let $\omega_1$ denote selection of Urn X and $\omega_2$ selection of Urn Y.
The internal belief about Urn X for a participant of type $k$ is $\mu_{k,it}(\omega_1\mid s)$.
The three rules imply
\begin{align}
\mu_{\text{Bayesian},it}(\omega_1\mid s)
&=\sum_{m\in\{A,B\}}\rho_{it}(m\mid s)\mu_m(\omega_1\mid s),
\label{eq:bayes_belief}\\
\mu_{\text{Best-fit},it}(\omega_1\mid s)
&=\mu_{\hat m_{it}^*}(\omega_1\mid s),
\label{eq:bestfit_belief}\\
\mu_{\text{Non-mover},it}(\omega_1\mid s)
&=\mu_0(\omega_1).
\label{eq:nonmover_belief}
\end{align}
A second layer of mean-zero normal error is added to the internal belief before censoring at zero and one:
\begin{equation}
\text{Belief}_{it}=\mu_{k,it}(\omega_1\mid s)+\varepsilon^2_{it},
\qquad \varepsilon^2_{it}\sim N(0,\sigma_{2,k}^2).
\label{eq:reporting_error}
\end{equation}
Reporting errors are independent across participants and rounds and independent of the first-layer errors.
The reporting-error standard deviation is common to all three types, $\sigma_{2,k}=\sigma_2$.
The Non-mover type has no fit-calculation error.

\paragraph*{Likelihood and estimation.}
Let $\phi$ denote the standard normal density.
Conditional on the first-layer errors, the density of a reported belief $b \in (0,1)$ is
\begin{equation}
\frac{1}{\sigma_{2,k}}\phi\left(\frac{b-\mu_{k,it}(\omega_1\mid s)}{\sigma_{2,k}}\right).
\label{eq:report_density_app}
\end{equation}
Reported beliefs at zero and one contribute the corresponding lower and upper normal tail probabilities to the likelihood.
For Bayesian and Best-fit types, I integrate the conditional
likelihood contribution over the two first-layer errors.
For Non-movers, I evaluate the same likelihood contribution with
$\mu_{k,it}(\omega_1\mid s)=\mu_0(\omega_1)$
and reporting-error standard deviation $\sigma_2$.
Let $\ell_{it}(k;\boldsymbol\sigma_k)$ denote the resulting report likelihood.\footnote{In structural estimation and participant classification, for numerical stability, report likelihoods below $10^{-10}$ are set to $10^{-10}$ before taking logarithms. The main results are robust to modifying this lower bound.}
Conditional on a participant's type, I assume that reported beliefs are independent across rounds. The participant likelihood is therefore the product of the report likelihoods:
\begin{equation}
L_i(k;\boldsymbol\sigma_k)
=\prod_{t\in\mathcal T_i^{SP}}\ell_{it}(k;\boldsymbol\sigma_k),
\label{eq:participant_like_app}
\end{equation}
where $\mathcal T_i^{SP}$ contains participant $i$'s 27 rounds in the \textit{SP} block.
With population shares $p_k\geq0$ and $\sum_kp_k=1$, the sample likelihood is
\begin{equation}
L(\boldsymbol\theta)=\prod_{i\in\mathcal I}\sum_{k\in\mathcal T}
p_kL_i(k;\boldsymbol\sigma_k).
\label{eq:mixture_like_app}
\end{equation}
The parameter vector contains two independent population shares and two error standard deviations.
Standard errors are calculated from 500 bootstrap samples drawn by resampling participants.
Table~\ref{tab:parameter_estimates} in Section~\ref{sec:structural_classification} reports the parameter estimates and standard errors.

\paragraph*{Participant classification.}
Using estimated population shares as prior type probabilities, participant $i$'s posterior probability of type $k$ is
\begin{equation}
\Pr(k\mid\{\text{Belief}_{it}\}_{t\in\mathcal T_i^{SP}})
=\frac{\hat p_kL_i(k;\hat{\boldsymbol\sigma}_k)}
{\sum_{j\in\mathcal T}\hat p_jL_i(j;\hat{\boldsymbol\sigma}_j)}.
\label{eq:posterior_type_app}
\end{equation}
I assign each participant to the type with the highest posterior probability.
This procedure classifies 132 participants as Bayesian, 33 as Best-fit, and 28 as Non-movers.
Table~\ref{tab:individual_estimates} in Section~\ref{sec:structural_classification} reports how many participants remain classified when the highest posterior probability must exceed 50\%, 70\%, or 90\%.


\clearpage

\subsection{Effects of Motivated Reasoning: Robustness Checks}
\label{app:core_results}
This subsection reports additional estimates of the effect of motivated reasoning. The regressions follow Equation~\eqref{eqn:ols} unless stated otherwise, and standard errors are clustered by participant. Regressions with participant fixed effects also include decision-problem-by-signal and round fixed effects.

Table~\ref{tab:core_robustness} reports the pooled effect of asymmetric payoffs under three specifications. Column (1) repeats the pooled estimate in Figure~\ref{fig:mr_by_type}. Column (2) replaces the participant covariates and decision-problem fixed effects with participant and decision-problem-by-signal fixed effects. Column (3) uses only beliefs from each participant's first block, so each participant contributes beliefs under only one payoff condition, and participant fixed effects cannot be included. The estimated effects are 1.7, 1.6, and 1.3 percentage points, respectively. The first-block estimate is less precise and not statistically significant, but it is similar in size.

\begin{table}[htbp]
\centering
\begin{threeparttable}
\small
\begin{tabular}{@{}c@{}}
\begin{adjustbox}{max width=\linewidth}
{
\def\sym#1{\ifmmode^{#1}\else\(^{#1}\)\fi}
\begin{tabular}{l*{3}{c}}
\toprule
DV: Reported Belief & (1) & (2) & (3) \\
\midrule
\(\mathbbm{1}_{\text{Asym}}\) & 0.017\sym{***} & 0.016\sym{***} & 0.013 \\
 & (0.005) & (0.005) & (0.009) \\
\midrule
Observations & 10422 & 10422 & 5211 \\
Participants & 193 & 193 & 193 \\
Mean of DV & 0.506 & 0.506 & 0.503 \\
SP comparison mean & 0.498 & 0.498 & 0.496 \\
\(R^2\) & 0.336 & 0.511 & 0.424 \\
Participant FE & No & Yes & No \\
DP \(\times\) signal FE & No & Yes & Yes \\
DP FE & Yes & No & No \\
Sample & All reports & All reports & First block \\
Participant covariates & Yes & No & No \\
\bottomrule
\end{tabular}
}

\end{adjustbox}
\end{tabular}
\begin{tablenotes}
\footnotesize
\item \textit{Notes:} The table reports OLS estimates. The dependent variable is reported belief. Columns (1) and (2) use all beliefs; column (3) uses beliefs from the first block. Unreported controls in column (1) include signal color, gender, race, cognitive score, risk attitude, overconfidence, and narrative extremeness. Column (1) also includes unreported round and decision-problem fixed effects. Column (2) includes unreported participant, decision-problem-by-signal, and round fixed effects. Column (3) includes unreported decision-problem-by-signal and round fixed effects. The \textit{SP} comparison mean uses beliefs in the \textit{SP} block in the same estimation sample. Standard errors clustered at the participant level appear in parentheses. $^{*}p<0.10$, $^{**}p<0.05$, $^{***}p<0.01$.
\end{tablenotes}
\end{threeparttable}
\caption{Overall Effects and First-Block Comparison}
\label{tab:core_robustness}
\end{table}

\clearpage

Table~\ref{tab:core_types_full} repeats the estimates by type in Figure~\ref{fig:mr_by_type}, replacing the participant covariates with participant fixed effects. As in Figure~\ref{fig:mr_by_type}, types are assigned as in Section~\ref{sec:structural_classification}. The estimated effect is 7.9 percentage points among Best-fit participants ($p<0.01$), compared with 0.5 percentage points among Bayesian participants and $-0.2$ percentage points among Non-movers. These estimates are close to those in Figure~\ref{fig:mr_by_type}.

\begin{table}[htbp]
\centering
\begin{threeparttable}
\small
\begin{tabular}{@{}c@{}}
\begin{adjustbox}{max width=\linewidth}
{
\def\sym#1{\ifmmode^{#1}\else\(^{#1}\)\fi}
\begin{tabular}{l*{4}{c}}
\toprule
DV: Reported Belief & (1) & (2) & (3) & (4) \\
\midrule
\(\mathbbm{1}_{\text{Asym}}\) & 0.016\sym{***} & 0.005 & 0.079\sym{***} & -0.002 \\
 & (0.005) & (0.003) & (0.023) & (0.009) \\
\midrule
Observations & 10422 & 7128 & 1782 & 1512 \\
Participants & 193 & 132 & 33 & 28 \\
Mean of DV & 0.506 & 0.513 & 0.476 & 0.511 \\
\(R^2\) & 0.511 & 0.627 & 0.372 & 0.579 \\
Sample & All participants & Bayesian & Best-fit & Non-mover \\
\bottomrule
\end{tabular}
}

\end{adjustbox}
\end{tabular}
\begin{tablenotes}
\footnotesize
\item \textit{Notes:} The table reports OLS estimates. The dependent variable is reported belief. Columns (1)--(4) use all participants, Bayesian participants, Best-fit participants, and Non-movers, respectively. The specification includes participant, decision-problem-by-signal, and round fixed effects. The Non-mover regression contains 28 participants. Standard errors clustered at the participant level appear in parentheses. $^{*}p<0.10$, $^{**}p<0.05$, $^{***}p<0.01$.
\end{tablenotes}
\end{threeparttable}
\caption{Motivated Reasoning Effects with Participant Fixed Effects}
\label{tab:core_types_full}
\end{table}

\clearpage

Table~\ref{tab:probability_three} extends Table~\ref{tab:mr_likelihood}, which interacts $\mathbbm{1}_{\text{Asym}}$ with each participant's posterior type probabilities. Column (2) repeats the specification in Table~\ref{tab:mr_likelihood}. Column (1) omits the participant covariates, and column (3) adds participant fixed effects, which absorb the posterior type probabilities. The coefficient on the interaction between $\mathbbm{1}_{\text{Asym}}$ and the Best-fit probability ranges from 0.077 to 0.082 across the three columns ($p<0.01$ in each).

\begin{table}[htbp]
\centering
\begin{threeparttable}
\small
\begin{tabular}{@{}c@{}}
\begin{adjustbox}{max width=\linewidth}
{
\def\sym#1{\ifmmode^{#1}\else\(^{#1}\)\fi}
\begin{tabular}{l*{3}{c}}
\toprule
DV: Reported Belief & (1) & (2) & (3) \\
\midrule
\(\mathbbm{1}_{\text{Asym}}\) & 0.004 & 0.004 & 0.004 \\
 & (0.004) & (0.004) & (0.003) \\
\(\mathbbm{1}_{\text{Asym}}\) \(\times\) Likelihood of Best-fit Type & 0.082\sym{***} & 0.082\sym{***} & 0.077\sym{***} \\
 & (0.022) & (0.023) & (0.023) \\
\(\mathbbm{1}_{\text{Asym}}\) \(\times\) Likelihood of Non-mover Type & -0.007 & -0.007 & -0.014 \\
 & (0.011) & (0.011) & (0.011) \\
Likelihood of Best-fit Type & -0.076\sym{***} & -0.069\sym{***} &  \\
 & (0.019) & (0.018) &  \\
Likelihood of Non-mover Type & -0.001 & 0.004 &  \\
 & (0.007) & (0.008) &  \\
\midrule
Observations & 10422 & 10422 & 10422 \\
Participants & 193 & 193 & 193 \\
Mean of DV & 0.506 & 0.506 & 0.506 \\
\(R^2\) & 0.340 & 0.343 & 0.515 \\
Participant FE & No & No & Yes \\
DP \(\times\) signal FE & No & No & Yes \\
DP FE & Yes & Yes & No \\
Participant covariates & No & Yes & Absorbed \\
Narrative extremeness & No & Yes & Absorbed \\
\bottomrule
\end{tabular}
}

\end{adjustbox}
\end{tabular}
\begin{tablenotes}
\footnotesize
\item \textit{Notes:} The table reports OLS estimates of how the effect of motivated reasoning varies with posterior type probabilities. The dependent variable is reported belief. Each column uses 10,422 beliefs from 193 participants. Posterior type probabilities are estimated from the 27 beliefs per participant in the \textit{SP} block. The Bayesian type is the reference category. Each regression includes $\mathbbm{1}_{\text{Asym}}$, the other types' posterior probabilities, and their interactions with $\mathbbm{1}_{\text{Asym}}$. Columns (1) and (2) include fixed effects for signal color, round, and decision problem. Column (2) also controls for gender, race indicators, cognitive score, risk attitude, overconfidence, and narrative extremeness. Column (3) includes participant, decision-problem-by-signal, and round fixed effects. Participant fixed effects absorb the posterior type probabilities in column (3). Standard errors clustered at the participant level appear in parentheses. $^{*}p<0.10$, $^{**}p<0.05$, $^{***}p<0.01$.
\end{tablenotes}
\end{threeparttable}
\caption{Posterior-Probability Interactions: Three Types}
\label{tab:probability_three}
\end{table}

\clearpage

The estimates above rely on the classification from the finite mixture model. As a robustness check that does not use this model, I construct a simple measure of how close each participant's reports are to Best-fit predictions. For each report in the \textit{SP} block, I calculate its absolute distance from the Bayesian benchmark minus its absolute distance from the Best-fit benchmark. For each participant and decision problem, the closeness to Best-fit predictions is the average of this difference over the participant's \textit{SP} reports in the other eight decision problems. Excluding the current decision problem ensures that the measure does not use reports from the decision problem in which the effect is estimated. Higher values mean that a participant's reports are closer to Best-fit predictions relative to Bayesian predictions. Most participants have negative values, and the median participant average is $-0.11$.

Table~\ref{tab:holdout_demographic} adds this measure and its interaction with $\mathbbm{1}_{\text{Asym}}$ to Equation~\eqref{eqn:ols}. The interaction coefficient is 0.544 ($p<0.01$). At the lower and upper quartiles of participant averages, $-0.15$ and $-0.06$, the implied effect of asymmetric payoffs is $-1.3$ and $3.5$ percentage points, respectively. Participants whose reports are closer to Best-fit predictions therefore respond more strongly to asymmetric payoffs, consistent with the estimates by type.

\begin{table}[htbp]
\centering
\begin{threeparttable}
\small
\begin{tabular}{@{}c@{}}
\begin{adjustbox}{max width=\linewidth}
{
\def\sym#1{\ifmmode^{#1}\else\(^{#1}\)\fi}
\begin{tabular}{l*{1}{c}}
\toprule
DV: Reported Belief & (1) \\
\midrule
\(\mathbbm{1}_{\text{Asym}}\) & 0.068\sym{***} \\
 & (0.011) \\
Closeness to Best-fit predictions & -0.406\sym{***} \\
 & (0.067) \\
\(\mathbbm{1}_{\text{Asym}}\) \(\times\) Closeness & 0.544\sym{***} \\
 & (0.090) \\
\midrule
Observations & 10422 \\
Participants & 193 \\
Mean of DV & 0.506 \\
\(R^2\) & 0.345 \\
\bottomrule
\end{tabular}
}

\end{adjustbox}
\end{tabular}
\begin{tablenotes}
\footnotesize
\item \textit{Notes:} The table reports OLS estimates. The dependent variable is reported belief. The sample contains 10,422 beliefs from 193 participants. Closeness to Best-fit predictions averages absolute Bayesian error minus absolute Best-fit error across beliefs in the \textit{SP} block outside the current decision problem. The regression includes $\mathbbm{1}_{\text{Asym}}$, closeness to Best-fit predictions, and their interaction. Unreported controls include gender, race indicators, cognitive score, risk attitude, overconfidence, and narrative extremeness. The specification includes fixed effects for signal color, round, and decision problem. Standard errors clustered at the participant level appear in parentheses. $^{*}p<0.10$, $^{**}p<0.05$, $^{***}p<0.01$.
\end{tablenotes}
\end{threeparttable}
\caption{Motivated Reasoning Effects and Closeness to Best-fit Predictions}
\label{tab:holdout_demographic}
\end{table}

\clearpage

\paragraph*{Decision Problems 1 and 4.}
Decision Problem 4 reproduces Decision Problem 1 after exchanging the state labels and model labels. The two problems therefore share the state prior, objective fit gap, and separation between model-specific posteriors. Their objectively best-fitting models favor Urn Y and Urn X, respectively. Figure~\ref{fig:mr_dp4} displays the effects in each problem and their difference. For Best-fit participants, the estimated effect of motivated reasoning is approximately 13.2 percentage points in Decision Problem 1 and 2.7 percentage points in Decision Problem 4. Their difference is 10.5 percentage points, with a two-sided $p$-value of 0.101. The pooled difference is 2.1 percentage points and has $p=0.126$.

\begin{figure}[htbp]
\centering
\includegraphics[width=0.9\linewidth]{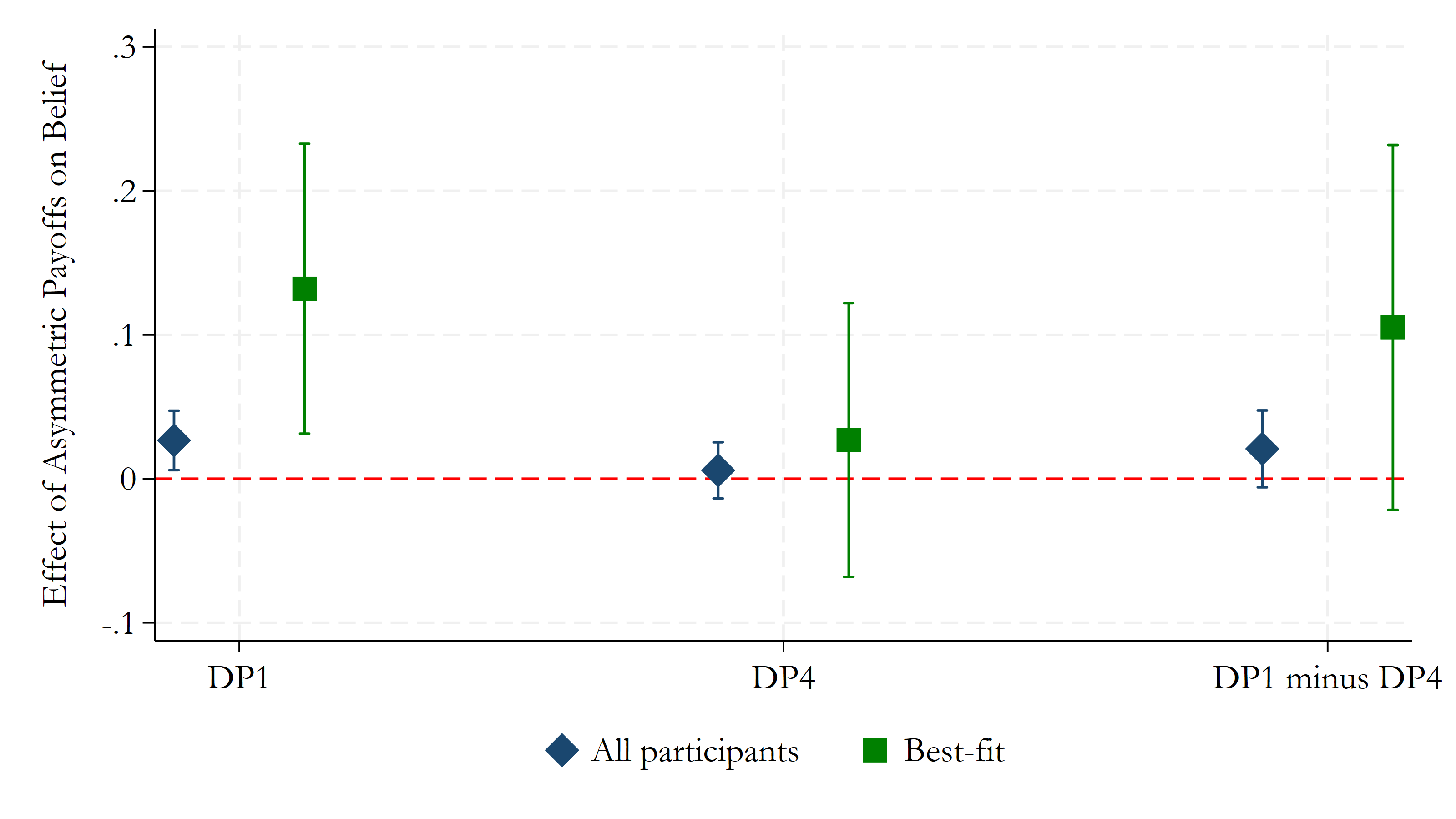}
\begin{minipage}{0.9\linewidth}\footnotesize
\textit{Notes:} The figure reports estimated effects of motivated reasoning using beliefs from Decision Problems 1 and 4. Navy diamonds show pooled estimates and green squares show Best-fit estimates. The samples contain 2,316 beliefs from 193 participants and 396 beliefs from 33 participants, respectively. Regressions include participant, decision-problem-by-signal, and round fixed effects and interact $\mathbbm{1}_{\text{Asym}}$ with an indicator for Decision Problem 1. The first two positions show the implied effects of motivated reasoning; the third shows their directly estimated difference. Bars indicate 95\% confidence intervals based on standard errors clustered at the participant level.
\end{minipage}
\caption{Motivated Reasoning Effects in Decision Problems 1 and 4}
\label{fig:mr_dp4}
\end{figure}

\clearpage

Table~\ref{tab:dp14_robustness} reports the estimates plotted in Figure~\ref{fig:mr_dp4}. For Best-fit participants, the difference between Decision Problems 1 and 4 is 10.5 percentage points with a standard error of 6.2 percentage points. The estimates have the predicted sign, but the difference is imprecisely estimated.

\begin{table}[htbp]
\centering
\begin{threeparttable}
\small
\begin{tabular}{@{}c@{}}
\begin{adjustbox}{max width=\linewidth}
{
\def\sym#1{\ifmmode^{#1}\else\(^{#1}\)\fi}
\begin{tabular}{l*{2}{c}}
\toprule
DV: Reported Belief & (1) & (2) \\
\midrule
\(\mathbbm{1}_{\text{Asym}}\) (DP4) & 0.006 & 0.027 \\
 & (0.010) & (0.047) \\
DP1 minus DP4 effect & 0.021 & 0.105 \\
 & (0.014) & (0.062) \\
\midrule
Observations & 2316 & 396 \\
Participants & 193 & 33 \\
Mean of DV & 0.518 & 0.525 \\
\(R^2\) & 0.228 & 0.378 \\
Sample & All participants & Best-fit \\
\bottomrule
\end{tabular}
}

\end{adjustbox}
\end{tabular}
\begin{tablenotes}
\footnotesize
\item \textit{Notes:} The table reports OLS estimates. The dependent variable is reported belief. The sample contains only Decision Problems 1 and 4. Columns (1) and (2) use all participants and Best-fit participants, respectively. The first row reports the effect of motivated reasoning in Decision Problem 4. The second reports the difference between Decision Problems 1 and 4. The specification includes participant, decision-problem-by-signal, and round fixed effects. Standard errors clustered at the participant level appear in parentheses. $^{*}p<0.10$, $^{**}p<0.05$, $^{***}p<0.01$.
\end{tablenotes}
\end{threeparttable}
\caption{Motivated Reasoning Effects in Decision Problems 1 and 4: Regression Estimates}
\label{tab:dp14_robustness}
\end{table}

\clearpage

\subsection{Type-Specific Error Standard Deviations}
\label{app:type_specific_errors}
This subsection reports the three-type classification when the error standard deviations in Equations~\eqref{eq:fit_error} and \eqref{eq:reporting_error} differ across types.
The fit-calculation standard deviation $\sigma_{1,k}$ is estimated separately for the Bayesian and Best-fit types, and the reporting standard deviation $\sigma_{2,k}$ is estimated separately for each of the three types.
The parameter vector therefore contains two independent population shares and five error standard deviations.
The behavioral rules, participant likelihood, and posterior assignment otherwise follow Appendix~\ref{app:structural_details}.\footnote{With the reporting error, the belief implied by the rule and the error can fall below zero or above one, but participants can only report beliefs between zero and one. Therefore, the likelihood uses normal densities for reports strictly between zero and one. For reports at zero and one, it uses the normal probabilities below zero and above one, respectively.}
Table~\ref{tab:ts_parameter_estimates} reports the estimates.

This specification classifies 123 participants as Bayesian, 63 as Best-fit, and seven as Non-movers, compared with 132, 33, and 28 in the main specification (Table~\ref{tab:individual_estimates}).
Table~\ref{tab:xtab_3_ts} compares these assignments with the main classification. All 33 participants classified as Best-fit in the main specification remain Best-fit. The other 30 Best-fit participants are classified as Bayesian (26) or Non-movers (4) in the main specification.
Figure~\ref{fig:ts_three_all} reports the effects of motivated reasoning by type. With participant covariates, the estimated effect is 4.9 percentage points among Best-fit participants ($p<0.001$) and 0.1 percentage points among Bayesian participants ($p=0.693$), compared with 8.4 and 0.5 percentage points in the main specification (Figure~\ref{fig:mr_by_type}).

\vspace{0.25cm}

\begin{table}[htbp]
\centering
\begin{threeparttable}
\small
\begin{tabular}{@{}c@{}}
\begin{adjustbox}{max width=\linewidth}
\begin{tabular}{lcc}
\toprule
\textbf{Parameter} & \textbf{Estimate} & \textbf{Bootstrap SE} \\
\hline
\multicolumn{1}{l}{\textit{Population Shares}} \\
\(p_1\) (Bayesian) & 0.641 & 0.059 \\
\(p_2\) (Best-fit) & 0.322 & 0.059 \\
\(p_3\) (Non-mover) & 0.036 & 0.013 \\
\addlinespace
\multicolumn{1}{l}{\textit{First-layer (fit-calculation) error SD}} \\
\(\sigma_{1,\text{Bayes}}\) & 0.004 & 0.001 \\
\(\sigma_{1,\text{Best-fit}}\) & 0.278 & 0.067 \\
\(\sigma_{1,\text{Non-mover}}\) & \multicolumn{2}{c}{---} \\
\addlinespace
\multicolumn{1}{l}{\textit{Second-layer (reporting) error SD}} \\
\(\sigma_{2,\text{Bayes}}\) & 0.103 & 0.010 \\
\(\sigma_{2,\text{Best-fit}}\) & 0.185 & 0.008 \\
\(\sigma_{2,\text{Non-mover}}\) & 0.001 & 0.000 \\
\bottomrule
\end{tabular}

\end{adjustbox}
\end{tabular}
\begin{tablenotes}
\footnotesize
\item \textit{Notes:} The table reports maximum likelihood estimates of population shares and type-specific error standard deviations. The sample contains 5,211 beliefs from 193 participants in the \textit{SP} block, with 27 beliefs per participant. Population shares are the estimated fractions of each type in the population from which participants are drawn. Standard errors are calculated from 500 bootstrap samples drawn by resampling participants.
\end{tablenotes}
\end{threeparttable}
\caption{Type-Specific Errors: Three-Type Parameter Estimates}
\label{tab:ts_parameter_estimates}
\end{table}

\clearpage

\begin{table}[htbp]
\centering
\begin{threeparttable}
\small
\begin{tabular}{@{}c@{}}
\begin{adjustbox}{max width=\linewidth}
{
\def\sym#1{\ifmmode^{#1}\else\(^{#1}\)\fi}
\begin{tabular}{l*{4}{c}}
\toprule
 & Bayesian & Best-fit & Non-mover & Total \\
\midrule
Bayesian & 106 & 26 & 0 & 132 \\
Best-fit & 0 & 33 & 0 & 33 \\
Non-mover & 17 & 4 & 7 & 28 \\
Total & 123 & 63 & 7 & 193 \\
\bottomrule
\end{tabular}
}

\end{adjustbox}
\end{tabular}
\begin{tablenotes}
\footnotesize
\item \textit{Notes:} The table compares participant classifications under the main three-type specification and the three-type specification with type-specific error standard deviations. Rows report the main classification, and columns report the alternative classification. Each entry gives the number of participants. Both specifications assign each of the 193 participants to the type with the highest posterior probability, using beliefs from the \textit{SP} block.
\end{tablenotes}
\end{threeparttable}
\caption{Type-Specific Errors and Main Three-Type Assignments}
\label{tab:xtab_3_ts}
\end{table}

\begin{figure}[htbp]
\centering
\includegraphics[width=0.88\linewidth]{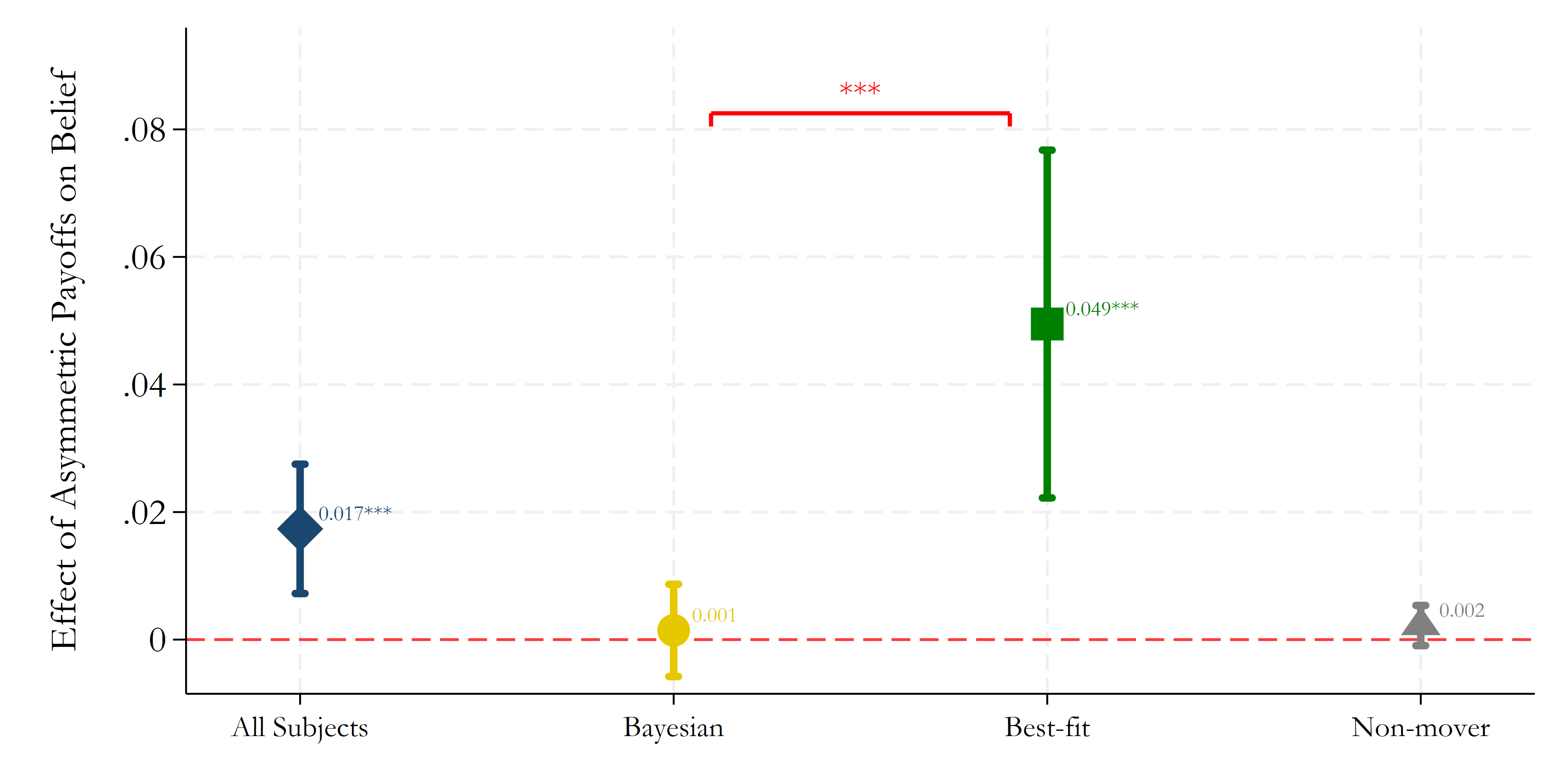}
\begin{minipage}{0.88\linewidth}
\footnotesize \textit{Notes:} The figure reports estimated coefficients on $\mathbbm{1}_{\text{Asym}}$ in Equation~\eqref{eqn:ols} for the pooled sample and separately for each behavioral type, using the classification with type-specific error standard deviations. The pooled, Bayesian, Best-fit, and Non-mover samples contain 193, 123, 63, and 7 participants, respectively. Each participant contributes 54 beliefs. Regressions control for gender, race, cognitive score, risk attitude, overconfidence, and narrative extremeness. All specifications include fixed effects for signal color, round, and decision problem. Bars indicate 95\% confidence intervals based on standard errors clustered at the participant level. The red bracket reports the significance of the difference between the Best-fit and Bayesian coefficients. $^{*}p<0.10$, $^{**}p<0.05$, $^{***}p<0.01$.
\end{minipage}
\caption{Type-Specific Errors: Motivated Reasoning Effects, All Types}
\label{fig:ts_three_all}
\end{figure}

\begin{table}[htbp]
\centering
\begin{threeparttable}
\small
\begin{tabular}{@{}c@{}}
\begin{adjustbox}{max width=\linewidth}
{
\def\sym#1{\ifmmode^{#1}\else\(^{#1}\)\fi}
\begin{tabular}{l*{4}{c}}
\toprule
DV: Reported Belief & (1) & (2) & (3) & (4) \\
\midrule
\(\mathbbm{1}_{\text{Asym}}\) & 0.017\sym{***} & 0.001 & 0.049\sym{***} & 0.002 \\
 & (0.005) & (0.004) & (0.014) & (0.001) \\
\midrule
Observations & 10422 & 6642 & 3402 & 378 \\
Participants & 193 & 123 & 63 & 7 \\
Mean of DV & 0.506 & 0.506 & 0.507 & 0.502 \\
\(R^2\) & 0.336 & 0.458 & 0.261 & 0.978 \\
Sample & All participants & Bayesian & Best-fit & Non-mover \\
\bottomrule
\end{tabular}
}

\end{adjustbox}
\end{tabular}
\begin{tablenotes}
\footnotesize
\item \textit{Notes:} The table reports OLS estimates. The dependent variable is reported belief. Columns (1)--(4) use the pooled sample, Bayesian participants, Best-fit participants, and Non-movers, respectively. Each participant contributes 54 beliefs. Unreported controls include gender, race indicators, cognitive score, risk attitude, overconfidence, and narrative extremeness. The specification includes fixed effects for signal color, round, and decision problem. Standard errors clustered at the participant level appear in parentheses. $^{*}p<0.10$, $^{**}p<0.05$, $^{***}p<0.01$.
\end{tablenotes}
\end{threeparttable}
\caption{Type-Specific Errors: Motivated Reasoning Regressions, Participant Covariates}
\label{tab:ts_three_source}
\end{table}

\begin{table}[htbp]
\centering
\begin{threeparttable}
\small
\begin{tabular}{@{}c@{}}
\begin{adjustbox}{max width=\linewidth}
{
\def\sym#1{\ifmmode^{#1}\else\(^{#1}\)\fi}
\begin{tabular}{l*{3}{c}}
\toprule
DV: Reported Belief & (1) & (2) & (3) \\
\midrule
\(\mathbbm{1}_{\text{Asym}}\) & 0.001 & 0.001 & 0.001 \\
 & (0.004) & (0.004) & (0.003) \\
\(\mathbbm{1}_{\text{Asym}}\) \(\times\) Likelihood of Best-fit Type & 0.050\sym{***} & 0.050\sym{***} & 0.045\sym{***} \\
 & (0.015) & (0.015) & (0.014) \\
\(\mathbbm{1}_{\text{Asym}}\) \(\times\) Likelihood of Non-mover Type & 0.005 & 0.005 & 0.001 \\
 & (0.005) & (0.005) & (0.010) \\
Likelihood of Best-fit Type & -0.026\sym{**} & -0.024\sym{**} &  \\
 & (0.013) & (0.012) &  \\
Likelihood of Non-mover Type & -0.008\sym{**} & -0.006 &  \\
 & (0.004) & (0.008) &  \\
\midrule
Observations & 10422 & 10422 & 10422 \\
Participants & 193 & 193 & 193 \\
Mean of DV & 0.506 & 0.506 & 0.506 \\
\(R^2\) & 0.335 & 0.339 & 0.513 \\
Participant FE & No & No & Yes \\
DP \(\times\) signal FE & No & No & Yes \\
DP FE & Yes & Yes & No \\
Participant covariates & No & Yes & Absorbed \\
Narrative extremeness & No & Yes & Absorbed \\
\bottomrule
\end{tabular}
}

\end{adjustbox}
\end{tabular}
\begin{tablenotes}
\footnotesize
\item \textit{Notes:} The table reports OLS estimates of how the effect of motivated reasoning varies with posterior type probabilities. The dependent variable is reported belief. Each column uses 10,422 beliefs from 193 participants. Posterior type probabilities are estimated from the 27 beliefs per participant in the \textit{SP} block. The Bayesian type is the reference category. Each regression includes $\mathbbm{1}_{\text{Asym}}$, the other types' posterior probabilities, and their interactions with $\mathbbm{1}_{\text{Asym}}$. Columns (1) and (2) include fixed effects for signal color, round, and decision problem. Column (2) also controls for gender, race indicators, cognitive score, risk attitude, overconfidence, and narrative extremeness. Column (3) includes participant, decision-problem-by-signal, and round fixed effects. Participant fixed effects absorb the posterior type probabilities in column (3). Standard errors clustered at the participant level appear in parentheses. $^{*}p<0.10$, $^{**}p<0.05$, $^{***}p<0.01$.
\end{tablenotes}
\end{threeparttable}
\caption{Type-Specific Errors: Posterior-Probability Interactions}
\label{tab:ts_three_probability}
\end{table}

With participant covariates, the effect of motivated reasoning is 4.8 percentage points larger among Best-fit participants than among Bayesian participants ($p<0.001$). A test that the effect of motivated reasoning is equal across all types gives $p=0.003$. Both tests come from a single regression that interacts each control and fixed effect with type.

\clearpage

\subsection{Alternative Mixture Specifications}
\label{app:robustness_types}

This subsection examines how the results change under alternative sets of behavioral types.
Each specification estimates population shares, the two error standard deviations shared across types, and the parameters of any Grether components using the 27 beliefs per participant in the \textit{SP} block.
The Gaussian error structure, participant likelihood, and posterior assignment follow Appendix~\ref{app:structural_details}.
Each additional Grether component has two parameters, $\beta_{1,k}$ and $\beta_{2,k}$, common to participants in that component, and shares the two error standard deviations with the other types.

\paragraph*{Grether components.}
Following \citet{grether1980bayes}, I allow prior odds and likelihood ratios to receive different weights when updating beliefs within each model. The following equations define the Grether components.
For each model, the transformation applies to its Bayesian posterior $\mu_m(\omega_1\mid s)$:
\begin{equation}
\mu^G_{m,k}(\omega_1\mid s)
=\frac{\mu_0(\omega_1)^{\beta_{1,k}}\mu_m(\omega_1\mid s)^{\beta_{2,k}}}
{\mu_0(\omega_1)^{\beta_{1,k}}\mu_m(\omega_1\mid s)^{\beta_{2,k}}
+\mu_0(\omega_2)^{\beta_{1,k}}\mu_m(\omega_2\mid s)^{\beta_{2,k}}}.
\label{eq:grether_posterior_app}
\end{equation}
Here, $\beta_{1,k}$ weights the state prior, and $\beta_{2,k}$ weights the model-specific posterior, which already incorporates the signal and prior.
In log odds, this transformation is equivalent to
\begin{equation}
\log\frac{\mu^G_{m,k}(\omega_1\mid s)}{\mu^G_{m,k}(\omega_2\mid s)}
=(\beta_{1,k}+\beta_{2,k})\log\frac{\mu_0(\omega_1)}{\mu_0(\omega_2)}
+\beta_{2,k}\log\frac{\pi_m(s\mid\omega_1)}{\pi_m(s\mid\omega_2)}.
\label{eq:grether_logodds_app}
\end{equation}
Thus, the effective coefficient on prior log odds is $\beta_{1,k}+\beta_{2,k}$.

The Grether-Bayesian component applies the same transformation
to posterior model probabilities.
Let $\hat P^+_{mit}(s)$ denote the positive perceived fit used
to evaluate this transformation.\footnote{
For numerical stability,
$\hat P^+_{mit}(s)=\max\{\hat P_{mit}(s),10^{-10}\}$.
}
Posterior model probabilities are proportional to
$\rho_0(m)\hat P^+_{mit}(s)$.
The transformation multiplies their $\beta_{2,k}$ powers
by the model prior raised to $\beta_{1,k}$.
Its weight on model A is therefore
\begin{equation}
\rho^G_{kit}(A\mid s)
=
\frac{
\rho_0(A)^{\beta_{1,k}+\beta_{2,k}}
[\hat P^+_{Ait}(s)]^{\beta_{2,k}}
}{
\rho_0(A)^{\beta_{1,k}+\beta_{2,k}}
[\hat P^+_{Ait}(s)]^{\beta_{2,k}}
+
\rho_0(B)^{\beta_{1,k}+\beta_{2,k}}
[\hat P^+_{Bit}(s)]^{\beta_{2,k}}
}.
\label{eq:grether_weight_app}
\end{equation}
The weight on model B is $1-\rho^G_{kit}(A\mid s)$. Here, $\rho_0(m)$ is the prior probability of model $m$.

The Grether-Bayesian internal belief is the weighted average of the transformed model-specific posteriors:
\begin{equation}
\mu_{k,it}(\omega_1\mid s)
=\rho^G_{kit}(A\mid s)\mu^G_{A,k}(\omega_1\mid s)
+[1-\rho^G_{kit}(A\mid s)]\mu^G_{B,k}(\omega_1\mid s).
\label{eq:grether_belief_app}
\end{equation}

The Grether-Best-fit component selects the model with the larger
transformed model probability.
It selects model A when
\begin{equation}
\rho_0(A)^{\beta_{1,k}+\beta_{2,k}}
[\hat P^+_{Ait}(s)]^{\beta_{2,k}}
>
\rho_0(B)^{\beta_{1,k}+\beta_{2,k}}
[\hat P^+_{Bit}(s)]^{\beta_{2,k}},
\label{eq:grether_selection_app}
\end{equation}
and model B otherwise.

In the experiment, $\rho_0(A)=\rho_0(B)=1/2$.
The model-prior terms therefore cancel from both the model weights
and the selection comparison.
Because $\beta_{2,k}>0$, the selection comparison reduces to
$\hat P^+_{Ait}(s)>\hat P^+_{Bit}(s)$.
Thus, the Grether-Best-fit component selects the model with the larger
positive perceived fit.

Its internal belief is the selected model's transformed posterior from Equation~\eqref{eq:grether_posterior_app}.
Both components add the shared normal reporting error and integrate
over the first-layer errors when evaluating report likelihoods.
Under this parameterization, $\beta_{1,k}=0$ and $\beta_{2,k}=1$
recover the Bayesian rule in the absence of errors.

With equal model priors, these parameter values also recover
the Best-fit rule in the absence of errors.
Setting both parameters to one gives an effective prior coefficient
of two and a signal-likelihood coefficient of one.

\paragraph*{Alternative sets of types.}
The two-type specification contains Bayesian and Best-fit components, allowing participants classified as Non-movers in the main specification to be reassigned.
The four-type specification adds one Grether-Bayesian component to the main three-type model.
Five-type specification I adds both Grether-Bayesian and Grether-Best-fit components to the main three-type model.
Five-type specification II adds two Grether-Bayesian components, A and B, with separately estimated parameters.
Component A imposes $\beta_{1,A}>\beta_{2,A}$, while component B imposes $\beta_{2,B}>\beta_{1,B}$.
These restrictions distinguish the two components within the implemented transformation.

The additional components change which participants receive the Bayesian and Best-fit classifications.
The comparisons below therefore assess sensitivity to the set of behavioral types and the resulting participant assignments.
Five-type specification II assigns 11 participants to Grether-Bayesian A. The four-type specification and five-type specifications I and II classify 29, 28, and 24 participants as Non-movers, respectively, compared with 28 in the main specification.

Each specification reports an assignment table, a figure of motivated reasoning effects, and two regression tables.
The regression tables use participant covariates and interactions with posterior type probabilities, respectively.
Participant covariates include gender, race, cognitive score, risk attitude, overconfidence, and narrative extremeness.
These regressions also include fixed effects for signal color, round, and decision problem.
Standard errors are clustered at the participant level.

\subsubsection{Two-Type Specification}

\begin{table}[htbp]
\centering
\begin{threeparttable}
\small
\begin{tabular}{@{}c@{}}
\begin{adjustbox}{max width=\linewidth}
{
\def\sym#1{\ifmmode^{#1}\else\(^{#1}\)\fi}
\begin{tabular}{l*{3}{c}}
\toprule
 & \multicolumn{3}{c}{} \\
 & Bayesian & Best-fit & Total \\
\midrule
Bayesian & 132 & 0 & 132 \\
Best-fit & 1 & 32 & 33 \\
Non-mover & 28 & 0 & 28 \\
Total & 161 & 32 & 193 \\
\bottomrule
\end{tabular}
}

\end{adjustbox}
\end{tabular}
\begin{tablenotes}
\footnotesize
\item \textit{Notes:} The table compares participant classifications under the main three-type specification and the alternative specification. Rows report the main classification, and columns report the alternative classification. Each entry gives the number of participants. Both specifications assign each of the 193 participants to the type with the highest posterior probability, using beliefs from the \textit{SP} block.
\end{tablenotes}
\end{threeparttable}
\caption{Two-Type and Main Three-Type Assignments}
\label{tab:xtab_3_2}
\end{table}

\begin{figure}[htbp]
\centering
\includegraphics[width=0.88\linewidth]{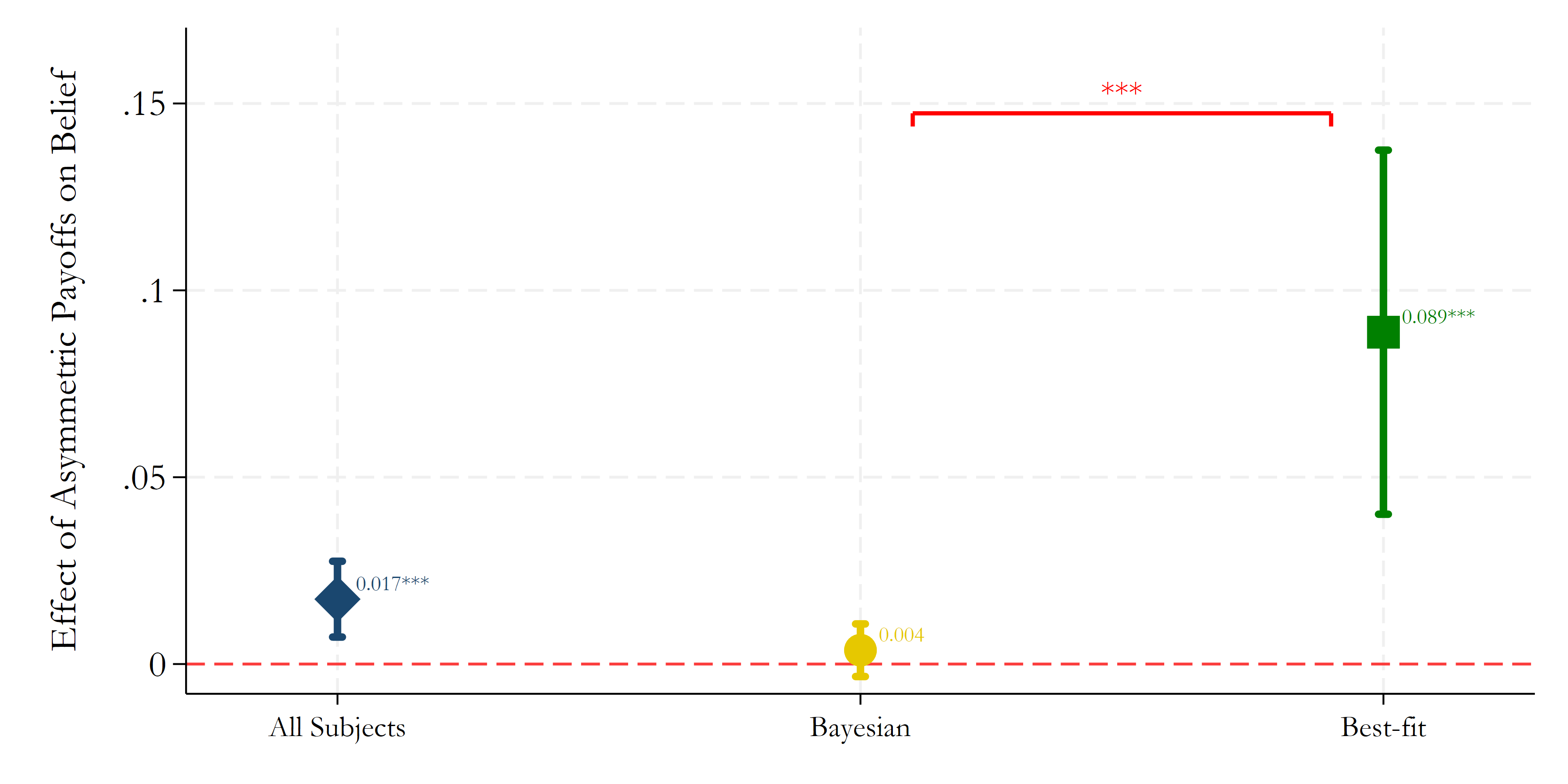}
\begin{minipage}{0.88\linewidth}
\footnotesize \textit{Notes:} The figure reports estimated coefficients on $\mathbbm{1}_{\text{Asym}}$ in Equation~\eqref{eqn:ols} for the pooled sample and separately for each behavioral type. The pooled, Bayesian, and Best-fit samples contain 193, 161, and 32 participants, respectively. Each participant contributes 54 beliefs. Regressions control for gender, race, cognitive score, risk attitude, overconfidence, and narrative extremeness. All specifications include fixed effects for signal color, round, and decision problem. Bars indicate 95\% confidence intervals based on standard errors clustered at the participant level. The red bracket reports the significance of the difference between the Best-fit and Bayesian coefficients. $^{*}p<0.10$, $^{**}p<0.05$, $^{***}p<0.01$.
\end{minipage}
\caption{Two-Type Motivated Reasoning Effects: All Types}
\label{fig:alternative_two_all}
\end{figure}

\begin{table}[htbp]
\centering
\begin{threeparttable}
\small
\begin{tabular}{@{}c@{}}
\begin{adjustbox}{max width=\linewidth}
{
\def\sym#1{\ifmmode^{#1}\else\(^{#1}\)\fi}
\begin{tabular}{l*{3}{c}}
\toprule
DV: Reported Belief & (1) & (2) & (3) \\
\midrule
\(\mathbbm{1}_{\text{Asym}}\) & 0.017\sym{***} & 0.004 & 0.089\sym{***} \\
 & (0.005) & (0.004) & (0.024) \\
\midrule
Observations & 10422 & 8694 & 1728 \\
Participants & 193 & 161 & 32 \\
Mean of DV & 0.506 & 0.512 & 0.476 \\
\(R^2\) & 0.336 & 0.414 & 0.230 \\
Sample & All participants & Bayesian & Best-fit \\
\bottomrule
\end{tabular}
}

\end{adjustbox}
\end{tabular}
\begin{tablenotes}
\footnotesize
\item \textit{Notes:} The table reports OLS estimates. The dependent variable is reported belief. Columns (1)--(3) use the pooled sample, Bayesian participants, and Best-fit participants, respectively. Each participant contributes 54 beliefs. Unreported controls include gender, race indicators, cognitive score, risk attitude, overconfidence, and narrative extremeness. The specification includes fixed effects for signal color, round, and decision problem. Standard errors clustered at the participant level appear in parentheses. $^{*}p<0.10$, $^{**}p<0.05$, $^{***}p<0.01$.
\end{tablenotes}
\end{threeparttable}
\caption{Two-Type Motivated Reasoning Regressions: Participant Covariates}
\label{tab:alternative_two_source}
\end{table}

\begin{table}[htbp]
\centering
\begin{threeparttable}
\small
\begin{tabular}{@{}c@{}}
\begin{adjustbox}{max width=\linewidth}
{
\def\sym#1{\ifmmode^{#1}\else\(^{#1}\)\fi}
\begin{tabular}{l*{3}{c}}
\toprule
DV: Reported Belief & (1) & (2) & (3) \\
\midrule
\(\mathbbm{1}_{\text{Asym}}\) & 0.003 & 0.003 & 0.002 \\
 & (0.004) & (0.004) & (0.003) \\
\(\mathbbm{1}_{\text{Asym}}\) \(\times\) Likelihood of Best-fit Type & 0.086\sym{***} & 0.086\sym{***} & 0.082\sym{***} \\
 & (0.023) & (0.023) & (0.023) \\
Likelihood of Best-fit Type & -0.078\sym{***} & -0.072\sym{***} &  \\
 & (0.019) & (0.019) &  \\
\midrule
Observations & 10422 & 10422 & 10422 \\
Participants & 193 & 193 & 193 \\
Mean of DV & 0.506 & 0.506 & 0.506 \\
\(R^2\) & 0.341 & 0.343 & 0.515 \\
Participant FE & No & No & Yes \\
DP \(\times\) signal FE & No & No & Yes \\
DP FE & Yes & Yes & No \\
Participant covariates & No & Yes & Absorbed \\
Narrative extremeness & No & Yes & Absorbed \\
\bottomrule
\end{tabular}
}

\end{adjustbox}
\end{tabular}
\begin{tablenotes}
\footnotesize
\item \textit{Notes:} The table reports OLS estimates of how the effect of motivated reasoning varies with posterior type probabilities. The dependent variable is reported belief. Each column uses 10,422 beliefs from 193 participants. Posterior type probabilities are estimated from the 27 beliefs per participant in the \textit{SP} block. The Bayesian type is the reference category. Each regression includes $\mathbbm{1}_{\text{Asym}}$, the other types' posterior probabilities, and their interactions with $\mathbbm{1}_{\text{Asym}}$. Columns (1) and (2) include fixed effects for signal color, round, and decision problem. Column (2) also controls for gender, race indicators, cognitive score, risk attitude, overconfidence, and narrative extremeness. Column (3) includes participant, decision-problem-by-signal, and round fixed effects. Participant fixed effects absorb the posterior type probabilities in column (3). Standard errors clustered at the participant level appear in parentheses. $^{*}p<0.10$, $^{**}p<0.05$, $^{***}p<0.01$.
\end{tablenotes}
\end{threeparttable}
\caption{Two-Type Posterior-Probability Interactions}
\label{tab:alternative_two_probability}
\end{table}

With participant covariates, the effect of motivated reasoning is 8.5 percentage points larger among Best-fit participants than among Bayesian participants ($p<0.001$). A test that the effect of motivated reasoning is equal across all types gives $p<0.001$. Both tests come from a single regression that interacts each control and fixed effect with type.

\FloatBarrier

\subsubsection{Four-Type Specification}

\begin{table}[htbp]
\centering
\begin{threeparttable}
\small
\begin{tabular}{@{}c@{}}
\begin{adjustbox}{max width=\linewidth}
{
\def\sym#1{\ifmmode^{#1}\else\(^{#1}\)\fi}
\begin{tabular}{l*{5}{c}}
\toprule
 & \multicolumn{5}{c}{} \\
 & Bayesian & Best-fit & Non-mover & Grether-Bayesian & Total \\
\midrule
Bayesian & 128 & 3 & 1 & 0 & 132 \\
Best-fit & 1 & 19 & 0 & 13 & 33 \\
Non-mover & 0 & 0 & 28 & 0 & 28 \\
Total & 129 & 22 & 29 & 13 & 193 \\
\bottomrule
\end{tabular}
}

\end{adjustbox}
\end{tabular}
\begin{tablenotes}
\footnotesize
\item \textit{Notes:} The table compares participant classifications under the main three-type specification and the alternative specification. Rows report the main classification, and columns report the alternative classification. Each entry gives the number of participants. Both specifications assign each of the 193 participants to the type with the highest posterior probability, using beliefs from the \textit{SP} block.
\end{tablenotes}
\end{threeparttable}
\caption{Four-Type and Main Three-Type Assignments}
\label{tab:xtab_3_4}
\end{table}

\begin{figure}[htbp]
\centering
\includegraphics[width=0.88\linewidth]{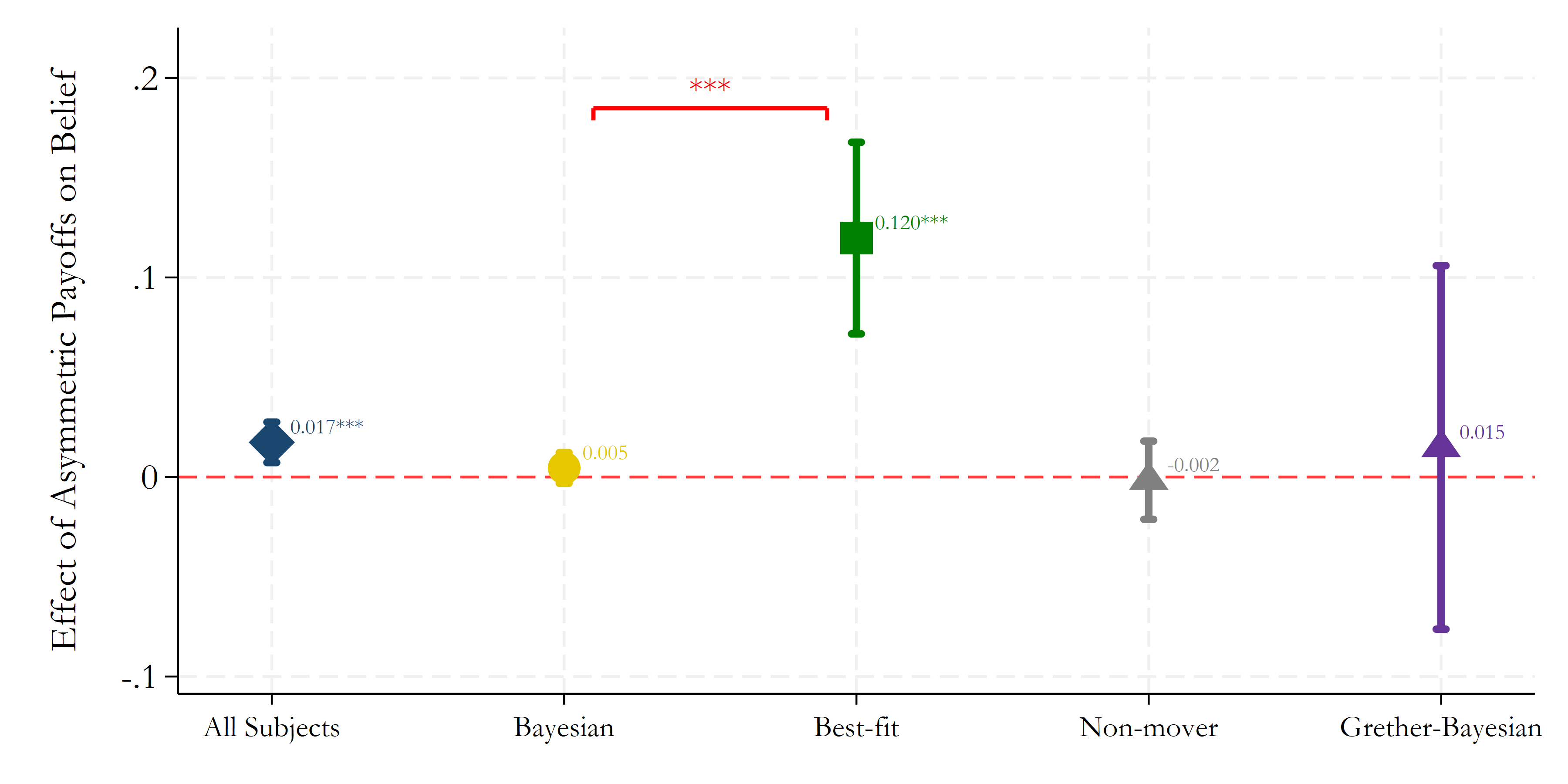}
\begin{minipage}{0.88\linewidth}
\footnotesize \textit{Notes:} The figure reports estimated coefficients on $\mathbbm{1}_{\text{Asym}}$ in Equation~\eqref{eqn:ols} for the pooled sample and separately for each behavioral type. The pooled sample contains 193 participants. The Bayesian, Best-fit, Non-mover, and Grether-Bayesian samples contain 129, 22, 29, and 13 participants, respectively. Each participant contributes 54 beliefs. Regressions control for gender, race, cognitive score, risk attitude, overconfidence, and narrative extremeness. All specifications include fixed effects for signal color, round, and decision problem. Bars indicate 95\% confidence intervals based on standard errors clustered at the participant level. The red bracket reports the significance of the difference between the Best-fit and Bayesian coefficients. $^{*}p<0.10$, $^{**}p<0.05$, $^{***}p<0.01$.
\end{minipage}
\caption{Four-Type Motivated Reasoning Effects: All Types}
\label{fig:alternative_four_all}
\end{figure}

\begin{table}[htbp]
\centering
\begin{threeparttable}
\small
\begin{tabular}{@{}c@{}}
\begin{adjustbox}{max width=\linewidth}
{
\def\sym#1{\ifmmode^{#1}\else\(^{#1}\)\fi}
\begin{tabular}{l*{5}{c}}
\toprule
DV: Reported Belief & (1) & (2) & (3) & (4) & (5) \\
\midrule
\(\mathbbm{1}_{\text{Asym}}\) & 0.017\sym{***} & 0.005 & 0.120\sym{***} & -0.002 & 0.015 \\
 & (0.005) & (0.004) & (0.023) & (0.010) & (0.042) \\
\midrule
Observations & 10422 & 6966 & 1188 & 1566 & 702 \\
Participants & 193 & 129 & 22 & 29 & 13 \\
Mean of DV & 0.506 & 0.511 & 0.472 & 0.510 & 0.511 \\
\(R^2\) & 0.336 & 0.417 & 0.262 & 0.519 & 0.403 \\
Sample & All participants & Bayesian & Best-fit & Non-mover & Grether-Bayesian \\
\bottomrule
\end{tabular}
}

\end{adjustbox}
\end{tabular}
\begin{tablenotes}
\footnotesize
\item \textit{Notes:} The table reports OLS estimates. The dependent variable is reported belief. Columns (1)--(4) use the pooled sample, Bayesian participants, Best-fit participants, and Non-movers, respectively. Column (5) uses Grether-Bayesian participants. Each participant contributes 54 beliefs. Unreported controls include gender, race indicators, cognitive score, risk attitude, overconfidence, and narrative extremeness. The specification includes fixed effects for signal color, round, and decision problem. Standard errors clustered at the participant level appear in parentheses. $^{*}p<0.10$, $^{**}p<0.05$, $^{***}p<0.01$.
\end{tablenotes}
\end{threeparttable}
\caption{Four-Type Motivated Reasoning Regressions: Participant Covariates}
\label{tab:alternative_four_source}
\end{table}

\begin{table}[htbp]
\centering
\begin{threeparttable}
\small
\begin{tabular}{@{}c@{}}
\begin{adjustbox}{max width=\linewidth}
{
\def\sym#1{\ifmmode^{#1}\else\(^{#1}\)\fi}
\begin{tabular}{l*{3}{c}}
\toprule
DV: Reported Belief & (1) & (2) & (3) \\
\midrule
\(\mathbbm{1}_{\text{Asym}}\) & 0.004 & 0.004 & 0.004 \\
 & (0.004) & (0.004) & (0.003) \\
\(\mathbbm{1}_{\text{Asym}}\) \(\times\) Likelihood of Best-fit Type & 0.117\sym{***} & 0.117\sym{***} & 0.114\sym{***} \\
 & (0.023) & (0.023) & (0.022) \\
\(\mathbbm{1}_{\text{Asym}}\) \(\times\) Likelihood of Non-mover Type & -0.007 & -0.007 & -0.013 \\
 & (0.011) & (0.011) & (0.010) \\
\(\mathbbm{1}_{\text{Asym}}\) \(\times\) Likelihood of Grether-Bayesian Type & 0.009 & 0.009 & 0.003 \\
 & (0.036) & (0.036) & (0.036) \\
Likelihood of Best-fit Type & -0.102\sym{***} & -0.097\sym{***} &  \\
 & (0.022) & (0.021) &  \\
Likelihood of Non-mover Type & 0.001 & 0.006 &  \\
 & (0.007) & (0.007) &  \\
Likelihood of Grether-Bayesian Type & -0.003 & 0.002 &  \\
 & (0.024) & (0.024) &  \\
\midrule
Observations & 10422 & 10422 & 10422 \\
Participants & 193 & 193 & 193 \\
Mean of DV & 0.506 & 0.506 & 0.506 \\
\(R^2\) & 0.343 & 0.346 & 0.518 \\
Participant FE & No & No & Yes \\
DP \(\times\) signal FE & No & No & Yes \\
DP FE & Yes & Yes & No \\
Participant covariates & No & Yes & Absorbed \\
Narrative extremeness & No & Yes & Absorbed \\
\bottomrule
\end{tabular}
}

\end{adjustbox}
\end{tabular}
\begin{tablenotes}
\footnotesize
\item \textit{Notes:} The table reports OLS estimates of how the effect of motivated reasoning varies with posterior type probabilities. The dependent variable is reported belief. Each column uses 10,422 beliefs from 193 participants. Posterior type probabilities are estimated from the 27 beliefs per participant in the \textit{SP} block. The Bayesian type is the reference category. Each regression includes $\mathbbm{1}_{\text{Asym}}$, the other types' posterior probabilities, and their interactions with $\mathbbm{1}_{\text{Asym}}$. Columns (1) and (2) include fixed effects for signal color, round, and decision problem. Column (2) also controls for gender, race indicators, cognitive score, risk attitude, overconfidence, and narrative extremeness. Column (3) includes participant, decision-problem-by-signal, and round fixed effects. Participant fixed effects absorb the posterior type probabilities in column (3). Standard errors clustered at the participant level appear in parentheses. $^{*}p<0.10$, $^{**}p<0.05$, $^{***}p<0.01$.
\end{tablenotes}
\end{threeparttable}
\caption{Four-Type Posterior-Probability Interactions}
\label{tab:alternative_four_probability}
\end{table}

With participant covariates, the effect of motivated reasoning is 11.5 percentage points larger among Best-fit participants than among Bayesian participants ($p<0.001$). A test that the effect of motivated reasoning is equal across all types gives $p<0.001$. Both tests come from a single regression that interacts each control and fixed effect with type.

\FloatBarrier

\subsubsection{Five-Type I Specification}

\begin{table}[htbp]
\centering
\begin{threeparttable}
\small
\begin{tabular}{@{}c@{}}
\begin{adjustbox}{max width=\linewidth}
{
\def\sym#1{\ifmmode^{#1}\else\(^{#1}\)\fi}
\begin{tabular}{l*{6}{c}}
\toprule
 & \multicolumn{6}{c}{} \\
 & Bayesian & Best-fit & Non-mover & Grether-Bayesian & Grether-Best-fit & Total \\
\midrule
Bayesian & 117 & 2 & 1 & 12 & 0 & 132 \\
Best-fit & 1 & 21 & 0 & 6 & 5 & 33 \\
Non-mover & 0 & 0 & 27 & 1 & 0 & 28 \\
Total & 118 & 23 & 28 & 19 & 5 & 193 \\
\bottomrule
\end{tabular}
}

\end{adjustbox}
\end{tabular}
\begin{tablenotes}
\footnotesize
\item \textit{Notes:} The table compares participant classifications under the main three-type specification and the alternative specification. Rows report the main classification, and columns report the alternative classification. Each entry gives the number of participants. Both specifications assign each of the 193 participants to the type with the highest posterior probability, using beliefs from the \textit{SP} block.
\end{tablenotes}
\end{threeparttable}
\caption{Five-Type I and Main Three-Type Assignments}
\label{tab:xtab_3_5}
\end{table}

\begin{figure}[htbp]
\centering
\includegraphics[width=0.88\linewidth]{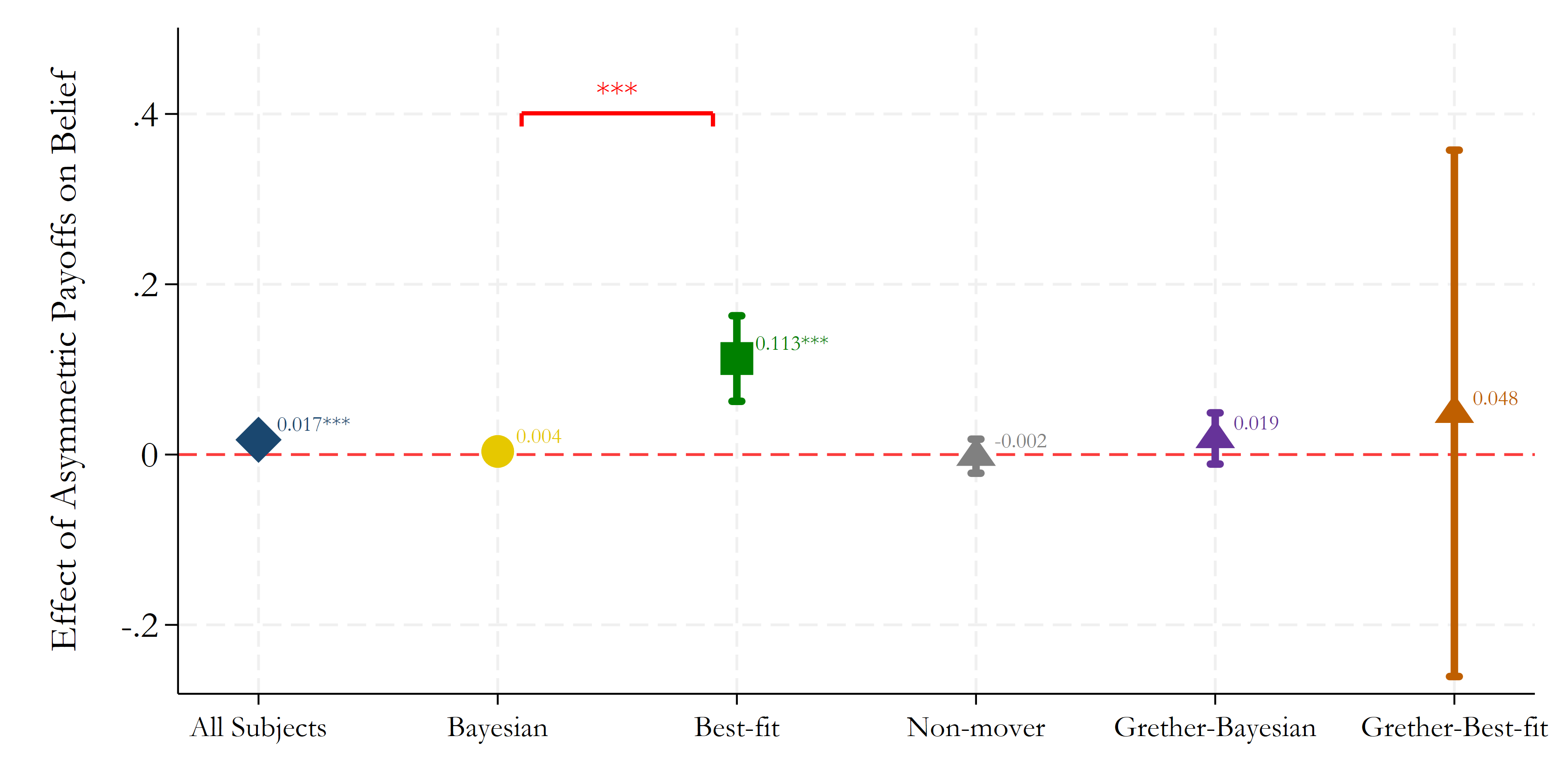}
\begin{minipage}{0.88\linewidth}
\footnotesize \textit{Notes:} The figure reports estimated coefficients on $\mathbbm{1}_{\text{Asym}}$ in Equation~\eqref{eqn:ols} for the pooled sample and separately for each behavioral type. The pooled sample contains 193 participants. The Bayesian, Best-fit, and Non-mover samples contain 118, 23, and 28 participants, respectively. The Grether-Bayesian and Grether-Best-fit samples contain 19 and 5 participants, respectively. Each participant contributes 54 beliefs. Regressions control for gender, race, cognitive score, risk attitude, overconfidence, and narrative extremeness. All specifications include fixed effects for signal color, round, and decision problem. Bars indicate 95\% confidence intervals based on standard errors clustered at the participant level. The red bracket reports the significance of the difference between the Best-fit and Bayesian coefficients. $^{*}p<0.10$, $^{**}p<0.05$, $^{***}p<0.01$.
\end{minipage}
\caption{Five-Type I Motivated Reasoning Effects: All Types}
\label{fig:alternative_five1_all}
\end{figure}

\begin{table}[htbp]
\centering
\begin{threeparttable}
\small
\begin{adjustbox}{max width=\linewidth}
{
\def\sym#1{\ifmmode^{#1}\else\(^{#1}\)\fi}
\begin{tabular}{l*{6}{c}}
\toprule
DV: Reported Belief & (1) & (2) & (3) & (4) & (5) & (6) \\
\midrule
\(\mathbbm{1}_{\text{Asym}}\) & 0.017\sym{***} & 0.004 & 0.113\sym{***} & -0.002 & 0.019 & 0.048 \\
 & (0.005) & (0.004) & (0.024) & (0.010) & (0.014) & (0.111) \\
\midrule
Observations & 10422 & 6372 & 1242 & 1512 & 1026 & 270 \\
Participants & 193 & 118 & 23 & 28 & 19 & 5 \\
Mean of DV & 0.506 & 0.508 & 0.477 & 0.509 & 0.533 & 0.493 \\
\(R^2\) & 0.336 & 0.422 & 0.243 & 0.530 & 0.498 & 0.560 \\
Sample & All participants & Bayesian & Best-fit & Non-mover & Grether-Bayesian & Grether-Best-fit \\
\bottomrule
\end{tabular}
}

\end{adjustbox}
\begin{tablenotes}
\footnotesize
\item \textit{Notes:} The table reports OLS estimates. The dependent variable is reported belief. Columns (1)--(4) use the pooled sample, Bayesian participants, Best-fit participants, and Non-movers, respectively. Columns (5) and (6) use Grether-Bayesian and Grether-Best-fit participants, respectively. Each participant contributes 54 beliefs. Unreported controls include gender, race indicators, cognitive score, risk attitude, overconfidence, and narrative extremeness. The specification includes fixed effects for signal color, round, and decision problem. Standard errors clustered at the participant level appear in parentheses. $^{*}p<0.10$, $^{**}p<0.05$, $^{***}p<0.01$.
\end{tablenotes}
\end{threeparttable}
\caption{Five-Type I Motivated Reasoning Regressions: Participant Covariates}
\label{tab:alternative_five1_source}
\end{table}

\begin{table}[htbp]
\centering
\begin{threeparttable}
\small
\begin{tabular}{@{}c@{}}
\begin{adjustbox}{max width=\linewidth}
{
\def\sym#1{\ifmmode^{#1}\else\(^{#1}\)\fi}
\begin{tabular}{l*{3}{c}}
\toprule
DV: Reported Belief & (1) & (2) & (3) \\
\midrule
\(\mathbbm{1}_{\text{Asym}}\) & 0.003 & 0.003 & 0.004 \\
 & (0.004) & (0.004) & (0.003) \\
\(\mathbbm{1}_{\text{Asym}}\) \(\times\) Likelihood of Best-fit Type & 0.108\sym{***} & 0.108\sym{***} & 0.105\sym{***} \\
 & (0.024) & (0.024) & (0.024) \\
\(\mathbbm{1}_{\text{Asym}}\) \(\times\) Likelihood of Non-mover Type & -0.006 & -0.006 & -0.012 \\
 & (0.011) & (0.011) & (0.010) \\
\(\mathbbm{1}_{\text{Asym}}\) \(\times\) Likelihood of Grether-Bayesian Type & 0.007 & 0.007 & 0.004 \\
 & (0.015) & (0.015) & (0.014) \\
\(\mathbbm{1}_{\text{Asym}}\) \(\times\) Likelihood of Grether-Best-fit Type & 0.034 & 0.034 & 0.016 \\
 & (0.074) & (0.074) & (0.077) \\
Likelihood of Best-fit Type & -0.089\sym{***} & -0.083\sym{***} &  \\
 & (0.024) & (0.023) &  \\
Likelihood of Non-mover Type & 0.002 & 0.008 &  \\
 & (0.007) & (0.007) &  \\
Likelihood of Grether-Bayesian Type & 0.019\sym{*} & 0.019\sym{*} &  \\
 & (0.011) & (0.011) &  \\
Likelihood of Grether-Best-fit Type & -0.031 & -0.027 &  \\
 & (0.042) & (0.044) &  \\
\midrule
Observations & 10422 & 10422 & 10422 \\
Participants & 193 & 193 & 193 \\
Mean of DV & 0.506 & 0.506 & 0.506 \\
\(R^2\) & 0.342 & 0.345 & 0.517 \\
Participant FE & No & No & Yes \\
DP \(\times\) signal FE & No & No & Yes \\
DP FE & Yes & Yes & No \\
Participant covariates & No & Yes & Absorbed \\
Narrative extremeness & No & Yes & Absorbed \\
\bottomrule
\end{tabular}
}

\end{adjustbox}
\end{tabular}
\begin{tablenotes}
\footnotesize
\item \textit{Notes:} The table reports OLS estimates of how the effect of motivated reasoning varies with posterior type probabilities. The dependent variable is reported belief. Each column uses 10,422 beliefs from 193 participants. Posterior type probabilities are estimated from the 27 beliefs per participant in the \textit{SP} block. The Bayesian type is the reference category. Each regression includes $\mathbbm{1}_{\text{Asym}}$, the other types' posterior probabilities, and their interactions with $\mathbbm{1}_{\text{Asym}}$. Columns (1) and (2) include fixed effects for signal color, round, and decision problem. Column (2) also controls for gender, race indicators, cognitive score, risk attitude, overconfidence, and narrative extremeness. Column (3) includes participant, decision-problem-by-signal, and round fixed effects. Participant fixed effects absorb the posterior type probabilities in column (3). Standard errors clustered at the participant level appear in parentheses. $^{*}p<0.10$, $^{**}p<0.05$, $^{***}p<0.01$.
\end{tablenotes}
\end{threeparttable}
\caption{Five-Type I Posterior-Probability Interactions}
\label{tab:alternative_five1_probability}
\end{table}

With participant covariates, the effect of motivated reasoning is 10.9 percentage points larger among Best-fit participants than among Bayesian participants ($p<0.001$). A test that the effect of motivated reasoning is equal across all types gives $p<0.001$. Both tests come from a single regression that interacts each control and fixed effect with type.

\FloatBarrier

\subsubsection{Five-Type II Specification}

\begin{table}[htbp]
\centering
\begin{threeparttable}
\small
\begin{adjustbox}{max width=\linewidth}
{
\def\sym#1{\ifmmode^{#1}\else\(^{#1}\)\fi}
\begin{tabular}{l*{6}{c}}
\toprule
 & \multicolumn{6}{c}{} \\
 & Bayesian & Best-fit & Non-mover & Grether-Bayesian A & Grether-Bayesian B & Total \\
\midrule
Bayesian & 121 & 3 & 1 & 7 & 0 & 132 \\
Best-fit & 1 & 19 & 0 & 0 & 13 & 33 \\
Non-mover & 1 & 0 & 23 & 4 & 0 & 28 \\
Total & 123 & 22 & 24 & 11 & 13 & 193 \\
\bottomrule
\end{tabular}
}

\end{adjustbox}
\begin{tablenotes}
\footnotesize
\item \textit{Notes:} The table compares participant classifications under the main three-type specification and the alternative specification. Rows report the main classification, and columns report the alternative classification. Each entry gives the number of participants. Both specifications assign each of the 193 participants to the type with the highest posterior probability, using beliefs from the \textit{SP} block.
\end{tablenotes}
\end{threeparttable}
\caption{Five-Type II and Main Three-Type Assignments}
\label{tab:xtab_3_5a}
\end{table}

\begin{figure}[htbp]
\centering
\includegraphics[width=0.88\linewidth]{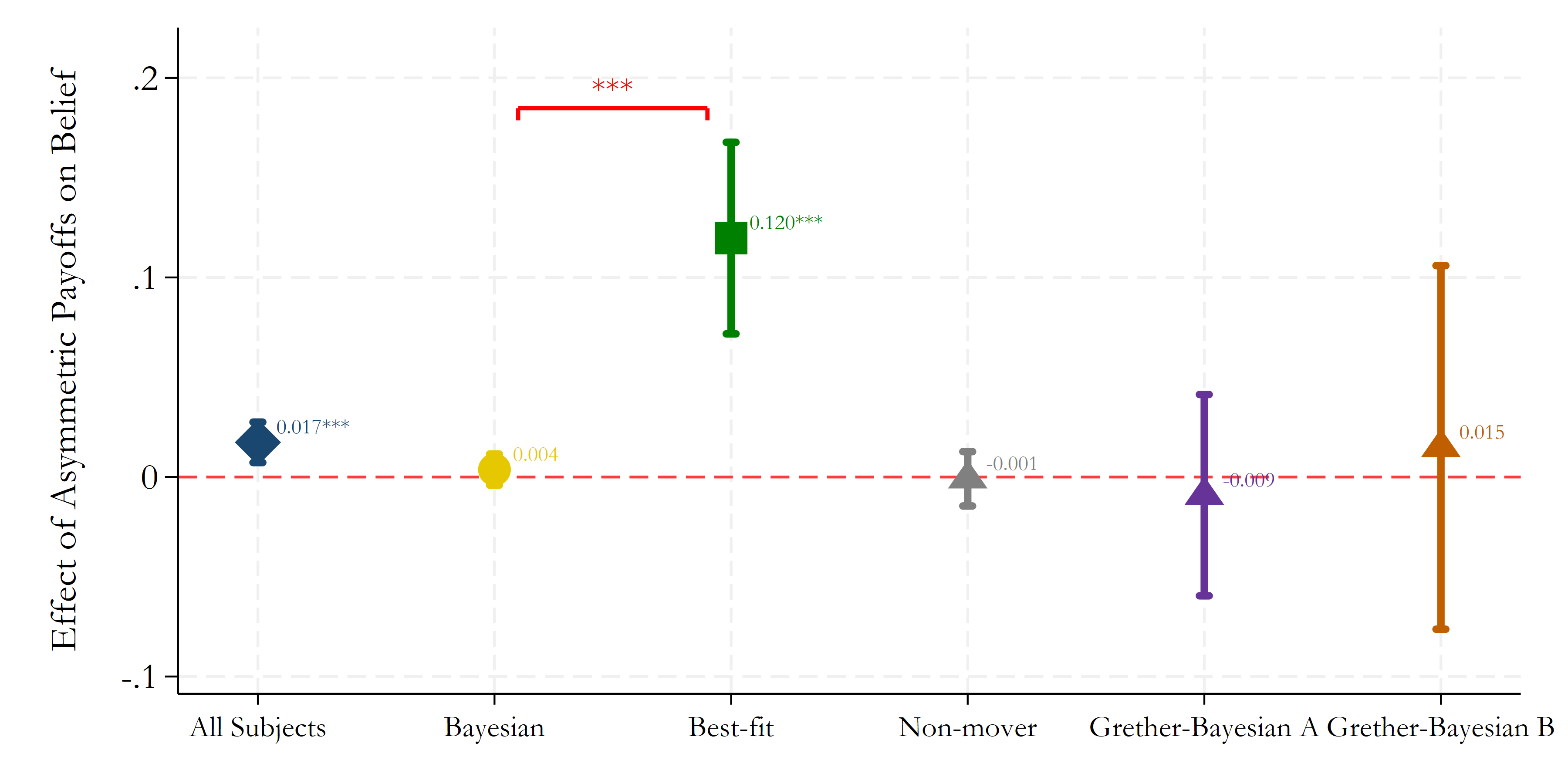}
\begin{minipage}{0.88\linewidth}
\footnotesize \textit{Notes:} The figure reports estimated coefficients on $\mathbbm{1}_{\text{Asym}}$ in Equation~\eqref{eqn:ols} for the pooled sample and separately for each behavioral type. The pooled sample contains 193 participants. The Bayesian, Best-fit, and Non-mover samples contain 123, 22, and 24 participants, respectively. The Grether-Bayesian A and B samples contain 11 and 13 participants, respectively. Each participant contributes 54 beliefs. Regressions control for gender, race, cognitive score, risk attitude, overconfidence, and narrative extremeness. All specifications include fixed effects for signal color, round, and decision problem. Bars indicate 95\% confidence intervals based on standard errors clustered at the participant level. The red bracket reports the significance of the difference between the Best-fit and Bayesian coefficients. $^{*}p<0.10$, $^{**}p<0.05$, $^{***}p<0.01$.
\end{minipage}
\caption{Five-Type II Motivated Reasoning Effects: All Types}
\label{fig:alternative_five2_all}
\end{figure}

\begin{table}[htbp]
\centering
\begin{threeparttable}
\small
\begin{adjustbox}{max width=\linewidth}
{
\def\sym#1{\ifmmode^{#1}\else\(^{#1}\)\fi}
\begin{tabular}{l*{6}{c}}
\toprule
DV: Reported Belief & (1) & (2) & (3) & (4) & (5) & (6) \\
\midrule
\(\mathbbm{1}_{\text{Asym}}\) & 0.017\sym{***} & 0.004 & 0.120\sym{***} & -0.001 & -0.009 & 0.015 \\
 & (0.005) & (0.004) & (0.023) & (0.007) & (0.023) & (0.042) \\
\midrule
Observations & 10422 & 6642 & 1188 & 1296 & 594 & 702 \\
Participants & 193 & 123 & 22 & 24 & 11 & 13 \\
Mean of DV & 0.506 & 0.512 & 0.472 & 0.509 & 0.505 & 0.511 \\
\(R^2\) & 0.336 & 0.441 & 0.262 & 0.680 & 0.287 & 0.403 \\
Sample & All participants & Bayesian & Best-fit & Non-mover & Grether-Bayesian A & Grether-Bayesian B \\
\bottomrule
\end{tabular}
}

\end{adjustbox}
\begin{tablenotes}
\footnotesize
\item \textit{Notes:} The table reports OLS estimates. The dependent variable is reported belief. Columns (1)--(4) use the pooled sample, Bayesian participants, Best-fit participants, and Non-movers, respectively. Columns (5) and (6) use Grether-Bayesian A and Grether-Bayesian B participants, respectively. Each participant contributes 54 beliefs. Unreported controls include gender, race indicators, cognitive score, risk attitude, overconfidence, and narrative extremeness. The specification includes fixed effects for signal color, round, and decision problem. Standard errors clustered at the participant level appear in parentheses. $^{*}p<0.10$, $^{**}p<0.05$, $^{***}p<0.01$.
\end{tablenotes}
\end{threeparttable}
\caption{Five-Type II Motivated Reasoning Regressions: Participant Covariates}
\label{tab:alternative_five2_source}
\end{table}

\begin{table}[htbp]
\centering
\begin{threeparttable}
\small
\begin{tabular}{@{}c@{}}
\begin{adjustbox}{max width=\linewidth}
{
\def\sym#1{\ifmmode^{#1}\else\(^{#1}\)\fi}
\begin{tabular}{l*{3}{c}}
\toprule
DV: Reported Belief & (1) & (2) & (3) \\
\midrule
\(\mathbbm{1}_{\text{Asym}}\) & 0.003 & 0.003 & 0.003 \\
 & (0.004) & (0.004) & (0.003) \\
\(\mathbbm{1}_{\text{Asym}}\) \(\times\) Likelihood of Best-fit Type & 0.117\sym{***} & 0.117\sym{***} & 0.114\sym{***} \\
 & (0.023) & (0.023) & (0.022) \\
\(\mathbbm{1}_{\text{Asym}}\) \(\times\) Likelihood of Non-mover Type & -0.001 & -0.001 & -0.008 \\
 & (0.008) & (0.008) & (0.008) \\
\(\mathbbm{1}_{\text{Asym}}\) \(\times\) Likelihood of Grether-Bayesian A Type & -0.001 & -0.001 & -0.004 \\
 & (0.022) & (0.022) & (0.021) \\
\(\mathbbm{1}_{\text{Asym}}\) \(\times\) Likelihood of Grether-Bayesian B Type & 0.010 & 0.010 & 0.004 \\
 & (0.035) & (0.035) & (0.036) \\
Likelihood of Best-fit Type & -0.101\sym{***} & -0.096\sym{***} &  \\
 & (0.022) & (0.021) &  \\
Likelihood of Non-mover Type & -0.004 & 0.002 &  \\
 & (0.006) & (0.007) &  \\
Likelihood of Grether-Bayesian A Type & -0.006 & -0.004 &  \\
 & (0.013) & (0.012) &  \\
Likelihood of Grether-Bayesian B Type & -0.004 & 0.001 &  \\
 & (0.024) & (0.024) &  \\
\midrule
Observations & 10422 & 10422 & 10422 \\
Participants & 193 & 193 & 193 \\
Mean of DV & 0.506 & 0.506 & 0.506 \\
\(R^2\) & 0.343 & 0.345 & 0.518 \\
Participant FE & No & No & Yes \\
DP \(\times\) signal FE & No & No & Yes \\
DP FE & Yes & Yes & No \\
Participant covariates & No & Yes & Absorbed \\
Narrative extremeness & No & Yes & Absorbed \\
\bottomrule
\end{tabular}
}

\end{adjustbox}
\end{tabular}
\begin{tablenotes}
\footnotesize
\item \textit{Notes:} The table reports OLS estimates of how the effect of motivated reasoning varies with posterior type probabilities. The dependent variable is reported belief. Each column uses 10,422 beliefs from 193 participants. Posterior type probabilities are estimated from the 27 beliefs per participant in the \textit{SP} block. The Bayesian type is the reference category. Each regression includes $\mathbbm{1}_{\text{Asym}}$, the other types' posterior probabilities, and their interactions with $\mathbbm{1}_{\text{Asym}}$. Columns (1) and (2) include fixed effects for signal color, round, and decision problem. Column (2) also controls for gender, race indicators, cognitive score, risk attitude, overconfidence, and narrative extremeness. Column (3) includes participant, decision-problem-by-signal, and round fixed effects. Participant fixed effects absorb the posterior type probabilities in column (3). Standard errors clustered at the participant level appear in parentheses. $^{*}p<0.10$, $^{**}p<0.05$, $^{***}p<0.01$.
\end{tablenotes}
\end{threeparttable}
\caption{Five-Type II Posterior-Probability Interactions}
\label{tab:alternative_five2_probability}
\end{table}

With participant covariates, the effect of motivated reasoning is 11.6 percentage points larger among Best-fit participants than among Bayesian participants ($p<0.001$). A test that the effect of motivated reasoning is equal across all types gives $p<0.001$. Both tests come from a single regression that interacts each control and fixed effect with type.

\clearpage

\subsection{Individual Characteristics, Narrative Opinions, and Other Behavioral Measures}
\label{app:individual_characteristics}

\subsubsection{Type Associations and Narrative Opinions}
\label{app:narrative_associations}

The type regressions describe associations between participant characteristics and modal type assignment.
Narrative opinions are collected after the laboratory tasks and provide an observational comparison across domains.
Their extremeness is the mean absolute distance from five across six responses on the eleven-point scale.

\begin{table}[htbp]
\centering
\begin{threeparttable}
\small
\begin{tabular}{@{}c@{}}
\begin{adjustbox}{max width=\linewidth}
{
\def\sym#1{\ifmmode^{#1}\else\(^{#1}\)\fi}
\begin{tabular}{l*{3}{c}}
\toprule
DV: Type Classification & (1) & (2) & (3) \\
\midrule
Female & 0.032 & -0.004 & -0.033 \\
 & (0.066) & (0.054) & (0.051) \\
Asian & 0.248\sym{**} & -0.165\sym{**} & -0.026 \\
 & (0.101) & (0.069) & (0.083) \\
White & 0.238\sym{**} & -0.234\sym{***} & 0.058 \\
 & (0.101) & (0.070) & (0.082) \\
Quiz Score & 0.056 & 0.025 & -0.074\sym{***} \\
 & (0.040) & (0.037) & (0.028) \\
Cognitive Score & 0.010 & -0.008 & -0.003 \\
 & (0.028) & (0.023) & (0.022) \\
Risk Attitude & 0.005\sym{***} & -0.001 & -0.004\sym{***} \\
 & (0.002) & (0.001) & (0.001) \\
Overconfidence & -0.075 & -0.034 & 0.112 \\
 & (0.080) & (0.064) & (0.069) \\
Narrative Extremeness & -0.094\sym{**} & 0.067\sym{**} & 0.024 \\
 & (0.037) & (0.029) & (0.029) \\
Prior Experiments & 0.014 & -0.016 & 0.002 \\
 & (0.011) & (0.010) & (0.009) \\
\midrule
Observations & 193 & 193 & 193 \\
Pseudo \(R^2\) & 0.109 & 0.121 & 0.099 \\
Mean of DV & 0.684 & 0.171 & 0.145 \\
Type & Bayesian & Best-fit & Non-mover \\
\bottomrule
\end{tabular}
}

\end{adjustbox}
\end{tabular}
\begin{tablenotes}
\footnotesize
\item \textit{Notes:} The table reports average marginal effects from binary logit regressions. The dependent variables indicate Bayesian, Best-fit, and Non-mover classification in columns (1)--(3), respectively. Each regression uses all 193 participants, with types assigned from their 27 beliefs in the \textit{SP} block. All displayed characteristics enter each regression. Narrative extremeness is the mean absolute distance from five across six zero-to-ten survey responses. Standard errors calculated using the delta method appear in parentheses. The dependent-variable mean gives the corresponding type's participant share. Pseudo $R^2$ refers to the underlying logit. $^{*}p<0.10$, $^{**}p<0.05$, $^{***}p<0.01$.
\end{tablenotes}
\end{threeparttable}
\caption{Participant Characteristics and Type Assignment}
\label{tab:type_determinant_full}
\end{table}

\begin{table}[htbp]
\centering
\begin{threeparttable}
\small
\begin{tabular}{@{}c@{}}
\begin{adjustbox}{max width=\linewidth}
\begin{tabular}{lccc}
\toprule
 & Spearman rho & p-value & Observations \\
\midrule
Laboratory and Narrative Extremeness & 0.136 & 0.060 & 193 \\
\bottomrule
\end{tabular}

\end{adjustbox}
\end{tabular}
\begin{tablenotes}
\footnotesize
\item \textit{Notes:} The table reports the Spearman rank correlation between laboratory belief extremeness and narrative extremeness. Belief extremeness averages absolute distance from one half across all 27 beliefs in the \textit{SP} block. Narrative extremeness averages absolute distance from five across all six zero-to-ten survey responses. The calculation uses 193 participants and includes no controls.
\end{tablenotes}
\end{threeparttable}
\caption{Rank Association between Laboratory and Narrative Extremeness}
\label{tab:lab_narrative_rank}
\end{table}

\begin{table}[htbp]
\centering
\begin{threeparttable}
\small
\begin{tabular}{@{}c@{}}
\begin{adjustbox}{max width=\linewidth}
{
\def\sym#1{\ifmmode^{#1}\else\(^{#1}\)\fi}
\begin{tabular}{l*{4}{c}}
\toprule
DV: Narrative Extremeness & (1) & (2) & (3) & (4) \\
\midrule
Best-fit & 0.409\sym{**} &  &  &  \\
 & (0.184) &  &  &  \\
Non-mover & 0.170 &  &  &  \\
 & (0.163) &  &  &  \\
SP Relative Benchmark Distance &  & 2.775\sym{**} & 2.780\sym{**} &  \\
 &  & (1.074) & (1.084) &  \\
Quiz Score &  &  & 0.073 & 0.082 \\
 &  &  & (0.071) & (0.073) \\
Other Seven Cognitive Tasks &  &  & 0.013 & 0.006 \\
 &  &  & (0.051) & (0.052) \\
Female &  &  & 0.056 & 0.037 \\
 &  &  & (0.127) & (0.130) \\
SP Absolute Bayesian Error &  &  &  & 1.399\sym{*} \\
 &  &  &  & (0.721) \\
Constant & 2.217\sym{***} & 2.573\sym{***} & 1.813\sym{**} & 1.348\sym{*} \\
 & (0.069) & (0.126) & (0.704) & (0.740) \\
\midrule
Observations & 193 & 193 & 193 & 193 \\
Mean of DV & 2.312 & 2.312 & 2.312 & 2.312 \\
\(R^2\) & 0.034 & 0.058 & 0.063 & 0.025 \\
Participant covariates & No & No & Yes & Yes \\
\bottomrule
\end{tabular}
}

\end{adjustbox}
\end{tabular}
\begin{tablenotes}
\footnotesize
\item \textit{Notes:} The table reports OLS estimates. The dependent variable is narrative extremeness, the mean absolute distance from five across six zero-to-ten responses. Each participant contributes one observation. Bayesian participants are the reference group in column (1). Relative benchmark distance equals mean absolute Bayesian error minus mean absolute Best-fit error over the 27 beliefs in the \textit{SP} block. Columns (2) and (3) use relative benchmark distance; column (4) uses mean absolute Bayesian error. Columns (3) and (4) add quiz score, performance on the other seven cognitive tasks, and gender. All covariates and constants are displayed. Robust standard errors appear in parentheses. $^{*}p<0.10$, $^{**}p<0.05$, $^{***}p<0.01$.
\end{tablenotes}
\end{threeparttable}
\caption{Narrative Opinions, Type Assignment, and Updating Performance}
\label{tab:narrative_associations_full}
\end{table}

\FloatBarrier

\subsubsection{Repeated Reports}
\label{app:repeatability}

Repeated-report differences measure consistency within identical observed information and preference conditions.

\begin{figure}[htbp]
\centering
\includegraphics[width=0.88\linewidth]{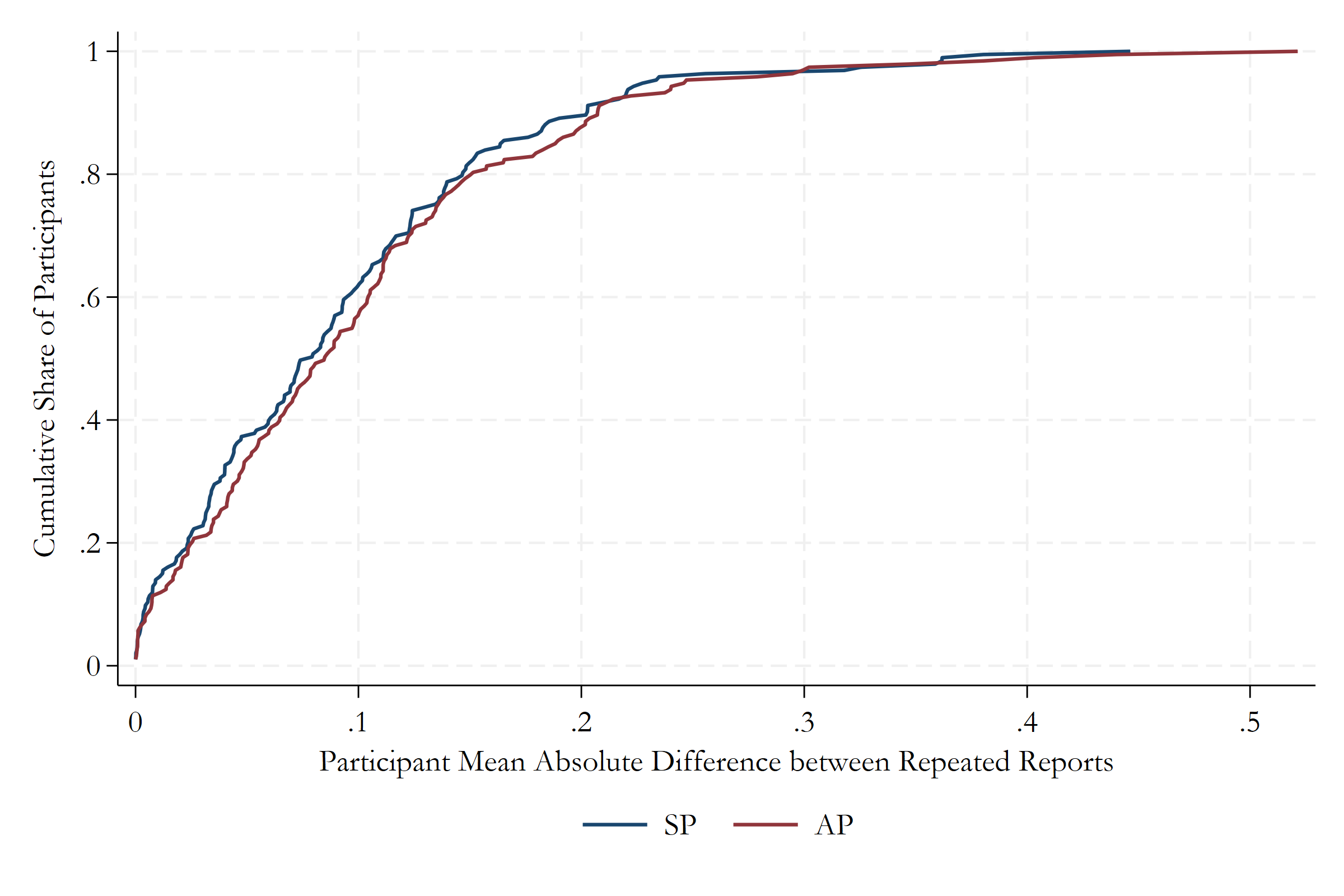}
\begin{minipage}{0.88\linewidth}
\footnotesize \textit{Notes:} The curves show cumulative distributions of participant-average absolute differences between beliefs reported under identical information in the same block. Repeated information conditions receive equal weight within each participant. The \textit{SP} and \textit{AP} curves each contain 193 participants, with equal weight per participant.
\end{minipage}
\caption{Distribution of Repeated-Report Differences}
\label{fig:repeatability_cdf_full}
\end{figure}

\FloatBarrier

\subsubsection{Implied Model Weights}
\label{app:weights}

An implied weight expresses a report as a linear combination of the lower and higher model-specific posterior probabilities.
The weight on the higher posterior equals the report's distance above the lower posterior divided by the gap between the two posteriors.
The single-model problem is excluded because its two posterior predictions coincide.
The report-level distributions retain every eligible report and cap weights below zero at zero and above one at one.

\FloatBarrier




\begin{figure}[htbp]
\centering
\includegraphics[width=0.88\linewidth]{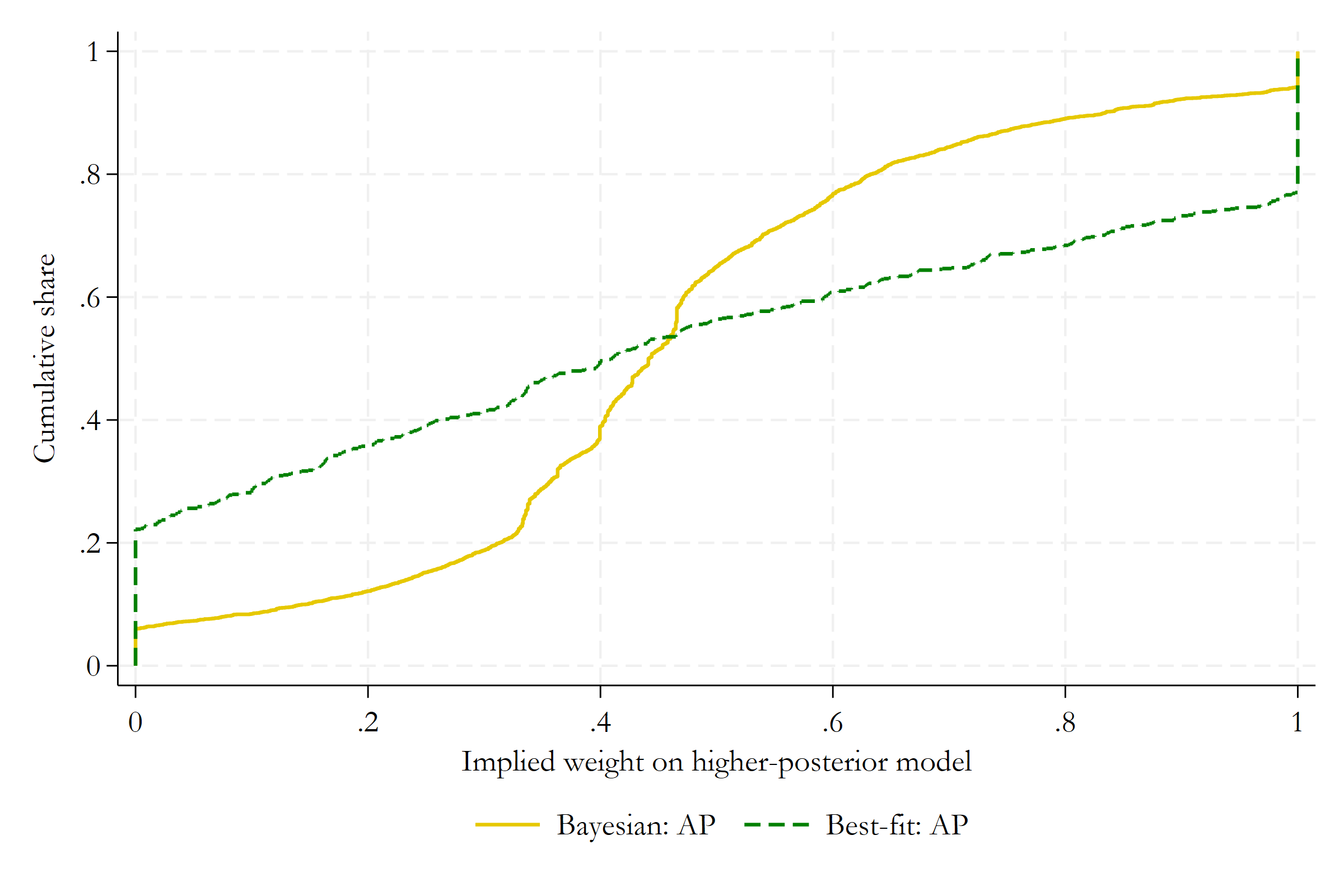}
\begin{minipage}{0.88\linewidth}\footnotesize
\textit{Notes:} The curves show cumulative distributions of weights on the higher model-specific posterior using beliefs from Decision Problems 1--4 and 6--9 in the \textit{AP} block. The Bayesian and Best-fit samples contain 3,168 beliefs from 132 participants and 792 beliefs from 33 participants, respectively. Weights below zero are set to zero, and weights above one are set to one. Each belief receives equal weight. The masses at zero and one include beliefs outside the interval between the two model-specific posteriors.
\end{minipage}
\caption{Report-level Implied Weights in the Asymmetric Payoff Block}
\label{fig:weights_capped_bp}
\end{figure}
\FloatBarrier

\newpage

\renewcommand{\thefigure}{\Alph{section}-\arabic{figure}}
\setcounter{figure}{0}
\renewcommand{\thetable}{\Alph{section}-\arabic{table}}
\setcounter{table}{0}

\section{Follow-up Study}\label{sec:robust}

The follow-up study was conducted at Purdue University in October 2025 with 85 participants. It is separately registered in the AEA RCT Registry (\href{https://doi.org/10.1257/rct.16882-1.0}{RCT 16882}). Participants completed the symmetric payoff block and supplementary tasks in a session lasting approximately 35 minutes. The belief-updating interface additionally displayed the correct posterior implied by each individual model in every scenario \citep{aina2025weighting}. This change reduces the computational demands of within-model updating while retaining the comparison between Bayesian averaging and Best-fit selection.

Table~\ref{tab:comparison} compares the shares of reports near each benchmark in the main and follow-up studies. Beliefs continue to cluster around both benchmarks when individual-model posteriors are supplied. The follow-up share near the Bayesian benchmark is lower, while the share near the Best-fit benchmark is higher. More reports fall within each window around the Bayesian benchmark than around the Best-fit benchmark. Because the studies involve different cohorts, these comparisons describe the persistence of the behavioral patterns across studies.

\vspace{0.5cm}

\begin{table}[htbp]
\centering
\begin{threeparttable}
\begin{tabular}{@{}c@{}}
\begin{adjustbox}{max width=\linewidth}
\begin{tabular}{lccccc}
\toprule
\multirow{2}{*}{Study} & \multirow{2}{*}{Type of Guess} & \multicolumn{2}{c}{Within 2 p.p.} & \multicolumn{2}{c}{Within 5 p.p.} \\
\cmidrule{3-4} \cmidrule{5-6}
 &  & \% & 95\%-CI & \% & 95\%-CI \\
\midrule
\multirow{2}{*}{Main Study} & Bayesian & 28.98 & [25.29, 32.67] & 47.75 & [43.21, 52.28] \\
 & Best-fit & 5.51 & [4.47,  6.54] & 12.15 & [10.35, 13.94] \\
\midrule
\multirow{2}{*}{Follow-up Study} & Bayesian & 20.65 & [16.75, 24.56] & 36.17 & [30.88, 41.45] \\
 & Best-fit & 10.54 & [7.29, 13.80] & 17.25 & [13.00, 21.51] \\
\bottomrule
\end{tabular}

\end{adjustbox}
\end{tabular}
\begin{tablenotes}\footnotesize
\item \textit{Notes:} The table reports the percentage of beliefs within the indicated distance of each benchmark in the main and follow-up studies. The samples contain 5,211 beliefs from 193 main-study participants and 2,295 beliefs from 85 follow-up participants, all from the \textit{SP} block. The 95\% confidence intervals are based on standard errors calculated from participant-level shares. The follow-up interface displays the posterior under each individual model.
\end{tablenotes}
\end{threeparttable}
\caption{Classification of Round-level Beliefs Across Studies}
\label{tab:comparison}
\end{table}

\clearpage

\begin{landscape}

\renewcommand{\headrulewidth}{0pt}
 \fancypagestyle{landscape}{ 
        \fancyhf{}
    \fancyfoot[C]{%
            \vspace*{-15.7cm} 
            \makebox[\textwidth][c]{\hspace{19cm}\rotatebox{90}{\thepage}} 
        }
    }                        \pagestyle{landscape}

\renewcommand{\thefigure}{\Alph{section}-\arabic{figure}}
\setcounter{figure}{0}
\renewcommand{\thetable}{\Alph{section}-\arabic{table}}
\setcounter{table}{0}

\section{Experimental Instructions and Interface}\label{app:instructions} 

\newcommand{\pagenote}[1]{\vfill\noindent\makebox[1.5\textwidth]{%
    \begin{minipage}{1.2\textwidth}\footnotesize\setstretch{1}%
    \textit{Notes:} #1\end{minipage}}}

\begin{center}
    \includegraphics[width=1.35\textwidth]{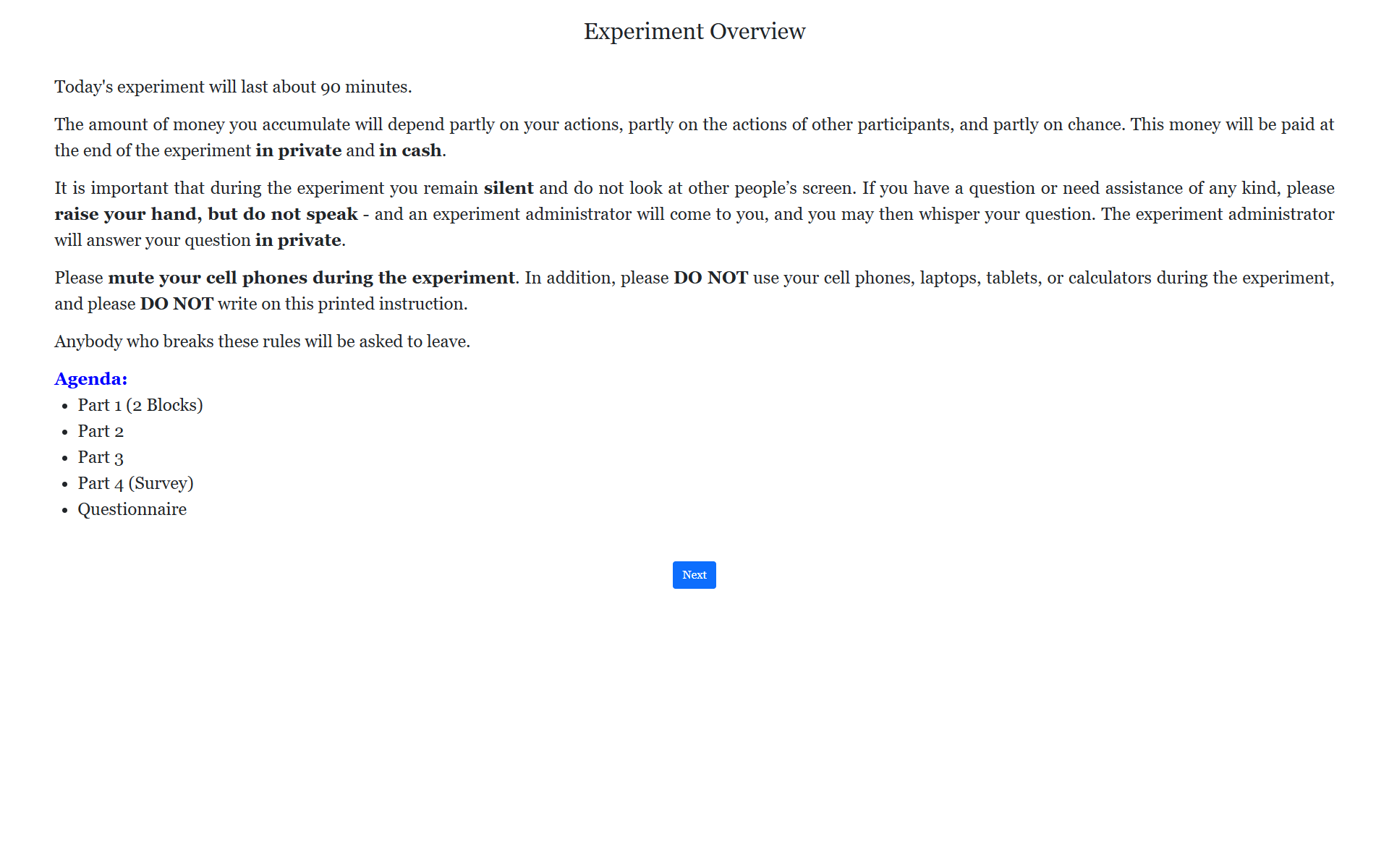}
\vspace{-4em}
\captionof{figure}{Experiment Overview}
\label{fig:instruction_page}
\end{center}

\newpage

\begin{center}
    \includegraphics[width=1.35\textwidth]{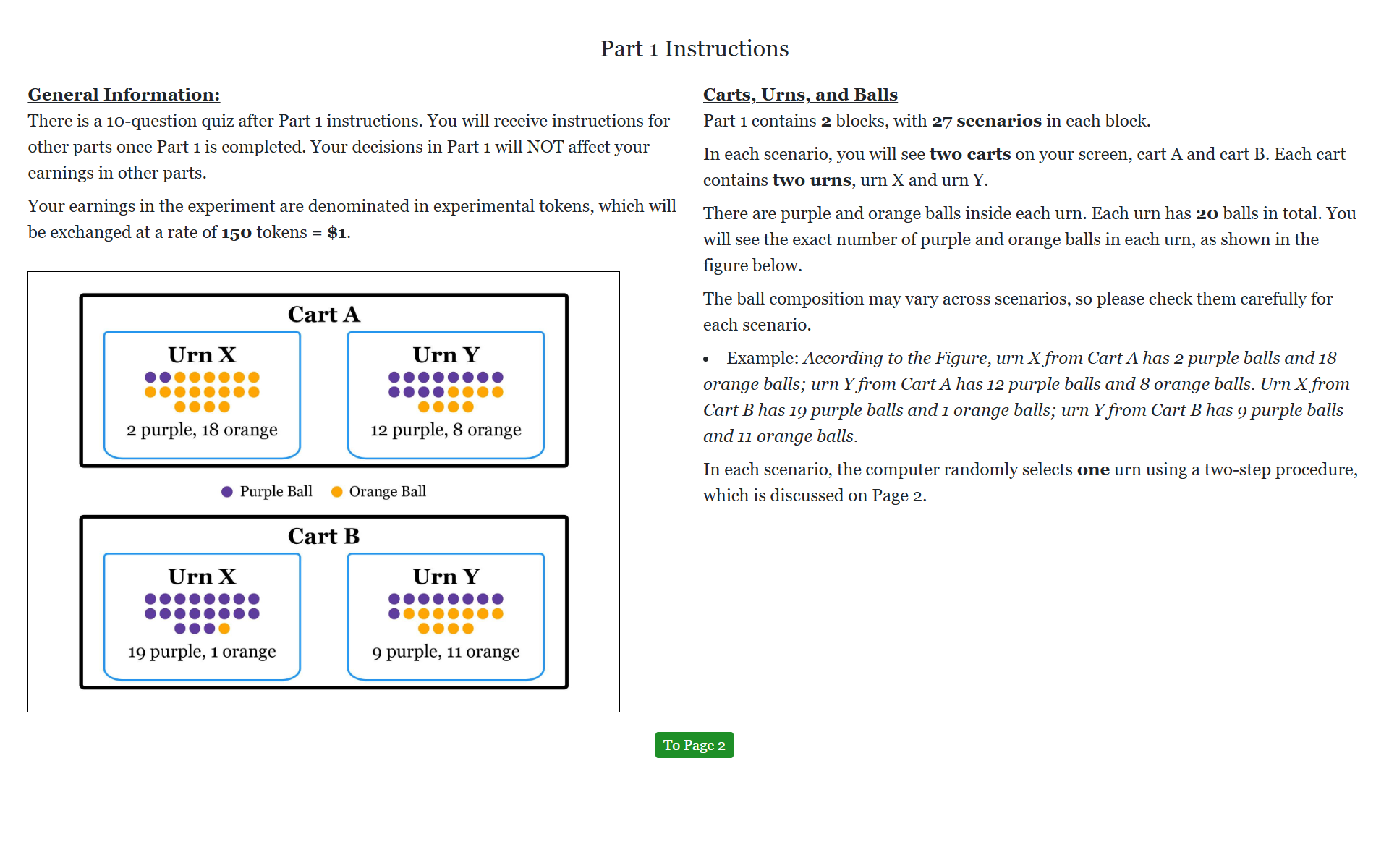}
\vspace{-2em}
\captionof{figure}{Part 1 Instruction P1} 
\label{fig:instruction_page_2}
\end{center}

\newpage

\begin{center}
    \includegraphics[width=1.35\textwidth]{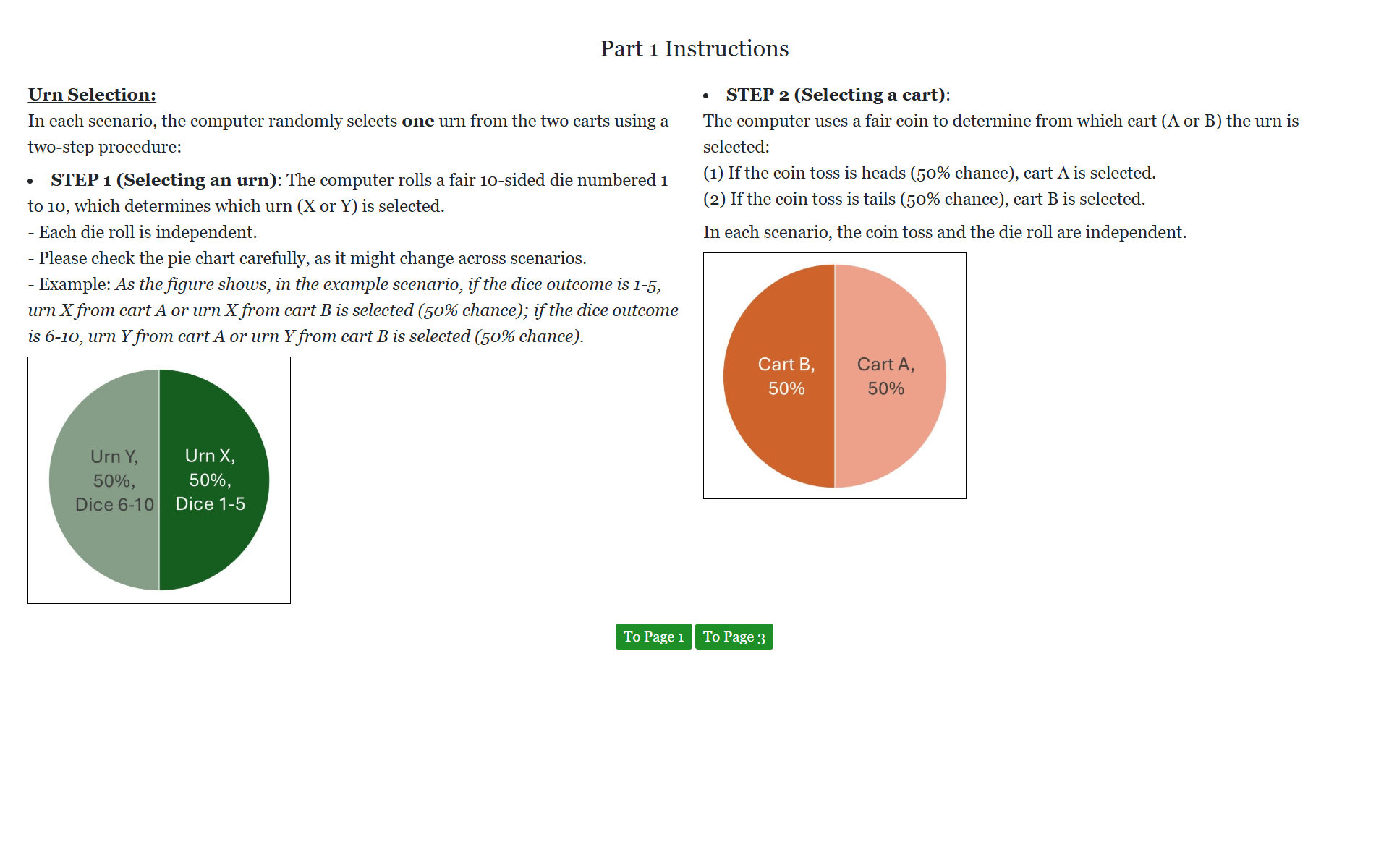}
\vspace{-2em}
\captionof{figure}{Part 1 Instruction P2} 
\label{fig:instruction_page_3}
\end{center}

\newpage

\begin{center}
    \includegraphics[width=1.35\textwidth]{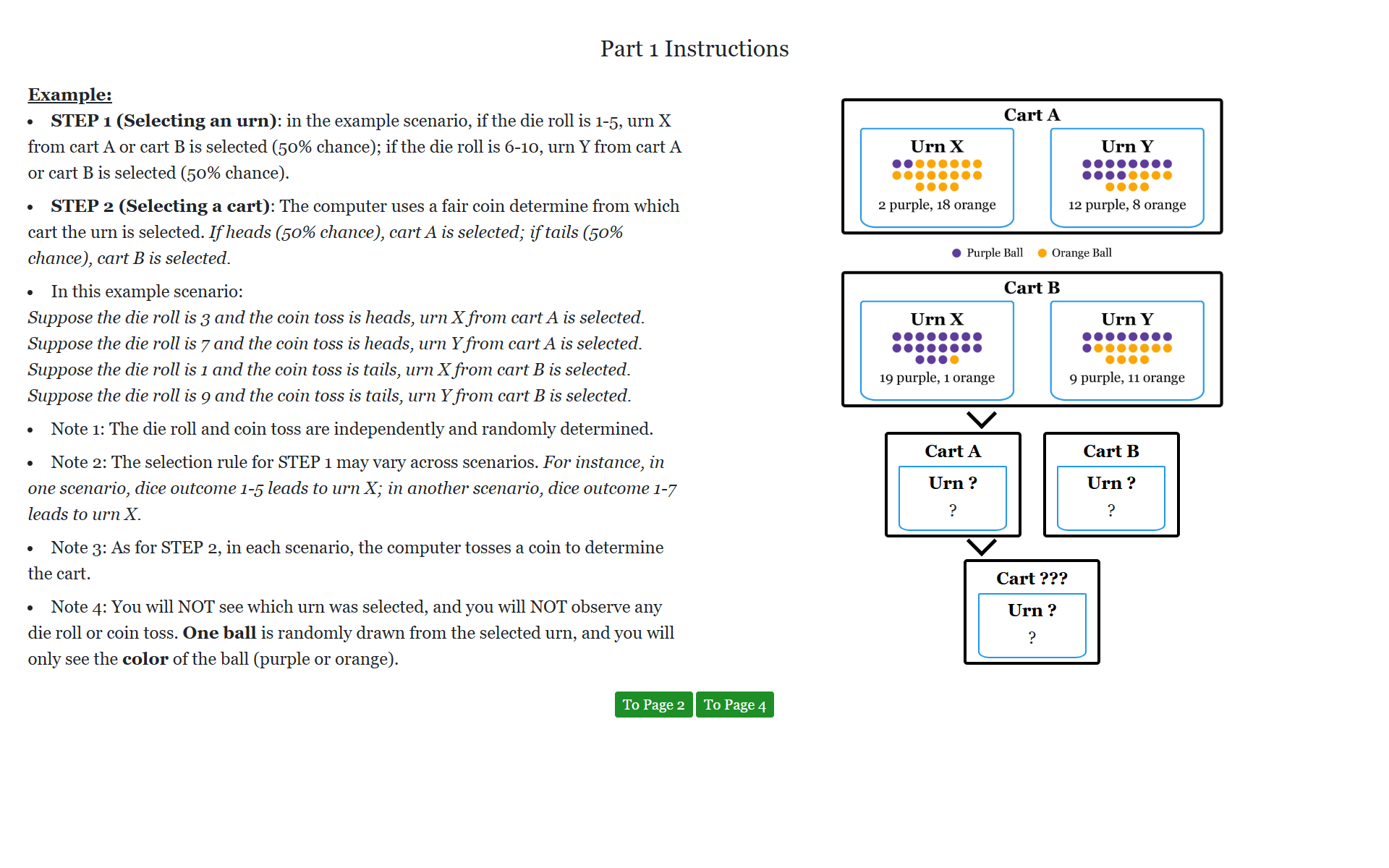}
\vspace{-2em}
\captionof{figure}{Part 1 Instruction P3} 
\label{fig:instruction_page_4}
\end{center}

\newpage

\begin{center}
    \includegraphics[width=1.35\textwidth]{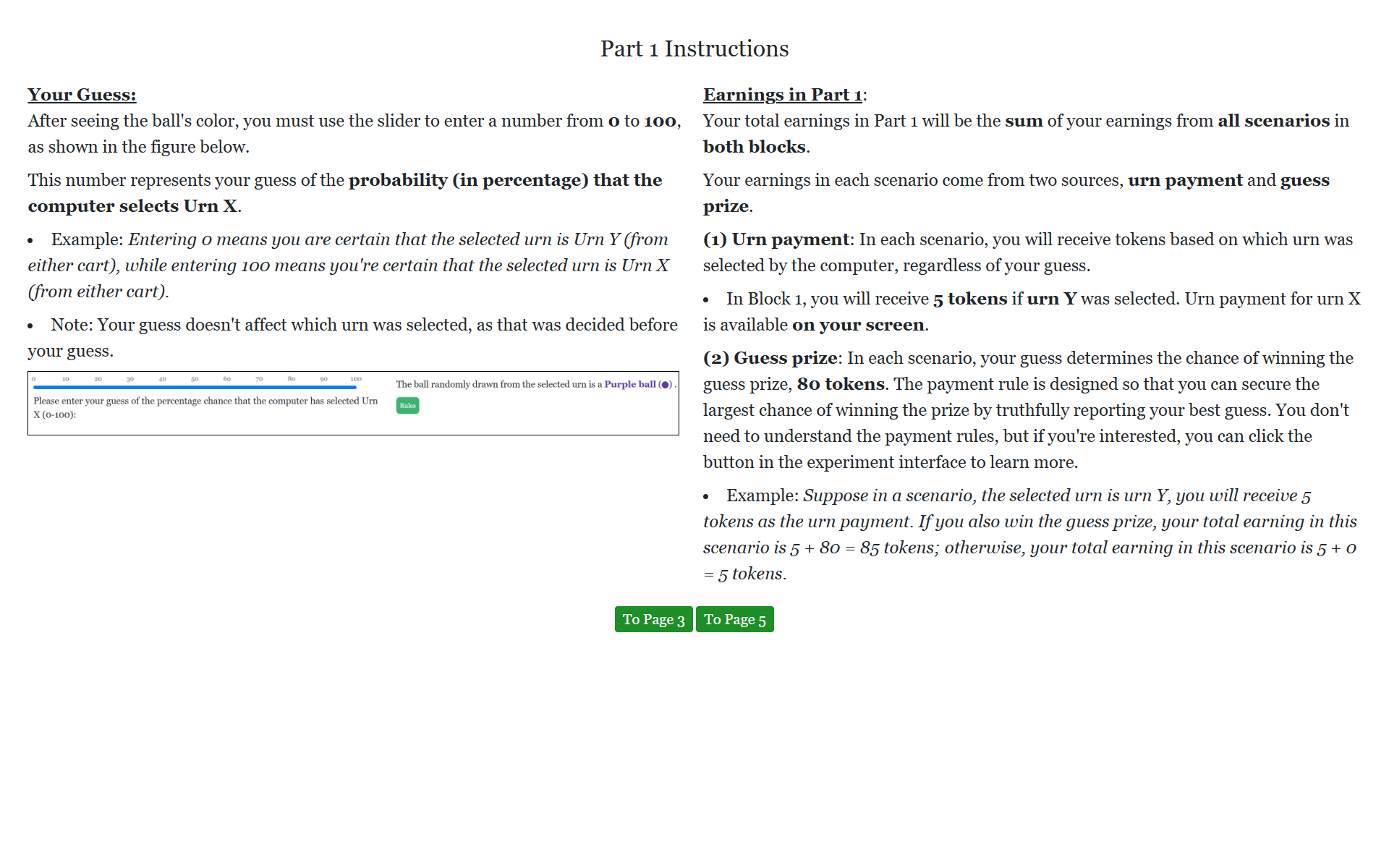}
\vspace{-2em}
\captionof{figure}{Part 1 Instruction P4} 
\label{fig:instruction_page_5}
\end{center}

\newpage

\begin{center}
    \includegraphics[width=1.35\textwidth]{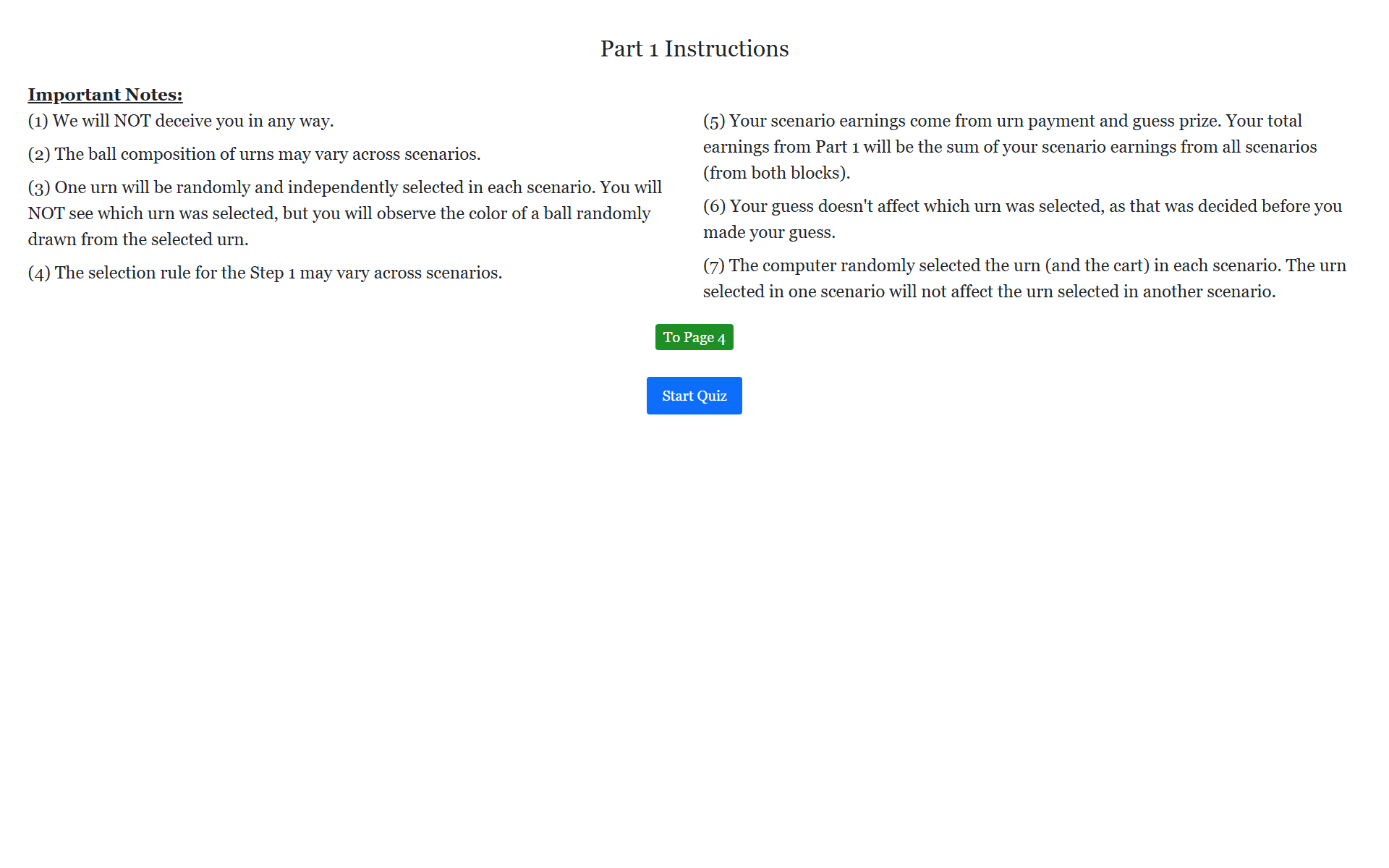}
\vspace{-8em}
\captionof{figure}{Part 1 Instruction P5} 
\label{fig:instruction_page_6}
\end{center}

\newpage

\begin{center}
    \includegraphics[width=1.35\textwidth]{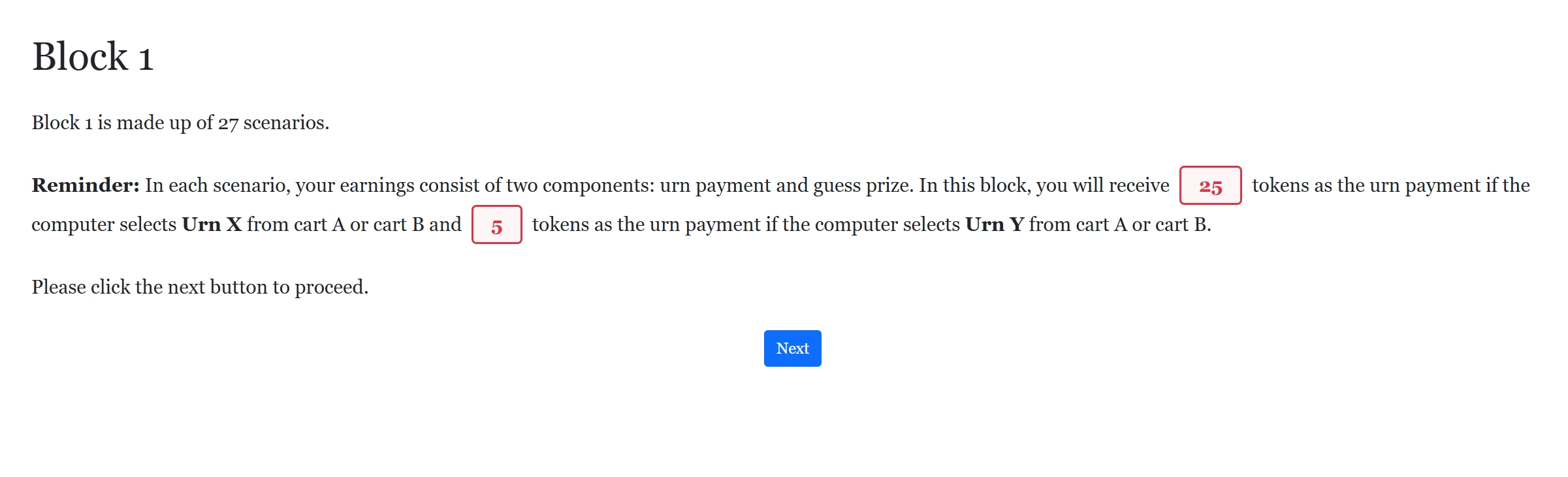}
\vspace{-2em}
\captionof{figure}{Part 1 Block 1 Overview} 
\label{fig:instruction_page_7}
\end{center}
\pagenote{Overview screen at the start of Block 1. The page initially has no Next button. Once all participants have arrived, the experimenter advances the session. The Next button then appears, and participants proceed to the first task.}

\newpage

\begin{center}
    \includegraphics[width=1.35\textwidth]{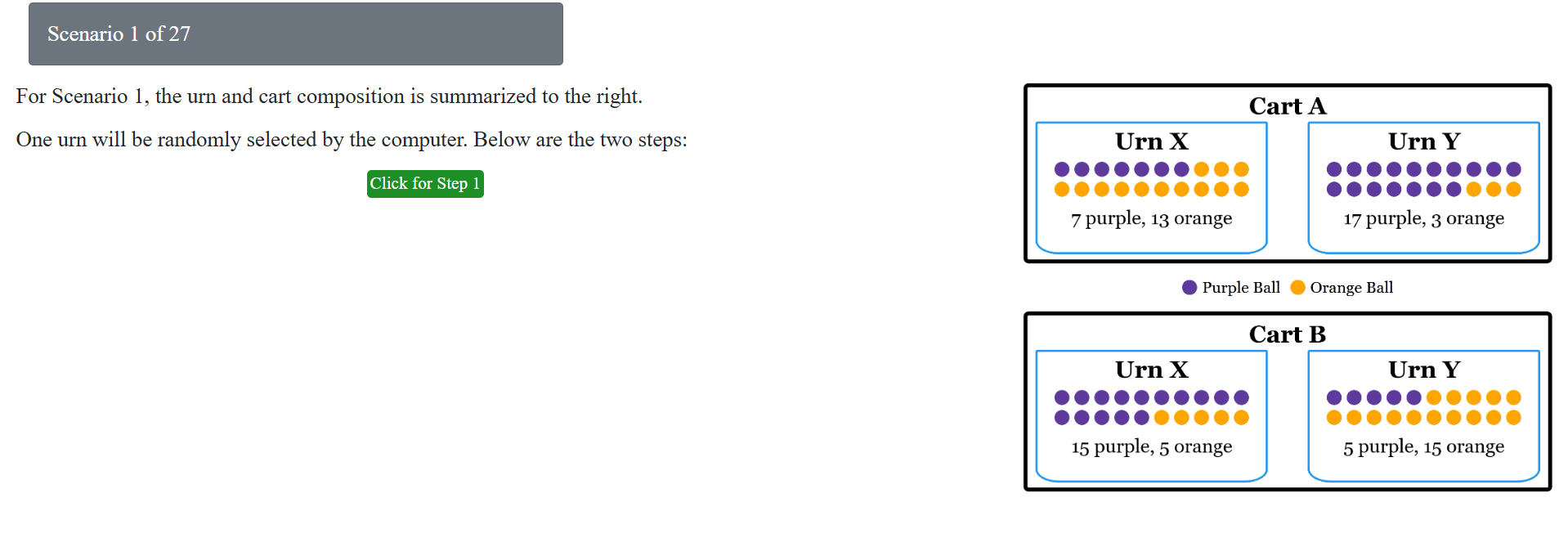}
\vspace{0em}
\captionof{figure}{Part 1 Example Task} 
\label{fig:instruction_page_8}
\end{center}

\newpage

\begin{center}
    \includegraphics[width=1.35\textwidth]{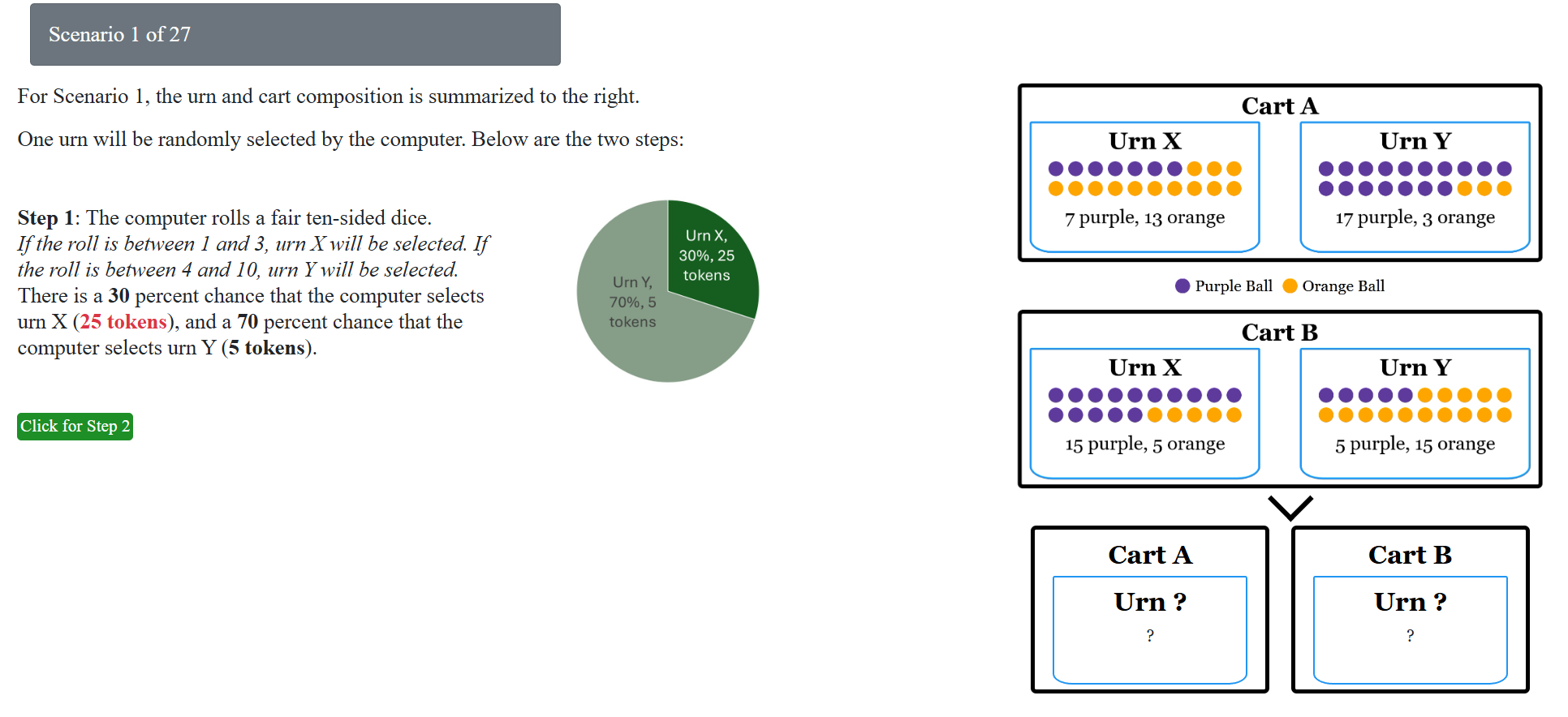}
\vspace{0em}
\captionof{figure}{Part 1 Example Task} 
\label{fig:instruction_page_9}
\end{center}

\newpage

\begin{center}
    \includegraphics[width=1.35\textwidth]{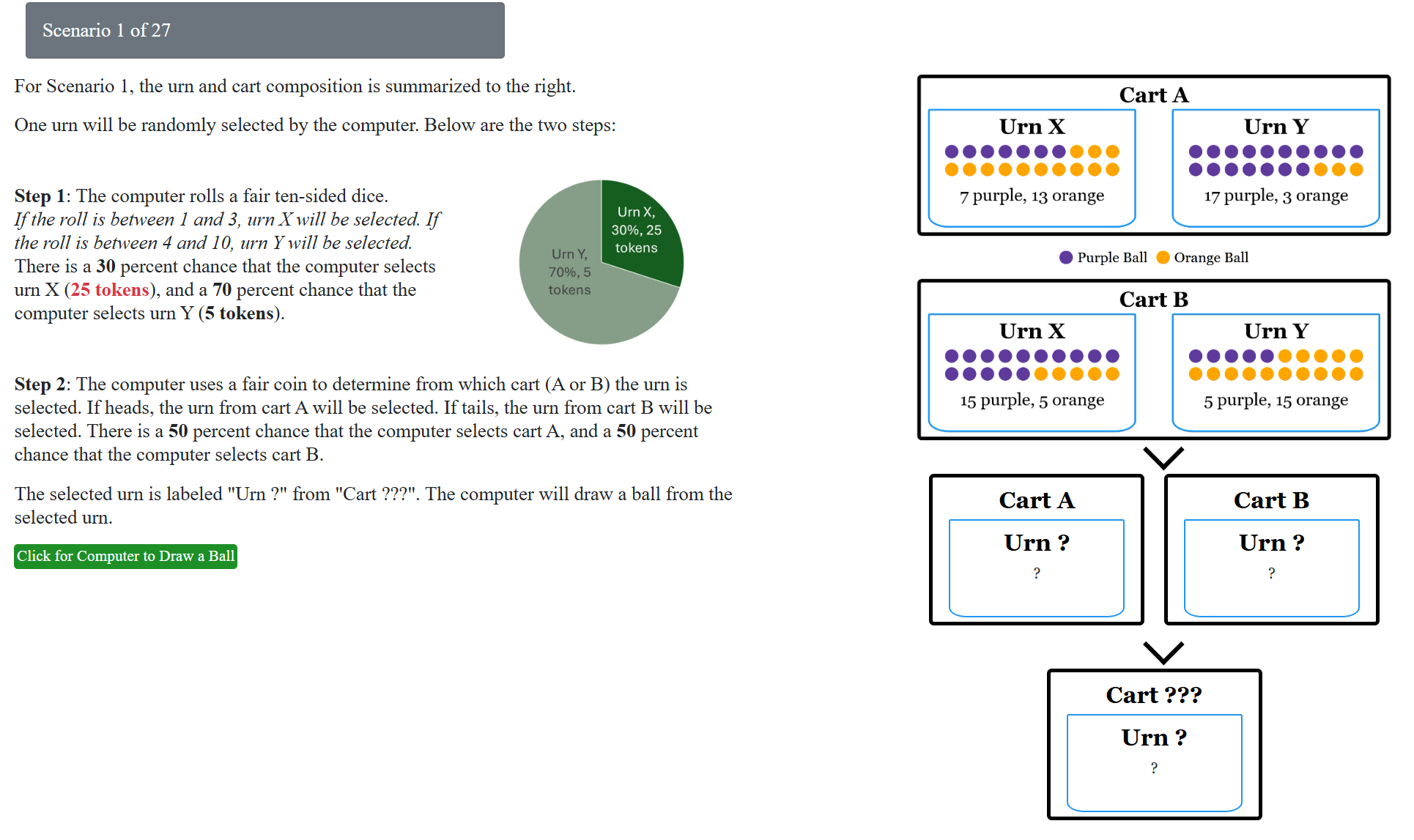}
\vspace{0em}
\captionof{figure}{Part 1 Example Task} 
\label{fig:instruction_page_10}
\end{center}

\newpage

\begin{center}
    \includegraphics[width=1.35\textwidth]{Files/UpdatingTask_04.png}
\vspace{0em}
\captionof{figure}{Part 1 Example Task} 
\label{fig:instruction_page_11}
\end{center}

\newpage

\begin{center}
    \includegraphics[width=1.35\textwidth]{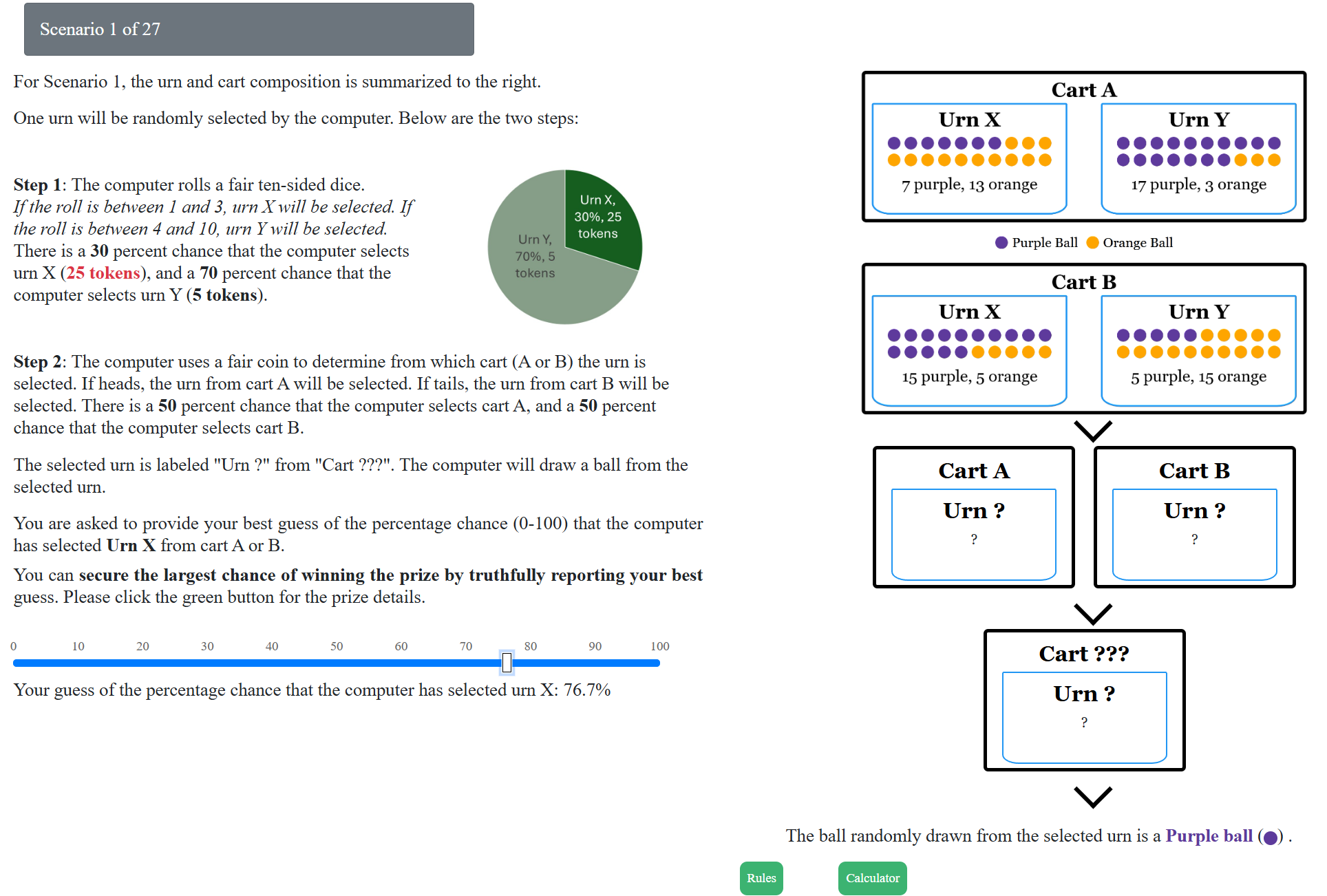}
\vspace{0em}
\captionof{figure}{Part 1 Example Task} 
\label{fig:instruction_page_12}
\end{center}

\newpage

\begin{center}
    \includegraphics[width=1.35\textwidth]{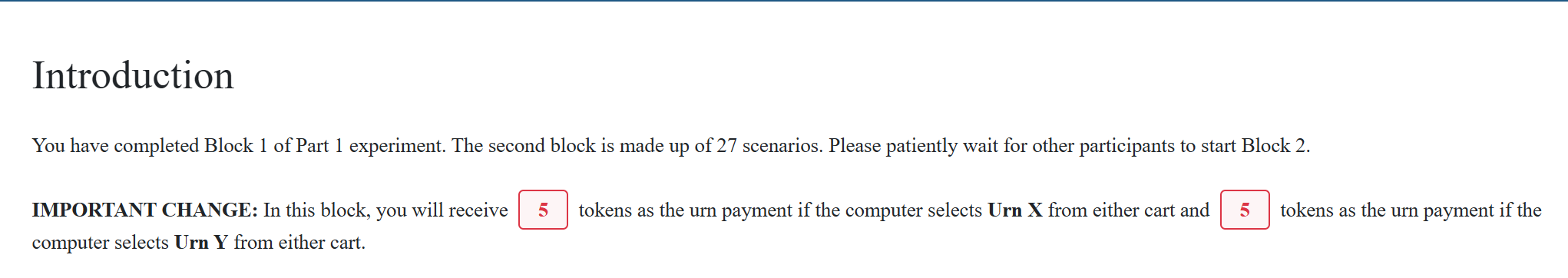}
\vspace{0em}
\captionof{figure}{Part 1 Block 2 Overview} 
\label{fig:instruction_page_13}
\end{center}
\pagenote{Overview screen at the start of Block 2. The page initially has no Next button. Once all participants have arrived, the experimenter advances the session. The Next button then appears, and participants proceed to the first task.}

\newpage

\begin{center}
    \includegraphics[width=1.35\textwidth]{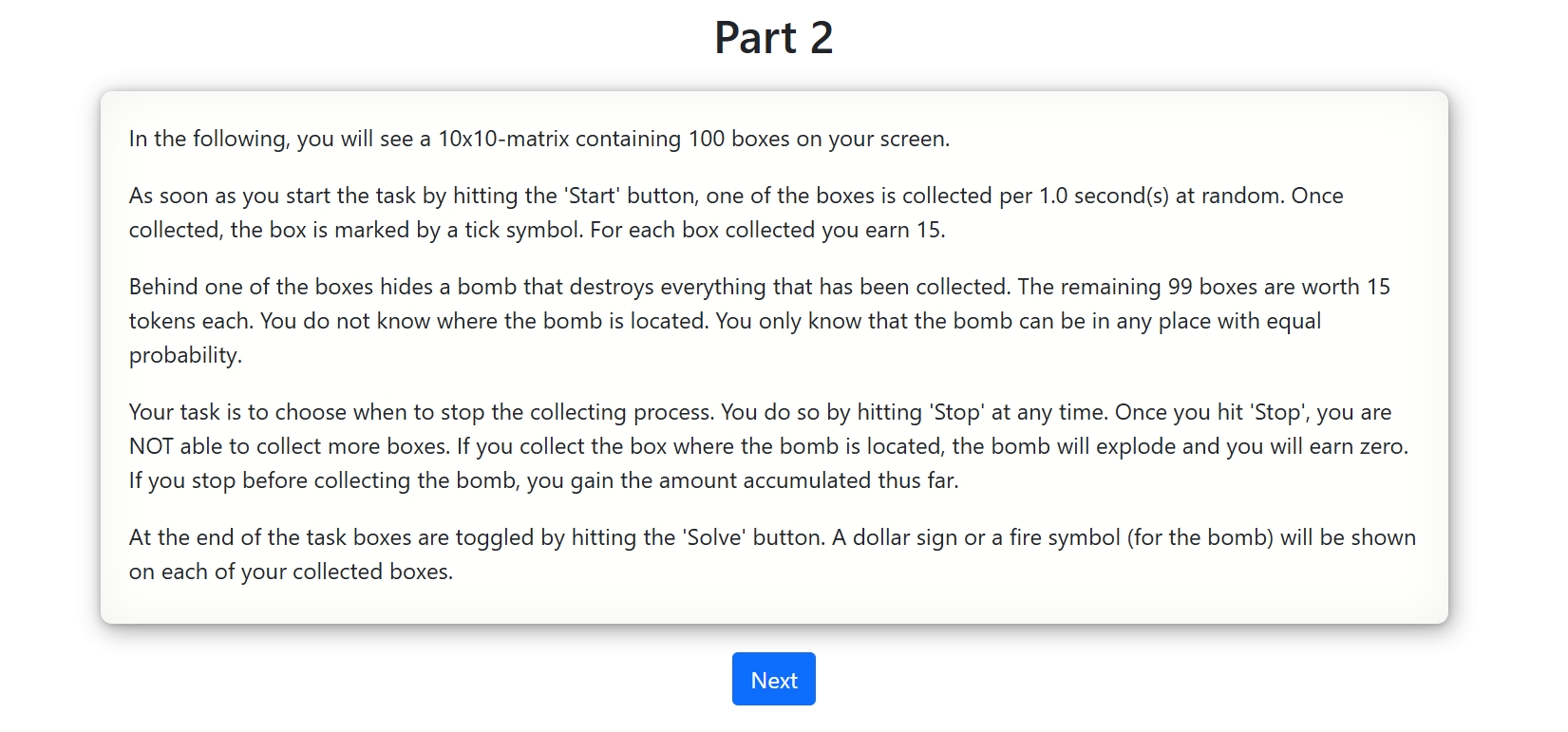}
\vspace{0em}
\captionof{figure}{BRET Task Instructions} 
\label{fig:instruction_page_14}
\end{center}

\newpage

\begin{center}
    \includegraphics[width=1.05\textwidth]{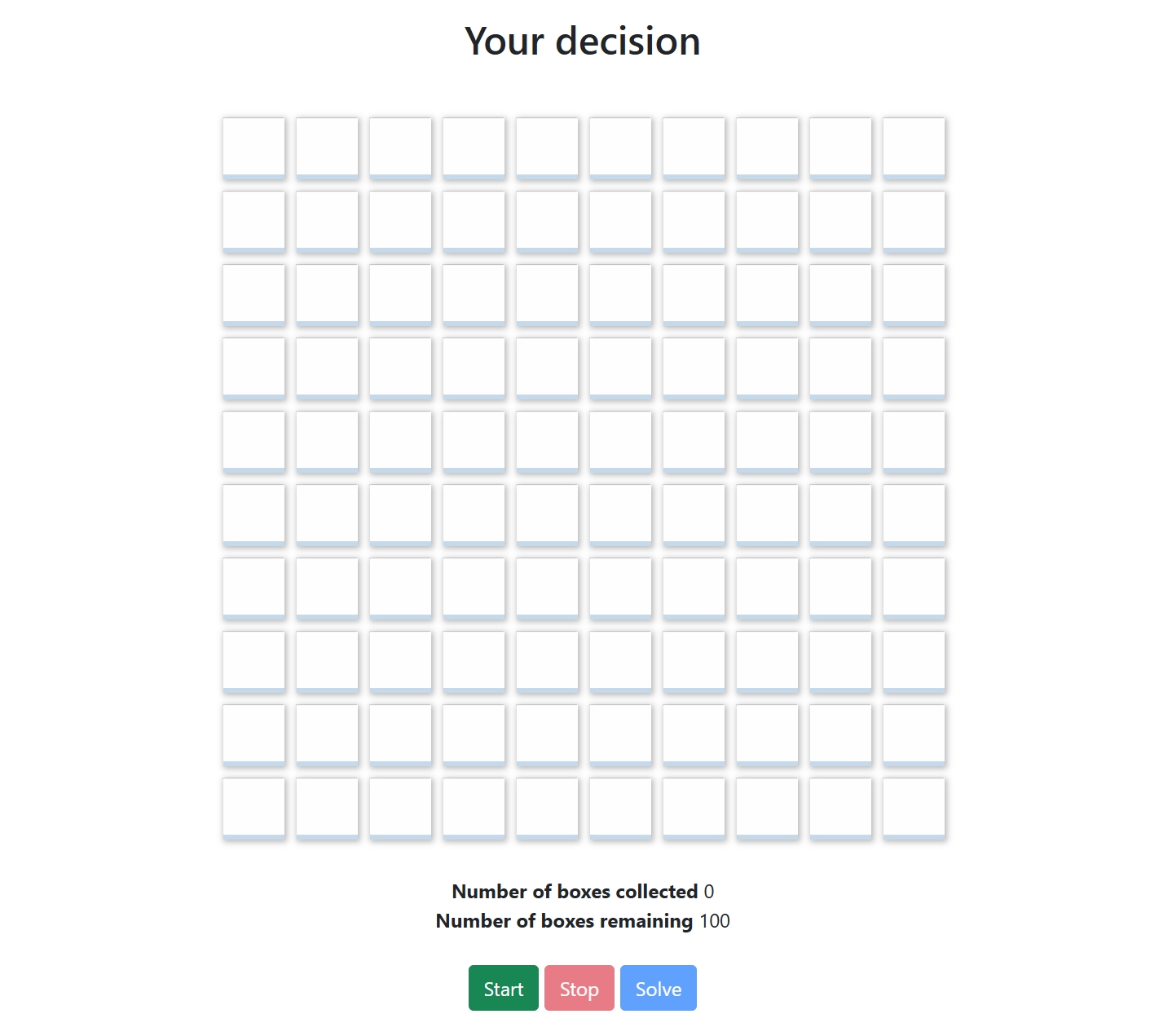}
\vspace{0em}
\captionof{figure}{BRET Task Interface}
\label{fig:instruction_page_15}
\end{center}

\newpage

\begin{center}
    \includegraphics[width=1.35\textwidth]{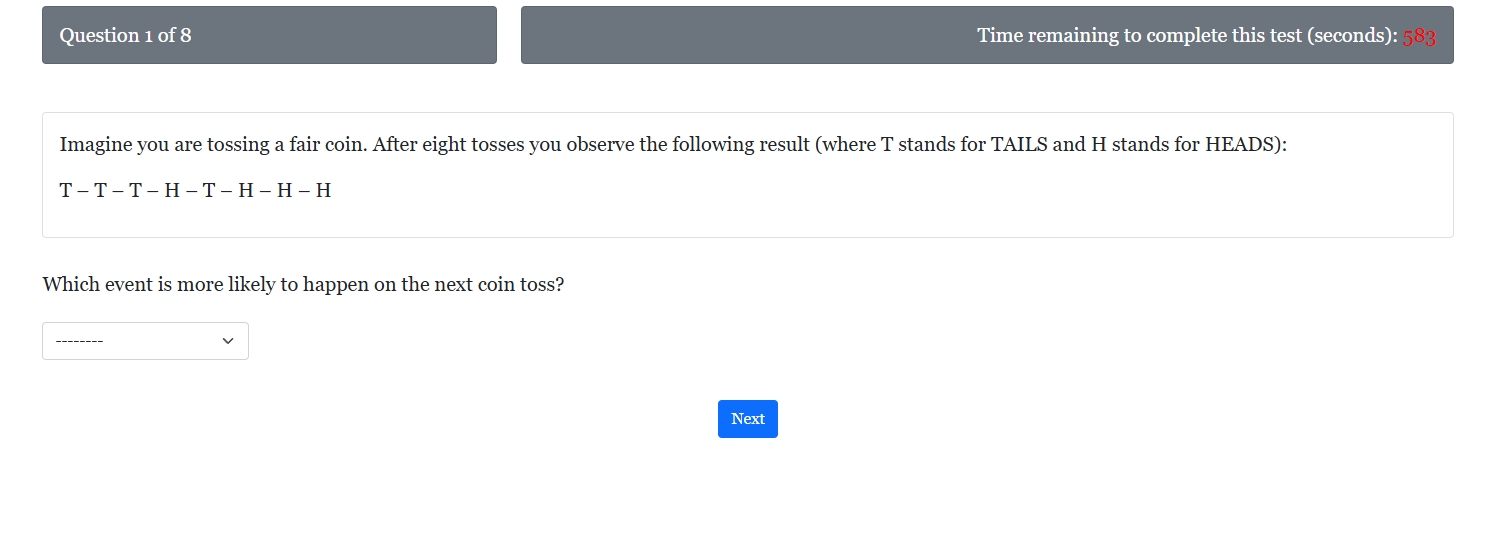}
\vspace{0em}
\captionof{figure}{Cognitive Test 1} 
\label{fig:instruction_page_16}
\end{center}

\newpage

\begin{center}
    \includegraphics[width=1.35\textwidth]{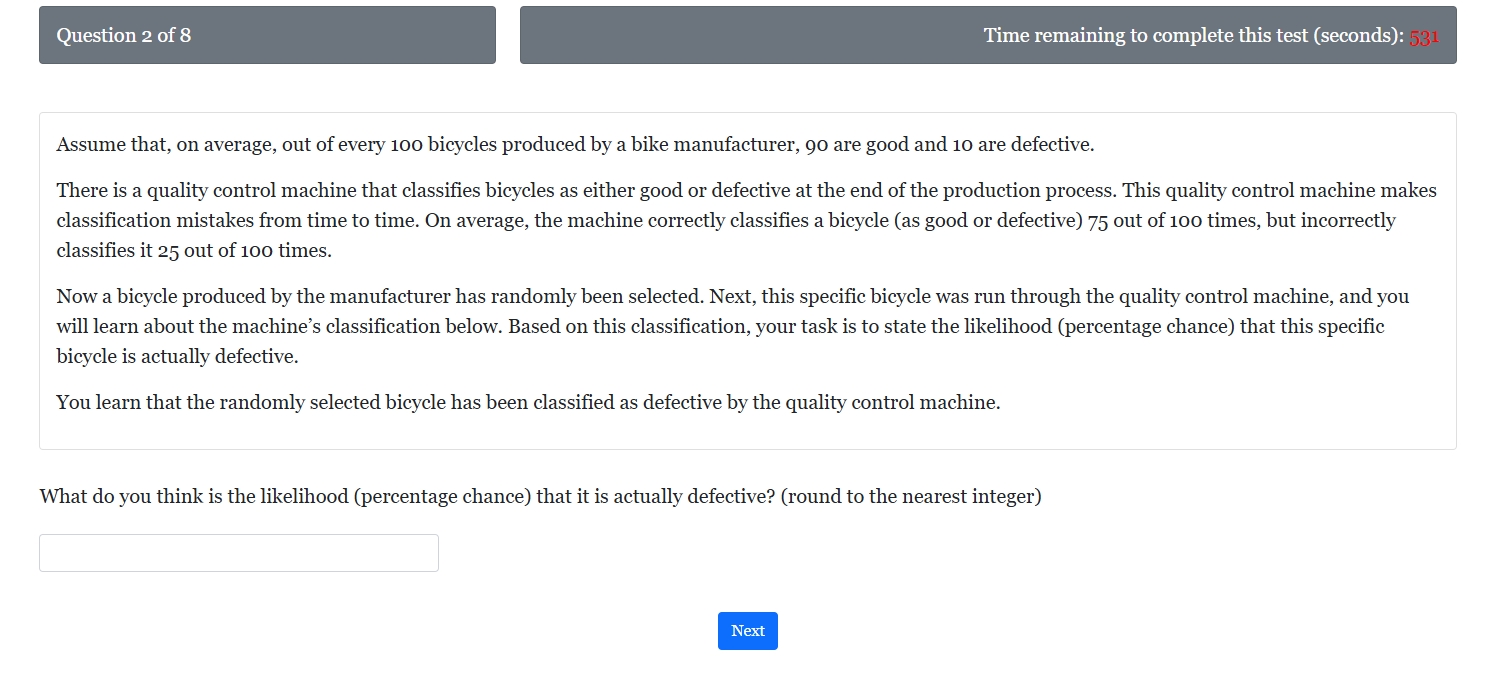}
\vspace{0em}
\captionof{figure}{Cognitive Test 2} 
\label{fig:instruction_page_17}
\end{center}

\newpage

\begin{center}
    \includegraphics[width=1.35\textwidth]{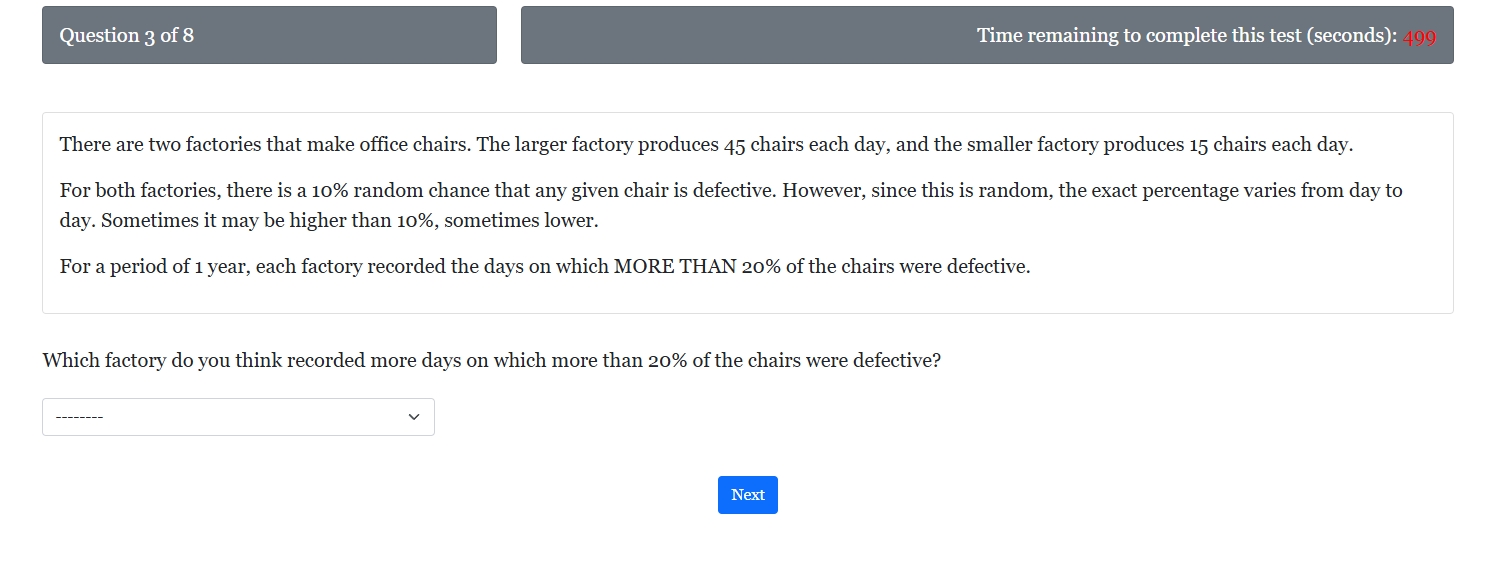}
\vspace{0em}
\captionof{figure}{Cognitive Test 3} 
\label{fig:instruction_page_18}
\end{center}

\newpage

\begin{center}
    \includegraphics[width=1.1\textwidth]{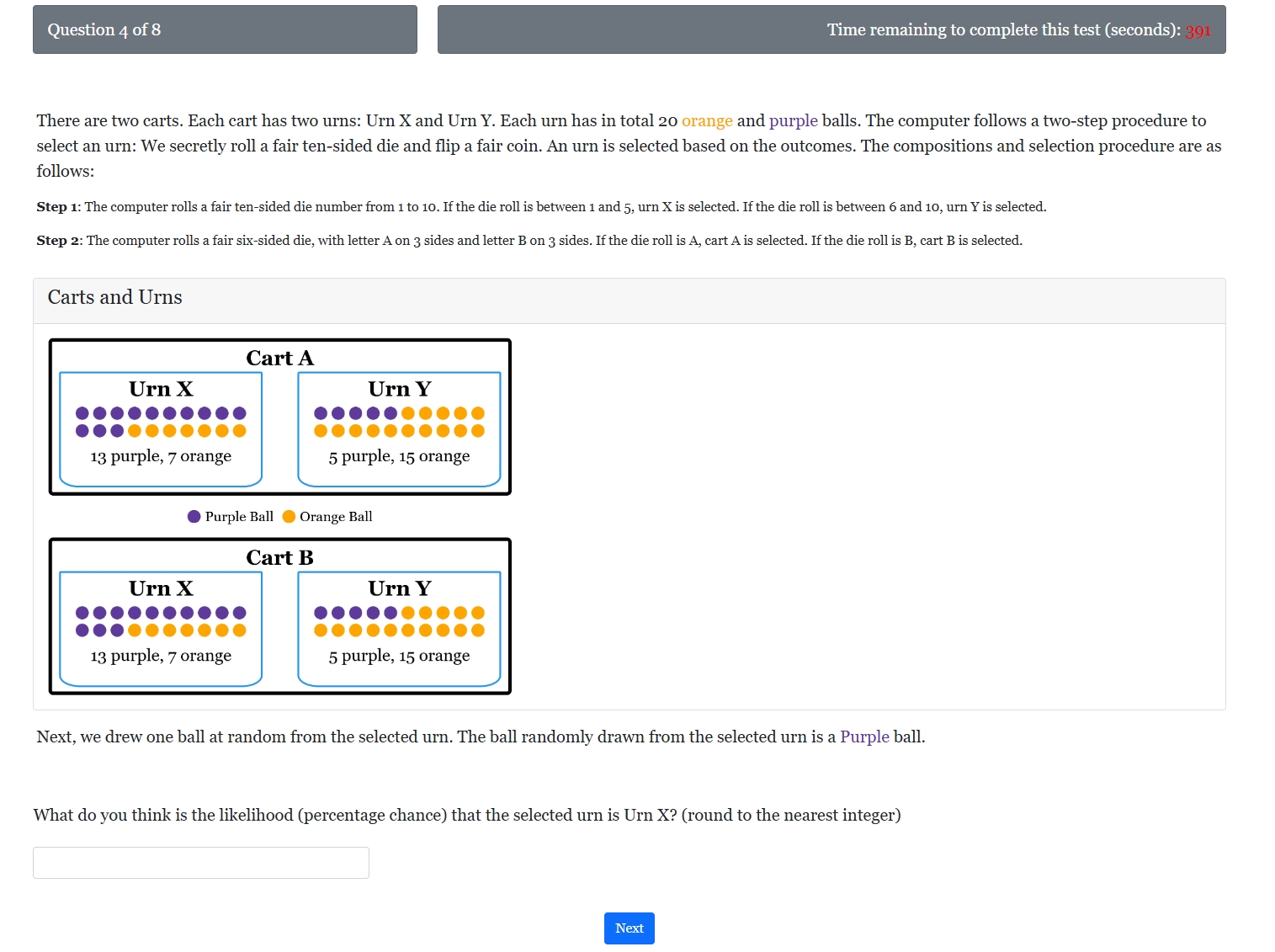}
\vspace{0em}
\captionof{figure}{Cognitive Test 4} 
\label{fig:instruction_page_19}
\end{center}

\newpage

\begin{center}
    \includegraphics[width=1.35\textwidth]{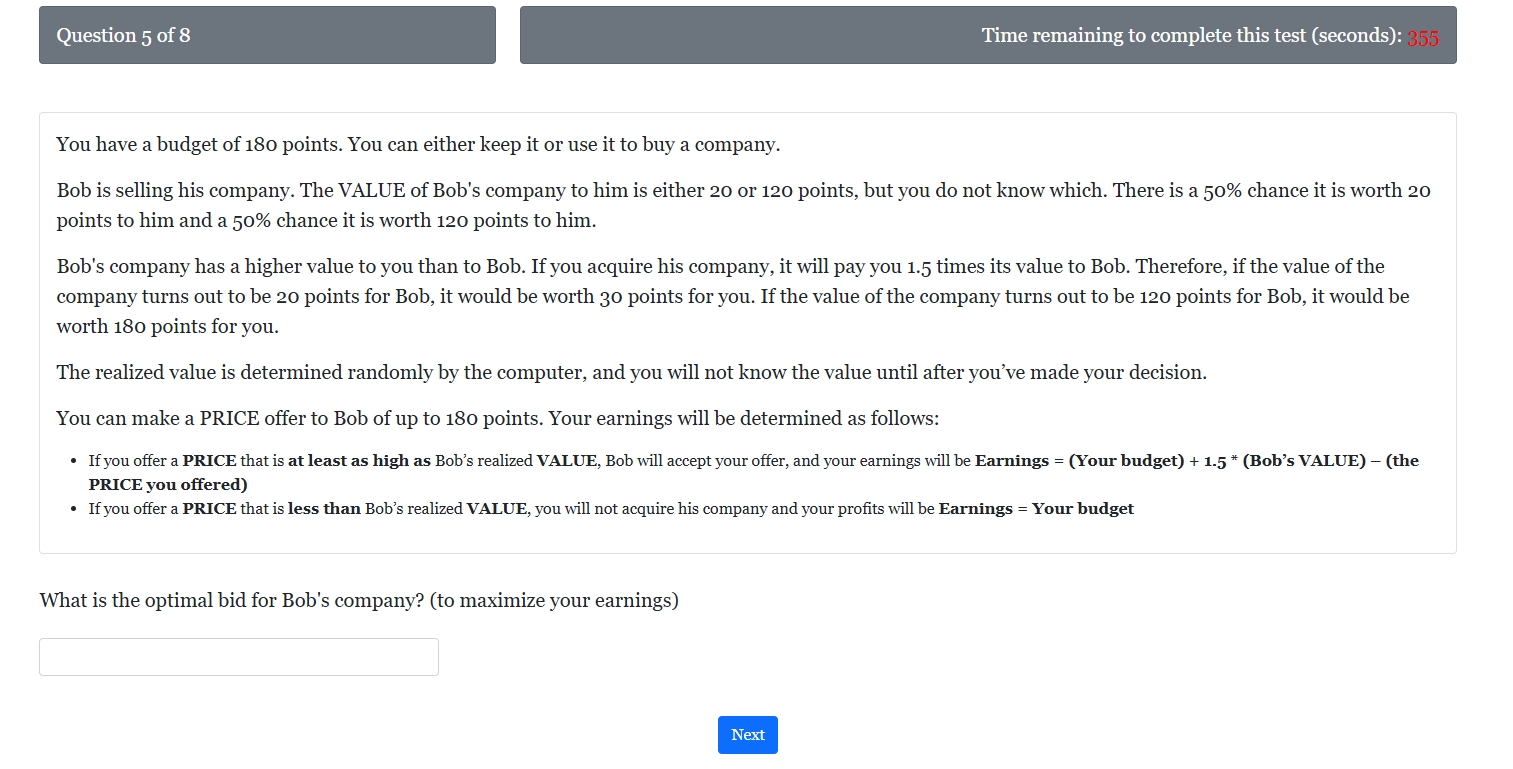}
\vspace{0em}
\captionof{figure}{Cognitive Test 5} 
\label{fig:instruction_page_20}
\end{center}

\newpage

\begin{center}
    \includegraphics[width=1.35\textwidth]{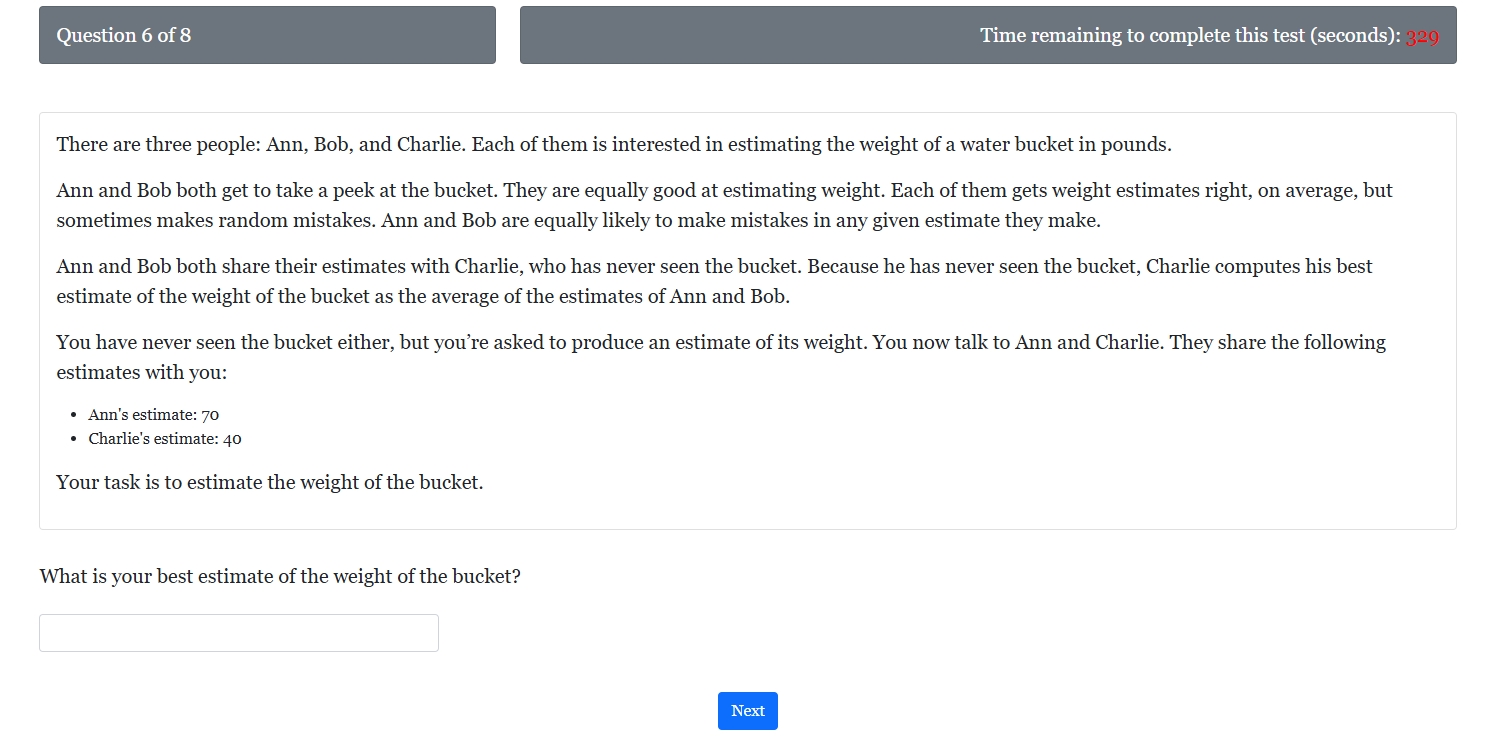}
\vspace{0em}
\captionof{figure}{Cognitive Test 6} 
\label{fig:instruction_page_21}
\end{center}

\newpage

\begin{center}
    \includegraphics[width=1.35\textwidth]{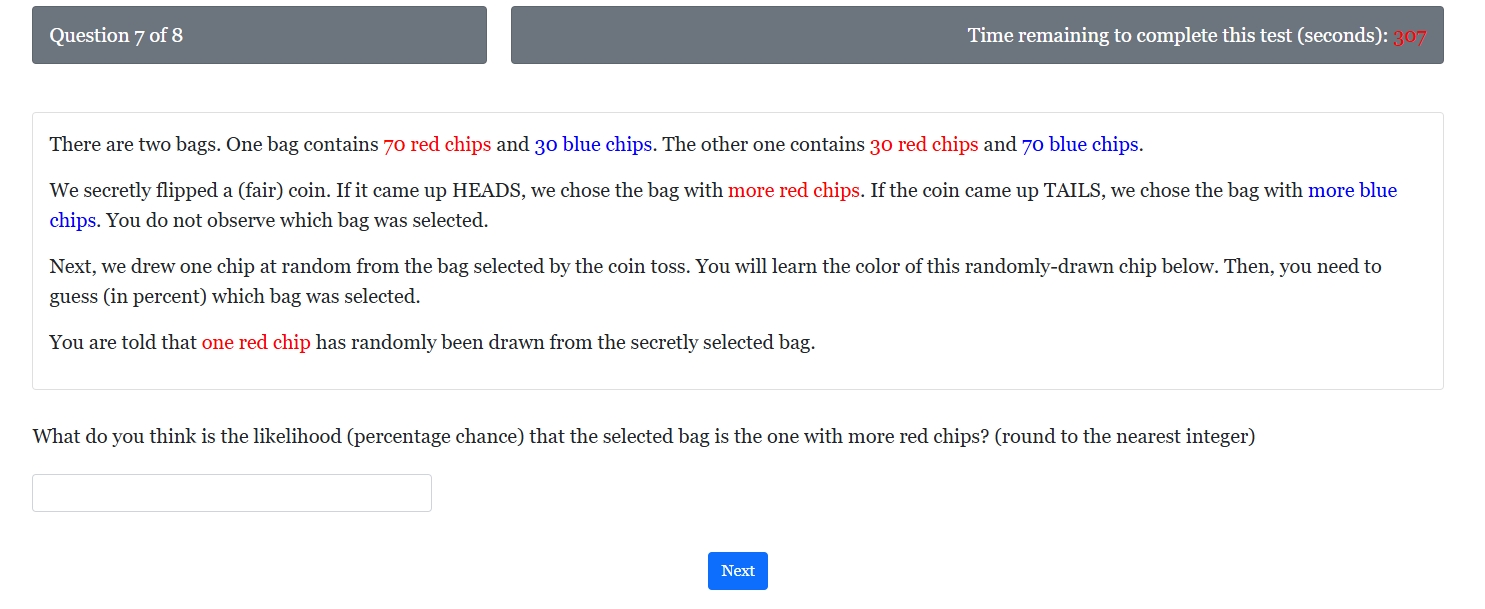}
\vspace{0em}
\captionof{figure}{Cognitive Test 7} 
\label{fig:instruction_page_22}
\end{center}

\newpage

\begin{center}
    \includegraphics[width=1.35\textwidth]{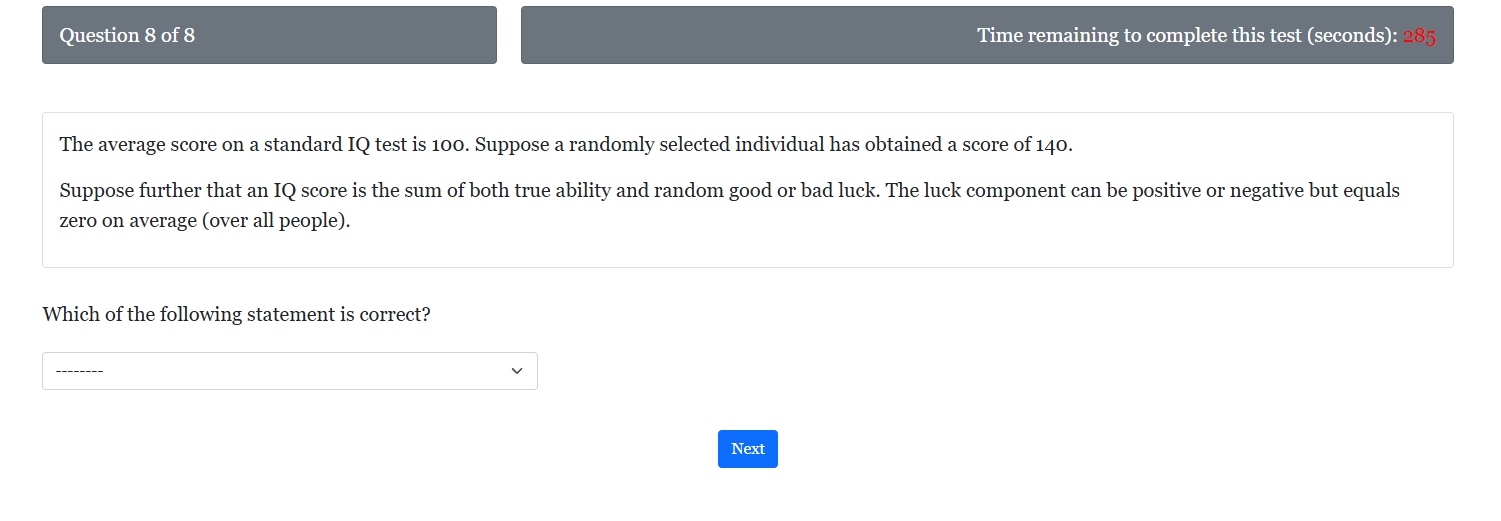}
\vspace{0em}
\captionof{figure}{Cognitive Test 8} 
\label{fig:instruction_page_23}
\end{center}

\end{landscape}

\end{document}